\documentclass[11pt,a4paper]{article}

\usepackage[T1]{fontenc}
\usepackage[utf8]{inputenc}
\usepackage{lmodern}
\usepackage[margin=1in]{geometry}
\usepackage{microtype}
\usepackage{amsmath,amssymb,amsthm}
\usepackage{placeins}
\usepackage{comment}
\usepackage{graphicx}
\usepackage{xcolor}
\usepackage{subcaption}
\usepackage{booktabs,longtable}
\usepackage{url}
\usepackage[authoryear]{natbib}
\usepackage[unicode,hidelinks]{hyperref}
\usepackage{cleveref}

\hypersetup{
  pdftitle={Stability Beyond Minty: Numerical Lyapunov Certification for Equilibrium Learning in Bayesian Bertrand Competition},
  pdfauthor={Jan H\"ohener and Martin Bichler},
  pdfsubject={Equilibrium learning in Bayesian Bertrand competition},
  pdfkeywords={Game theory, Equilibrium computation, Numerical certification, Sum-of-squares optimization, Bayesian games}
}
\setkeys{Gin}{draft=false}
\bibpunct[, ]{(}{)}{,}{a}{}{,}

\def\newblock{\ }

\newtheorem{theorem}{Theorem}
\newtheorem{lemma}[theorem]{Lemma}
\theoremstyle{definition}
\newtheorem{definition}{Definition}[section]
\newtheorem{assumption}{Assumption}[section]
\newtheorem{corollary}{Corollary}[section]
\newtheorem{proposition}{Proposition}[section]
\newtheorem{problem}{Problem}[section]
\theoremstyle{remark}
\newtheorem{remark}{Remark}[section]

\title{\LARGE\bfseries Stability Beyond Minty:\\[0.3em]
Numerical Lyapunov Certification for Equilibrium Learning in Bayesian Bertrand Competition\thanks{This project has received funding from the European Research Council (ERC) under the European Union's Horizon Europe research and innovation programme (grant agreement No 101198689).}}
\author{Jan H{\"o}hener \qquad Martin Bichler\thanks{Corresponding author: \href{mailto:m.bichler@tum.de}{\nolinkurl{m.bichler@tum.de}}. ORCID: \href{https://orcid.org/0000-0001-5491-2935}{0000-0001-5491-2935}.}\\[0.5em]
\small Department of Computer Science, Technical University of Munich\\
\small Germany}
\date{}

\begin{document}
\maketitle
\begin{abstract}
Computing Bayesian Nash equilibria with continuous types and actions
requires solving for strategy functions. Learning methods provide a
flexible approach, but standard variational conditions often fail to
explain their convergence. We study symmetric Bayesian Bertrand
competition with private marginal costs and piecewise-linear pricing
strategies.

We prove that Minty stability fails arbitrarily close to equilibrium
in the continuous-strategy game. It also fails in sufficiently fine
piecewise-linear approximations when sufficiently flat pricing
segments are permitted.

We develop a sum-of-squares semidefinite formulation for global
Lyapunov functions of the projected Riemannian dynamics. Feasibility
of the stated positive-rate conditions implies global asymptotic
stability. A direct proximal Lyapunov argument yields almost-sure
last-iterate convergence under persistent martingale noise and
square-summable step sizes; a separate tracking result covers exact
and vanishing-error gradients with more general diminishing steps.

The computational design covers 243 systems with two to four
strategy segments, varying firms, cost distributions, demand, and
learning geometry. The search returns solver-feasible SOS
representations throughout, including 111 systems with explicitly
verified Minty failure. Of these 243 outputs, 237 have a positive
requested decrease rate; six have a rate rounded to zero and are
assessed separately. Numerical-certificate acceptance additionally
requires reconstruction and Lyapunov checks at stated tolerances.
Euclidean candidates are quadratic throughout, entropic candidates
are almost always quadratic, and Burg candidates sometimes have
higher degree. The results demonstrate the computational reach of
the Lyapunov approach beyond Minty and reveal substantial differences
in search complexity across learning geometries.
\end{abstract}
\begin{center}
\begin{minipage}{0.9\textwidth}
\end{minipage}
\end{center}

\section{Introduction}
\label{sec:introduction}

Computing Bayesian Nash equilibria in games with continuous types and actions is an infinite-dimensional problem: a pure strategy maps a privately observed type into an action. Symmetry and particular distributional assumptions sometimes yield closed-form solutions \citep{vickrey1961}, but more general environments lead to nonlinear differential or integral equations that can be difficult to solve reliably. For example, numerical methods for asymmetric auctions can be sensitive to boundary conditions and initialization \citep{MarshallMeurerRichardStromquist1994}. Bayesian Bertrand competition presents the same challenge. Each firm selects a pricing function of its private marginal cost, taking its competitors' pricing functions and cost distributions into account. Reliable equilibrium computation supports the analysis of competition, policy interventions, and adaptive pricing procedures \citep{LiuSunXu2025}.

Learning-based equilibrium solvers provide a flexible computational approach. Neural Pseudogradient Ascent represents strategies by neural networks and updates their parameters through simultaneous gradient learning \citep{bichler2021npga,bichler2023learning}. Simultaneous Online Dual Averaging discretizes type and action spaces and learns distributional strategies \citep{bichler2023soda}. These methods recover known equilibria in benchmark environments and compute approximate equilibria in more complex games. Yet favorable trajectories and small terminal equilibrium residuals do not establish convergence of the underlying algorithm. General games can exhibit cycling or more complicated recurrent behavior \citep{hart2003uncoupled,milionis2023impossibility,sanders2018prevalences}, and comprehensive ex-post verification can itself be costly \citep{pieroth2025}.

The most widely used sufficient conditions for last-iterate convergence rely on monotonicity or variational stability. Strong monotonicity supports uniqueness and convergence results \citep{facchinei2003finite}. The weaker Minty condition supplies a directional relation between the game operator and an equilibrium that underlies convergence arguments for mirror descent, dual averaging, and related methods \citep{mertikopoulos2019learning}. When this structure is absent, convergence analysis becomes specific to the game, its representation, and the learning geometry. For Bayesian games with continuous types and actions, directly verifiable global conditions in this regime remain scarce \citep{HoangSoumisZaccour2019,BichlerKohringOberlechnerPieroth2023,bichler2025beyond}.

We study this problem in homogeneous-good Bayesian Bertrand competition with independent private costs and polynomial demand. The model has an explicitly known affine symmetric equilibrium, which belongs to every piecewise-linear approximation used below. This makes it possible to vary the approximation dimension without changing the equilibrium target or introducing equilibrium-approximation error. Nevertheless, the continuous-strategy game violates Minty in every neighborhood of equilibrium. We also construct a finite resolution threshold beyond which explicit projected strategies violate Minty when their permitted minimum slope is sufficiently small. The threshold belongs to the witness construction; it does not establish that Minty holds at coarser resolutions.

We analyze proximal regularized Robbins--Monro learning as a symmetric equilibrium-computation procedure. A shared parameter vector represents the common pricing function. The economic primitives determine the unilateral-deviation gradient, the piecewise-linear representation determines a polyhedral feasible set, and the regularizer determines a state-dependent metric. The associated continuous-time system is a projected Riemannian dynamic. Thus the stability question concerns a specified game, approximation, and learning rule. It does not describe independently adapting firms with separate parameters and noise processes.

Our constructive approach searches for polynomial Lyapunov functions satisfying positivity, strict decrease, and compatibility with the metric-dependent projection on every active boundary face. The rational game gradient becomes polynomial after multiplication by a strictly positive denominator. Sum-of-squares (SOS) representations then turn the global sufficient inequalities into a semidefinite feasibility problem. The positive-rate SOS conditions provide sufficient conditions
for global asymptotic stability. The same boundary condition also controls the normal correction in a discrete proximal step, yielding a direct stochastic convergence proof. For exact and almost-surely vanishing gradient errors, we additionally prove that the affine interpolation tracks the projected dynamics, without requiring square-summable steps.

The computational study varies 27 economic configurations, three approximation dimensions, and three learning geometries, producing 243 projected systems. All SOS searches return feasible numerical representations. At the experimental lower-slope bound, exact rational witnesses establish Minty failure in 37 of the 81 projected games, corresponding to 111 systems. The principal computational limit is the growth of the multiplier Gram matrices: the full design covers two to four segments, with three additional successful five-segment probes. The returned Lyapunov candidates are quadratic in every Euclidean system and 77 of 81 entropic systems; Burg candidates have degree two, four, or six. 
The search outcomes document solver-reported feasibility and differences in search complexity.

\subsection{Contributions}
The paper makes four main contributions.
\begin{enumerate}
\item \textbf{Variational failure in a continuous-type Bayesian game.} We prove local Minty failure in function space and identify a sufficient resolution threshold $m^\star$ for a family of projected Minty witnesses, with
\[
 m^\star\le\max\{2,\ell(n-1)+s\}.
\]
The result explains how a coarse approximation can fail to reveal the variational obstruction present in the underlying game.

\item \textbf{A complete Lyapunov-to-algorithm implication for the stated gradient oracles.} A boundary-compatible Lyapunov function gives a one-step bound for the primal proximal recursion. This proves almost-sure convergence under persistent martingale noise with square-summable steps, without assuming stochastic tracking at the boundary. A separate compactness argument establishes tracking for exact and vanishing-error gradients under general diminishing steps with an infinite sum.

\item \textbf{Constructive global numerical stability analysis.}
We combine the rational Bertrand gradient, a structured polynomial
basis, and metric-dependent boundary conditions in an SOS
semidefinite formulation. Feasibility of the positive-rate
conditions implies stability over the full feasible set.
The numerical-certification procedure evaluates the reconstructed
representations through coefficient, semidefinite, and Lyapunov
checks.

\item \textbf{Systematic computational evidence beyond Minty.}
The complete $243$-system design yields solver-feasible SOS
representations under Euclidean, entropic, and Burg geometries,
including all $111$ systems with explicitly verified Minty failure.
The regularizer has a pronounced effect on the degrees and sizes
of the returned candidates. Numerical-certificate acceptance is
assessed separately from solver termination. Although SDP size limits the strategy resolutions that can be
explored computationally, the complete design with two to four
segments permits a systematic assessment across economic
primitives and learning geometries.
The search returns solver-feasible Lyapunov candidates throughout
this design, including instances where Minty stability fails.
\end{enumerate}

Section~\ref{sec:related-work} reviews the related literature. Section~\ref{sec:model} introduces the model, equilibrium, projected gradient, and Minty-failure result. Section~\ref{sec:rrm} specifies the proximal learning rule and gradient oracles. Section~\ref{sec:convergence} develops the Lyapunov implication and SOS formulation, and Section~\ref{sec:certified-scope} reports the computational study. Section~\ref{sec:conclusion} concludes. Detailed proofs and instance records appear in the supplement.

\section{Related Work}
\label{sec:related-work}

\paragraph{Learning dynamics beyond variational stability.} The stochastic-approximation perspective originates with \citet{RobbinsMonro1951} and is developed in dynamical-systems form by \citet{Benaim1999}, \citet{Borkar2008}, and \citet{BenaimHofbauerSorin2005}. Regularized learning equips the feasible set with a non-Euclidean geometry and leads, in continuous time, to Hessian or projected Riemannian dynamics \citep{MertikopoulosSandholm2018}. For stable, monotone, or variationally stable games, this literature provides algorithm-independent convergence results and canonical Lyapunov functions \citep{HofbauerSandholm2009,mertikopoulos2019learning}. Much less is known once these operator conditions fail. One must then analyze the particular vector field generated jointly by the game, approximation, and regularizer, and a Lyapunov function is generally not available in closed form.

Our paper occupies precisely this less-developed regime. The Minty condition fails locally in the underlying function space and, beyond a critical resolution, in the projected game. We therefore do not weaken Minty or substitute another variational inequality. Instead, we construct and verify a Lyapunov function for the resulting constrained dynamics. For exact and vanishing-error oracles, a direct compactness argument proves tracking of the projected dynamics. For persistent martingale errors and square-summable steps, we instead prove convergence directly from a proximal Lyapunov estimate. This avoids imposing a tracking hypothesis on the nonlinear boundary projection; the limit-set interpretation of the tracking case follows the stochastic-approximation perspective of \citet{BenaimHofbauerSorin2005}.

\paragraph{Continuous-type Bayesian equilibrium computation.} Numerical methods for Bayesian games must represent strategy functions or related infinite-dimensional objects. \citet{HoangSoumisZaccour2019} use return functions and establish convergence of a best-reply-based computation under topological conditions. Neural Pseudogradient Ascent and Simultaneous Online Dual Averaging instead learn parametrized or discretized strategies \citep{bichler2021npga,bichler2023soda}; related work develops ex post verification procedures for the resulting approximate equilibria \citep{pieroth2025}. These contributions primarily address equilibrium computation and verification rather than global last-iterate stability of the learning dynamics.

The closest convergence analysis is \citet{BichlerKohringOberlechnerPieroth2023}, who study bargaining games with continuous types and actions. They show that global monotonicity can fail, prove local convergence for a special uniform-prior linear-strategy model, and document empirical convergence more broadly. Likewise, \citet{bichler2025beyond} develops a function-space variational formulation and nonmonotonicity counterexamples for standard auction games. The present paper advances this line in three distinct directions: it proves a local Minty failure result for a Bayesian Bertrand family; it identifies when finite-dimensional strategy approximations inherit that failure; and, most importantly, it supplies constructive global Lyapunov certificates for the projected learning dynamics. To our knowledge, this is the first global Lyapunov/SOS stability analysis of a learning rule arising from a Bayesian game with both continuous types and continuous actions.

\paragraph{Computational Lyapunov functions and SOS programming.} The Lyapunov direct method establishes stability through a scalar function decreasing along solutions \citep{Lyapunov1992,Khalil2002}. For polynomial systems, sum-of-squares optimization makes the search for such a function computationally tractable by replacing polynomial nonnegativity with semidefinite constraints. The construction draws on Putinar-type semialgebraic positivity certificates and the SOS hierarchy \citep{Putinar1993,Parrilo2000,Lasserre2001}; \citet{PapachristodoulouPrajna2005} survey its use in nonlinear-systems analysis. Constrained systems require additional care because the vector field changes on active faces, which motivates the boundary conditions used below \citep{souaibyComputationLyapunovFunctions2020}.

Within operations research, \citet{HohmannWarringtonLygeros2020} combine moment/SOS methods with dual dynamic programming for polynomial control problems. Our objective and construction differ. We certify stability rather than approximate a value function; the game gradient is rational before positive-denominator clearing; and the certificate must be compatible with a metric-dependent projection on every active face. The novelty is thus not SOS optimization by itself, but its integration with a continuous-type Bayesian game, a strategy-function approximation, and three learning geometries. The output is a global SOS representation whose numerical validity can be assessed from reconstructed coefficient identities and semidefinite residuals.

\paragraph{Bayesian Bertrand competition.} The model is connected to first-price auctions and Bayesian price competition. Standard auction theory supplies the first-order characterization of the symmetric benchmark \citep{Krishna2010}, and the all-or-nothing-demand case $s=0$ corresponds to a first-price auction after a change of variables. Our earlier variational analysis \citep{bichler2025beyond} covers that limiting connection; the present Minty result extends it to polynomial demand and identifies a critical approximation resolution.

Most closely related economically, \citet{LiuSunXu2025} study Bayesian differentiated-product price competition under multinomial-logit demand. Their uniqueness and computational convergence results use strong monotonicity of a sign-adjusted payoff pseudo-gradient, and their numerical work indicates that the sufficient condition can fail in relevant instances. We address the complementary case in which even the weaker Minty condition fails locally. Rather than abandon convergence analysis at that boundary, we certify the particular projected dynamics directly. This distinction positions the paper between equilibrium-computation research in continuous Bayesian games and computational stability analysis in operations research.

\paragraph{Algorithmic pricing.} The broader motivation is related to work on learning and autonomous pricing, including the computational evidence of \citet{Calvano2020} and the overview of \citet{BichlerDurmannOberlechner2025}. Our object is deliberately different: the shared parameter vector implements a symmetric equilibrium solver, not independently adapting pricing agents. The contribution is a certification method for that numerical procedure and should not be read as a behavioral model of autonomous collusion.

\section{Bayesian Bertrand Model}\label{sec:model}

This section introduces the Bayesian Bertrand environment and the objects used throughout the analysis. We first specify the homogeneous-good pricing model and its complete-information benchmark, then introduce private costs and characterize the symmetric Bayesian Nash equilibrium. We explain why the Beta$(1,\ell)$ family is particularly useful for our purposes, show that the underlying continuous-strategy game fails the Minty condition even locally, and finally construct the finite-dimensional piecewise-linear representation and the associated unilateral-deviation game gradient that drives the learning dynamics studied in the remainder of the paper.

\subsection{Environment and complete-information benchmark}

We consider $n\ge2$ firms producing a homogeneous good. Firm $i$ has marginal cost $c_i\in[0,1]$ and chooses a price $p_i\in[0,1]$. Let $D:\mathbb R_+\to\mathbb R_+$ denote market demand at the lowest price and let
\begin{equation}
	A_i(\mathbf p) = \chi_{\{p_i<\min_{j\neq i}p_j\}} + \frac{1}{m(\mathbf p)} \chi_{\{p_i=\min_{j\neq i}p_j\}},
\end{equation}
where $m(\mathbf p)$ is the number of firms tying at the lowest price. Firm $i$ sells \(q_i(\mathbf p)=D(p_i)A_i(\mathbf p)\) and earns \( U_i^{\exp}(\mathbf p)=q_i(\mathbf p)(p_i-c_i). \) We study the polynomial demand family
\begin{equation}
	D^s(p)=(1-p)^s,\qquad s\in\mathbb N_0,
\end{equation}
and focus computationally on $s\in\{0,1,2\}$, corresponding to all-or-nothing, linear, and quadratic demand. This family provides increasingly curved demand benchmarks while preserving the polynomial structure required for numerical stability certification. Define the associated profit kernel
\begin{equation}
	k^s(p,c):=D^s(p)(p-c)=(1-p)^s(p-c).
\end{equation}
Then \( U_i^{\exp}(\mathbf p)=k^s(p_i,c_i)A_i(\mathbf p). \)

\begin{remark}
	Homogeneous goods imply that, conditional on winning the market, demand depends only on the winning firm's own price. The economic formulation can therefore accommodate other baseline demand functions by redefining $D$ and $k$. The SOS formulation developed below, however, additionally exploits the polynomial structure of $D^s$.
\end{remark}

Under complete information and a common marginal cost $c_i=c$, the unique symmetric pure-strategy Nash equilibrium is marginal-cost pricing, \(p_i^*=c\) for every $i$, and all firms earn zero profit. For $n\ge3$, asymmetric pure-strategy equilibria also exist: at least two firms may price at $c$ while the others price above $c$. Thus uniqueness here refers to the symmetric price profile.

\subsection{Incomplete information}

We now let marginal costs be privately observed. A pure strategy of firm $i$ is a pricing function \( \beta_i:C_i\to P_i \), and \( \boldsymbol\beta=(\beta_1,\ldots,\beta_n) \) denotes a strategy profile. Expected utility is
\begin{equation}
	U_i(\boldsymbol\beta) = \mathbb E_{\mathbf c} \left[U_i^{\exp}(\boldsymbol\beta(\mathbf{c}))\right]=\mathbb E_{\mathbf c} \left[ k^s(\beta_i(c_i),c_i) A_i\bigl(\beta_1(c_1),\ldots,\beta_n(c_n)\bigr) \right].
\end{equation}

\begin{assumption}\label{ass:iid}
	Throughout the main analysis:
	\begin{enumerate}
		\item[(i)] \textbf{Independent private costs.} Costs are i.i.d.\ according to an absolutely continuous distribution $F$ on $[0,1]$ with density $f>0$ almost everywhere.

		\item[(ii)] \textbf{Symmetric environment and strategies.} Firms have identical payoff functions and type and action spaces, and we restrict attention to symmetric profiles.

		\item[(iii)] \textbf{Regularity class for symmetric profiles and computation.} We analyze pricing strategies in
		\begin{equation}
			\mathcal B_\delta := \left\{ \beta\in V: 0\le\beta\le1,\; \beta'(c)\ge\delta>0\ \text{a.e.},\; \beta(1)=1 \right\}, \qquad V=W^{1,1}([0,1]).
		\end{equation}
	\end{enumerate}
\end{assumption}

Independence is primarily a tractability restriction rather than an economic claim that firms' costs must be statistically independent. Correlated-cost models allow a firm's beliefs about rivals to depend on its own cost realization \citep{LiuSunXu2025}. We impose independence to isolate the stability question from this additional informational channel.

Because $F$ is continuous and admissible strategies are strictly increasing, price ties occur with probability zero. Let
\begin{equation}
	H(z) = 1-(1-F(z))^{n-1}, \qquad h(z)=H'(z),
\end{equation}
denote the distribution and density of the lowest rival cost \(Z=\min_{j\ge2}c_j\). Against a common rival strategy $\tilde\beta$, define the winning probability
\begin{equation}
	W_{\tilde\beta}(p) := \int_0^1 \chi_{\{p<\tilde\beta(z)\}} \,dH(z).
\end{equation}
For $p$ in the interior of the range of $\tilde\beta$, \( W_{\tilde\beta}(p) = 1-H(\tilde\beta^{-1}(p)). \) Hence the expected payoff of a focal firm using $\beta$ is
\begin{equation}
	U(\beta,\tilde{\boldsymbol\beta}_{-1}) = \int_0^1 k^s(\beta(c),c) W_{\tilde\beta}(\beta(c)) \,dF(c).
	\label{eq:reduced-utility}
\end{equation}

\subsection{Symmetric equilibrium}

\begin{definition}\label{def:sym-bne-ref}
A symmetric Bayesian Nash equilibrium (sym-BNE) is an admissible pricing strategy $\beta^*\in\mathcal B_\delta$ such that
\begin{equation}
 U(\beta^*,\boldsymbol\beta^*_{-1})\ge U(\beta,\boldsymbol\beta^*_{-1})
 \qquad\text{for every measurable }\beta:[0,1]\to[0,1].
\end{equation}
Equivalently, $\beta^*(c)$ is a best response among all prices in $[0,1]$ for $F$-almost every cost $c$.
\end{definition}

This is the economic equilibrium concept: the regularity restrictions on the candidate equilibrium do not restrict a firm's profitable deviations. Restricting deviations to $\mathcal B_\delta$, or to a piecewise-linear approximation of that set, instead defines a constrained strategy game. An equilibrium in the sense of Definition~\ref{def:sym-bne-ref} remains an equilibrium of every such restricted game containing its strategy. Uniqueness in the unrestricted economic game does not, by itself, exclude additional equilibria of a restricted game.

\begin{proposition}[Symmetric equilibrium]\label{prop:symmetric-bne}
	Under Assumption~\ref{ass:iid}, any sym-BNE in the sense of Definition~\ref{def:sym-bne-ref} for $D^s(p)=(1-p)^s$ satisfies, for almost every $c\in(0,1)$,
	\begin{equation}
		\beta'(c) = \frac{h(c)}{1-H(c)} \frac{(1-\beta(c))(\beta(c)-c)} {1-(s+1)\beta(c)+sc}.
		\label{eq:sym-bne-ode}
	\end{equation}
	If costs follow $\mathrm{Beta}(1,\ell)$ and $0<\delta\le r$, the unique admissible symmetric equilibrium is
	\begin{equation}
		\beta^*(c) = \frac{(\ell(n-1)+s)c+1}{\ell(n-1)+s+1} = 1-r(1-c), \qquad r= \frac{\ell(n-1)+s}{\ell(n-1)+s+1}.
		\label{eq:closed-form-bne}
	\end{equation}
\end{proposition}
For later use, write
\begin{equation}
	a=\ell(n-1), \qquad \alpha=a+s,
	\label{eq:alpha-r}
\end{equation}
implying $r=\frac{\alpha}{\alpha+1}$.

The proof is given in the supplement \ref{app:proof_symmetric_bne}. Equation~\eqref{eq:sym-bne-ode} follows by differentiating the conditional payoff \( k^s(p,c)[1-H(\beta^{-1}(p))] \). For Beta$(1,\ell)$ costs and prices in the range of $\beta^*$, the deviation payoff is proportional to $(1-p)^{a+s}(p-c)$ and is uniquely maximized at $p=\beta^*(c)$. The supplement~\ref{app:proof_symmetric_bne} also checks deviations below that price range. Consequently the affine candidate is a BNE against unrestricted measurable deviations, as well as an equilibrium of every projected game considered here.

For the Beta$(1,\ell)$ family,
\begin{equation}
	F(c)=1-(1-c)^\ell, \qquad f(c)=\ell(1-c)^{\ell-1}, \qquad \frac{h(c)}{1-H(c)} = \frac{\ell(n-1)}{1-c}.
	\label{eq:beta-prior}
\end{equation}

\begin{proposition}[Characterisation of affine equilibrium solutions]
	\label{prop:affine-class}
	Let $F$ be absolutely continuous on $[0,1]$, with $F(0)=0$ and $f>0$. The equilibrium equation \eqref{eq:sym-bne-ode} admits an admissible affine solution \( \beta(c)=1-r(1-c) \) if and only if
	\begin{equation}
		(1-c)\frac{f(c)}{1-F(c)} \equiv \kappa
	\end{equation}
	for some constant $\kappa>0$, equivalently \( 1-F(c)=(1-c)^\kappa \). Thus affine equilibrium solutions arise exactly for Beta$(1,\kappa)$ priors, in which case
	\begin{equation}
		r= \frac{\kappa(n-1)+s} {\kappa(n-1)+s+1}.
	\end{equation}
\end{proposition}

The proof is given in the supplement, \ref{app:proof_affine}.

\begin{remark}[Role of the Beta$(1,\ell)$ prior]\label{rem:beta-kl}
	The main significance of the Beta$(1,\ell)$ assumption is that the equilibrium is affine. Hence $\beta^*$ belongs exactly to every piecewise-linear approximation space used below. Increasing the projection dimension therefore changes the learning dynamics without introducing equilibrium-approximation error. 
    For integer $n,s,\ell$ and rational knots, the corresponding
equilibrium parameter vector $x^*$ is also rational.
More general integer-parameter Beta priors retain polynomial game-gradient structure but generally lose this exact affine target as we discuss in the supplement \ref{app:general-priors}.
\end{remark}

\subsection{Finite-dimensional parametrization}
\label{sec:parametrization}

To obtain a finite-dimensional learning problem, partition $[0,1]$ into \( I_k=((k-1)/m,k/m] \), $k\in[m]$, and parameterize the slope of the pricing function by $\mathbf x=(x_1,\ldots,x_m)$:
\begin{equation}
	\beta'_{\mathbf x}(c) = \sum_{k=1}^m x_k\chi_{I_k}(c), \qquad \beta_{\mathbf x}(c) = 1-\int_c^1\beta'_{\mathbf x}(z)\,dz.
	\label{eq:pl-parametrization}
\end{equation}
The corresponding parameter set is
\begin{equation}
	\mathcal B_\delta^m = \left\{ \mathbf x\in\mathbb R^m: x_i\ge\delta,\ i\in[m], \quad \sum_{i=1}^m x_i\le m \right\}.
	\label{eq:constrained-set}
\end{equation}
The first constraints ensure strict monotonicity, while the last is equivalent to $\beta_{\mathbf x}(0)\ge0$. By construction, $\beta_{\mathbf x}(1)=1$.

For later differentiation define the \textit{sensitivity kernel}
\begin{equation}
	\psi_i(c) := -\frac{\partial\beta_{\mathbf x}(c)}{\partial x_i} = \lambda(I_i\cap[c,1]) =
	\begin{cases}
		\frac1m, &0\le c\le\frac{i-1}{m},\\[2pt] \frac{i}{m}-c, &\frac{i-1}{m}<c\le\frac{i}{m},\\[2pt] 0, &\frac{i}{m}<c\le1,
	\end{cases}
	\label{eq:psi}
\end{equation}
where $\lambda$ denotes Lebesgue measure. The representation by interval length also applies directly to non-uniform partitions.

\subsection{The projected game gradient}
\label{sec:game-gradient}

For a focal parameter vector $\mathbf y$ and common rival parameter $\mathbf x$, write
\begin{equation}
	u(\mathbf y;\mathbf x) := U(\beta_{\mathbf y}, \beta_{\mathbf x},\ldots,\beta_{\mathbf x}),
\end{equation}
and define the symmetric unilateral-deviation gradient
\begin{equation}
	v_i(\mathbf x) := \left. \frac{\partial u(\mathbf y;\mathbf x)} {\partial y_i} \right|_{\mathbf y=\mathbf x}, \qquad \mathbf v(\mathbf x) = (v_1(\mathbf x),\ldots,v_m(\mathbf x)).
	\label{eq:game-gradient-definition}
\end{equation}
Thus $\mathbf v$ differentiates the focal firm's strategy before imposing symmetry on the profile.

\begin{proposition}[Projected game gradient]
	\label{prop:game-gradient}
	Let $k$ be continuously differentiable in $p$, with $k$ and $\partial_pk$ bounded. Under Assumption~\ref{ass:iid}, Beta$(1,\ell)$ costs, and $\mathbf x\in\mathcal B_\delta^m$,
	\begin{align}
		v_i^{(s,\ell,n)}(\mathbf x)
		&= \int_0^1 \ell(1-c)^{\ell n-2}\psi_i(c)
		\left[ \frac{\ell(n-1)k(\beta_{\mathbf x}(c),c)} {\beta'_{\mathbf x}(c)} \right. \nonumber\\
		&\hspace{5.3cm}\left. {}- (1-c)\partial_pk(\beta_{\mathbf x}(c),c) \right]dc
		\label{eq:compact-gradient}
	\end{align}
	for every $i\in[m]$.
\end{proposition}

The proof is given in the supplement \ref{app:proof-game-gradient}. For integer $\ell$ and polynomial $k^s$, piecewise integration of \eqref{eq:compact-gradient} yields the rational representation used in the SOS analysis.

\begin{remark}[First-order interpretation]\label{rem:equilibrium}
	The vector $\mathbf v(\mathbf x)$ is the marginal incentive for a unilateral deviation from a symmetric profile. Hence every interior symmetric Nash equilibrium of the projected game satisfies \( \mathbf v(\mathbf x)=0 \). The converse need not hold without additional concavity or pseudoconcavity conditions: a zero of $\mathbf v$ is, in general, only a first-order stationary point.
\end{remark}

For the equilibrium \(\beta^*(c)=1-r(1-c)\), the bracket in \eqref{eq:compact-gradient} vanishes pointwise, since
\begin{equation}
	\frac{\ell(n-1)k(\beta^*(c),c)}{r} = (1-c)\partial_pk(\beta^*(c),c).
\end{equation}
Thus the affine equilibrium is represented by a rest point $\mathbf x^*=r\mathbf 1$ for every $m$.

\begin{remark}[Symmetric self-play]\label{rem:symmetric-scope}
	The learning dynamics studied below use a single shared parameter vector and therefore evolve on the symmetric manifold \( \mathcal D = \{(\mathbf x,\ldots,\mathbf x): \mathbf x\in\mathcal B_\delta^m\}. \) This invariance is imposed by the equilibrium-computation procedure rather than derived from the underlying $n$-player dynamics. Independent learners, whose parameter vectors and gradient noise evolve separately, are not covered by the present analysis.
\end{remark}

\begin{remark}[Fixed projection dimension]\label{rem:no-transfer}
	All stability results are for a fixed $m$. Because $\beta^*$ is affine, the target equilibrium is represented exactly for every $m$, but changing $m$ changes the deviation set and may introduce additional stationary points or boundary rest points. A global Lyapunov certificate rules these out for the particular projected dynamics being studied; convergence for every fixed $m$ does not by itself imply convergence in the original infinite-dimensional strategy space.
\end{remark}

\subsection{Failure of Minty stability}
\label{sec:continuous-minty}

A standard route to last-iterate convergence is the Minty variational inequality. Let \(D_1U(\beta,\boldsymbol\beta_{-1})[d]\) denote the G\^ateaux derivative with respect to the focal firm's strategy, holding the rivals fixed at the same strategy $\beta$. This derivative for our Bertrand competition model is

\begin{align*}
	D_1U(\beta,\boldsymbol{\tilde{\beta}}_{-1})[d]
	= \int_0^1 d(c) \Bigg[&\partial_p k(\beta(c),c)W_{\tilde{\beta}}(\beta(c)) {}-k(\beta(c),c)\chi_{\{\tilde{\beta}(0)<\beta(c)<1\}} \\
	&\qquad\times
	\frac{h\!\left(\tilde{\beta}^{-1}(\beta(c))\right)}
	{\tilde{\beta}'\!\left(\tilde{\beta}^{-1}(\beta(c))\right)}\Bigg]f(c)\,dc.
\end{align*}

The derivation is provided in the supplement, Section~\ref{app:gateaux}. The symmetric Minty condition at $\beta^*$ requires
\begin{equation}
	D_1U(\beta,\boldsymbol\beta_{-1}) [\beta-\beta^*] \le 0 \qquad \text{for every }\beta\in\mathcal B_\delta .
	\label{eq:symmetric-minty}
\end{equation}

The following result shows that this condition is not merely violated far from equilibrium. It fails locally in the original function space and, for sufficiently rich piecewise-linear approximations, also after projection.

\begin{theorem}[Failure of the Minty condition]
	\label{thm:no-minty}
	Consider the Bertrand model under Assumption~\ref{ass:iid} with $n\in\mathbb{N}_{\geq2}$ players, with $D^s(p)=(1-p)^s$, $s\in\mathbb N_0$, and Beta$(1,\ell)$ costs, $\ell\in\mathbb N$.

	\begin{enumerate}
		\item[(i)] \emph{Continuous-strategy game.} For every lower-slope bound $0<\delta<r$, the symmetric equilibrium $\beta^*$ does not satisfy \eqref{eq:symmetric-minty}, even locally: arbitrarily close to $\beta^*$ in $\mathcal B_\delta$ there are admissible strategies that violate the Minty inequality.

		\item[(ii)] \emph{Projected game.} For the uniform $m$-segment parametrization of Section~\ref{sec:parametrization}, there exists a finite sufficient witness resolution $m^\star=m^\star(n,\ell,s)$ such that, for every $m\ge m^\star$ and sufficiently small lower-slope bound $\delta$, the projected game also violates the Minty condition. Moreover, \(m^\star\le \max\{2,\ell(n-1)+s\}\).
	\end{enumerate}
\end{theorem}

Since monotonicity of the game operator implies the Minty condition at any solution of the associated variational inequality, Theorem~\ref{thm:no-minty} also rules out monotonicity (and hence strong monotonicity) on the corresponding strategy domain.

\paragraph{Proof sketch.}
The function-space construction perturbs $\beta^*$ only on $[0,2/j]$, using slopes $r-\eta$ and $r+\eta$ on successive intervals of length $1/j$, where $0<\eta<r-\delta$. The perturbations converge to $\beta^*$ in $W^{1,1}([0,1])$, while their Minty expression equals
\[
 \frac{B(0)\eta^3}{r(r^2-\eta^2)}j^{-2}+O(j^{-3}),
 \qquad B(c)=k^s(\beta^*(c),c)h(c)f(c),\quad B(0)>0.
\]
It is therefore positive for sufficiently large $j$.

For the projected result, flatten the first cell to slope $\delta$ and optimize its intercept. With $t=1/m$ and
\[
 A_k(t)=\int_{1-t}^{1}u^{\ell n-2+k}\,du,\qquad k=0,1,2,
\]
the coefficient of the leading $1/\delta$ term is positive when
\[
 \rho_m:=\frac{A_0(1/m)A_2(1/m)}{A_1(1/m)^2}
 <1+\frac{1}{4\alpha(\alpha+1)},\qquad \alpha=\ell(n-1)+s.
\]
Since $\rho_m\to1$, this criterion holds eventually. The bound
$\rho_m\le1+[3(2m-1)^2]^{-1}$ gives the stated upper bound on $m^\star$.
The supplement~\ref{app:proof-no-minty} gives the complete constructions and proof. Throughout, $m^\star$ denotes the threshold for this sufficient witness criterion; smaller values of $m$ are not classified as Minty-stable.

\section{Regularized Robbins--Monro Learning}
\label{sec:rrm}

This section introduces the learning dynamics studied in the remainder of the paper. We first formulate the proximal regularized Robbins--Monro recursion and the stochastic-approximation conditions imposed on its gradient oracle. We then describe the Burg, entropic, and Euclidean geometries used in the numerical study. The choice of oracle determines the stochastic-approximation error, whereas the regularizer determines the geometry of the induced dynamics.

\subsection{RRM algorithm and stochastic approximation}
\label{sec:rrm-algorithm}

We use a proximal regularized Robbins--Monro recursion within the stochastic approximation framework of \citet{RobbinsMonro1951}, building on \citet{BenaimHofbauerSorin2005} and \citet{MertikopoulosHsiehCevher2024}. The framework covers familiar gradient and mirror-descent learning rules while allowing both stochastic and biased gradient feedback. Let
\[
	\mathcal X=\mathcal B_\delta^m\subset\mathbb R^m
\]
denote the finite-dimensional feasible strategy set introduced in Section~\ref{sec:parametrization}. The ambient parameter space is $\mathbb R^m$, whereas all feasible learning iterates lie in $\mathcal X$.

\begin{definition}[Regularizer]
	A function \( h:\mathcal X\to\mathbb R \) is a regularizer on $\mathcal X$ if it is continuous and strongly convex. 
	Its associated mirror map is \( Q(y) = \operatorname*{arg\,max}_{x\in\mathcal X} \left\{ \langle y,x\rangle-h(x) \right\}. \) Where $h$ is twice differentiable on the interior of $\mathcal X$, we write \( H(x)=\nabla^2 h(x). \)
\end{definition}

The inverse Hessian $H(x)^{-1}$ determines the local preconditioning of the game gradient. The corresponding continuous-time interior field is \(H(x)^{-1}v(x)\); on the boundary of $\mathcal X$ it must be projected onto the feasible tangent cone.

\begin{definition}[Proximal regularized Robbins--Monro recursion]
	\label{def:rrm}
	Given a regularizer $h$, step sizes $\gamma_j>0$, and a gradient estimate $\hat v_j$, define
	\begin{equation}
		\label{eq:rrm-greedy}
		X_{j+1} = \operatorname*{arg\,max}_{x\in\mathcal X} \left\{ \gamma_j\langle \hat v_j,x\rangle - D_h(x,X_j) \right\},
	\end{equation}
	where \( D_h(x,x') = h(x)-h(x')-\langle\nabla h(x'),x-x'\rangle \) is the Bregman divergence generated by $h$.

	Equivalently, using the mirror map
	\[
		Q(y) = \operatorname*{arg\,max}_{x\in\mathcal X} \{\langle y,x\rangle-h(x)\},
	\]
	the update can be written as
	\begin{equation}
		\label{eq:rrm-mirror}
		X_{j+1} = Q\!\left( \nabla h(X_j)+\gamma_j\hat v_j \right).
	\end{equation}
\end{definition}

The recursion~\eqref{eq:rrm-greedy} is the primal proximal form of regularized stochastic approximation. It should be distinguished from the dual-averaging RRM template of \citet{MertikopoulosHsiehCevher2024}, which maintains a separate cumulative dual variable. We use the proximal form because its small-step deterministic limit is the projected Riemannian dynamics studied in Section~\ref{sec:convergence}.

When the oracle is exact, \( \hat v_j=v(X_j), \) \eqref{eq:rrm-greedy} reduces to deterministic mirror ascent. With the Euclidean regularizer it becomes projected gradient ascent, while non-Euclidean regularizers such as negative entropy generate standard mirror-ascent updates. The stochastic formulation additionally allows martingale noise and systematic approximation bias.

The analysis below concerns symmetric self-play: all firms share the same parameter vector $X_j$ and hence remain on the symmetric manifold described in Remark~\ref{rem:symmetric-scope}. We do not interpret \eqref{eq:rrm-greedy} as a model of independently adapting pricing agents.

Let $\mathcal F_n$ denote the history before the $n$th oracle call and decompose the gradient estimate as
\begin{equation}
	\hat v_n=v(X_n)+U_n+b_n,
	\qquad \mathbb E[U_n\mid\mathcal F_n]=0,
	\label{eq:oracle-decomposition}
\end{equation}
where $U_n$ is a martingale-difference error and $b_n$ is an $\mathcal F_n$-measurable bias term.

\begin{assumption}[Stochastic approximation conditions]
	\label{ass:parameter conditions}
	Let $v$ be the projected game gradient of Proposition~\ref{prop:game-gradient}.
	\begin{enumerate}
		\item[(H0)] \textit{Regularity.} The vector field $v$ is bounded and Lipschitz continuous on $\mathcal X$.

		\item[(H1)] \textit{Step sizes.} For some $\ell_\gamma\in(0,1]$, \( \gamma_n=n^{-\ell_\gamma} \), so that $\gamma_n\to0$ and \( \sum_n\gamma_n=\infty \).

		\item[(H2)] \textit{Martingale noise.} For some moment order $\zeta\ge2$, with deterministic constants in the following bounds,
		\begin{equation}
			\mathbb E \left[ \lVert U_n\rVert_*^\zeta \,\middle|\, \mathcal F_n \right] = \mathcal O\bigl(n^{-\zeta\ell_\sigma}\bigr), \qquad \ell_\gamma+\ell_\sigma>\frac12.
		\end{equation}
		The additional square-summable-step condition used for persistent-noise convergence is stated separately in Corollary~\ref{cor:Convergence-result}.

		\item[(H3)] \textit{Bias.} For the same $\zeta$,
		\begin{equation}
			\mathbb E \left[ \lVert b_n\rVert_*^\zeta \,\middle|\, \mathcal F_n \right] = \mathcal O\bigl(n^{-\zeta\ell_b}\bigr), \qquad \zeta\ell_b>1.
		\end{equation}
	\end{enumerate}
\end{assumption}

\begin{remark}[Interpretation of Assumption~\ref{ass:parameter conditions}]
	\label{rem:assumption-form}
	Condition (H0) follows from the structure of the projected Bertrand game. In particular, Proposition~\ref{prop:game-gradient} expresses each component of $v$ as an integral involving the piecewise-constant slopes $\beta'_{\mathbf x}\ge\delta$. Hence the relevant denominators are uniformly bounded away from zero on the compact set $\mathcal X$. For integer $\ell$ and polynomial demand, the resulting components are rational functions whose common denominator divides \( p(x)=\prod_kx_k \).
\end{remark}

Condition (H1) is the standard Robbins--Monro requirement that steps vanish while their cumulative sum remains infinite. Condition (H2) ensures square-summability of the weighted martingale noise; in particular,
\begin{equation}
	\sum_j \gamma_j^2 \mathbb E \left[ \lVert U_j\rVert_*^2 \mid\mathcal F_j \right] <\infty.
\end{equation}
Finally, (H3) implies $b_n\to0$ almost surely by Markov's inequality and the Borel--Cantelli lemma. The following result discharges the mean-dynamics tracking requirement for the exact-gradient algorithm and for increasing-accuracy stochastic oracles.

\begin{proposition}[Tracking by the proximal update]
\label{prop:proximal-tracking}
Suppose that $\mathcal X$ is the compact polyhedron of Section~\ref{sec:parametrization}, $v$ satisfies (H0), and $h$ is three times continuously differentiable on a neighborhood of $\mathcal X$ with uniformly positive-definite Hessian. Let $(X_n)$ follow \eqref{eq:rrm-greedy}, with $\gamma_n\to0$ and $\sum_n\gamma_n=\infty$. If
\begin{equation}
 \label{eq:vanishing-oracle-error}
 \|U_n+b_n\|_*\longrightarrow0\qquad\text{almost surely},
\end{equation}
then its affine interpolation on $t_n=\sum_{j<n}\gamma_j$ is almost surely an asymptotic pseudo-trajectory of the projected dynamics~\eqref{eq:PD}.
\end{proposition}

\begin{proof}[Proof outline]
The proximal optimality condition is
\begin{equation}
 \label{eq:proximal-normal-inclusion}
 \overline H_n\frac{X_{n+1}-X_n}{\gamma_n}
 =v(X_n)+U_n+b_n-\xi_{n+1},\qquad
 \xi_{n+1}\in N_{\mathcal X}(X_{n+1}),
\end{equation}
where $\overline H_n=\int_0^1H(X_n+u(X_{n+1}-X_n))\,du$.
On each sample path satisfying \eqref{eq:vanishing-oracle-error}, the scaled increments and the selected normal corrections are uniformly bounded. Time-shifted interpolations therefore admit uniformly convergent subsequences. Passing to the limit in the normal-cone inequalities shows that every limit solves
\begin{equation}
 \label{eq:normal-cone-dynamics}
 H(x)\dot x=v(x)-\xi,\qquad \xi\in N_{\mathcal X}(x),\quad x\in\mathcal X.
\end{equation}
Viability gives $\langle\xi,\dot x\rangle=0$ almost everywhere, so the limit is the metric projection in \eqref{eq:PD}. Existence, uniqueness, continuous dependence, and the complete passage to the limit are proved in the supplement~\ref{app:tracking-proof}. The argument uses the bounded normal corrections selected by the recursion; it does not apply a compact-valued differential-inclusion theorem to an unrestricted normal cone.
\end{proof}

The condition holds immediately for exact gradients. It also holds for increasing-accuracy oracles: if $\zeta\ell_\sigma>1$, (H2) and Borel--Cantelli give $U_n\to0$ almost surely, while (H3) gives $b_n\to0$. Persistent martingale noise is covered separately by the direct Lyapunov argument in Corollary~\ref{cor:Convergence-result}, with square-summable steps.

\subsection{Gradient oracles}
\label{sec:gradient-oracles}

The oracle and the regularizer play distinct roles. The oracle determines the noise and bias terms in \eqref{eq:oracle-decomposition}; the regularizer determines the geometry of the update. Thus the same oracle can be combined with any of the regularizers considered below. Table~\ref{tab:rrm_bounds} summarizes the oracle specifications used in our analysis.

\begin{table}[htpb]
\centering
\caption{Gradient oracles for \eqref{eq:rrm-greedy}. The moment and bias conditions of Assumption~\ref{ass:parameter conditions} apply whenever the corresponding errors are present.}
\label{tab:rrm_bounds}
\begin{tabular}{lll}
\toprule
Oracle & Step-size exponent & Convergence route \\
\midrule
Exact gradient & $0<\ell_\gamma\le1$ & Tracking and Lyapunov stability \\
Increasing accuracy, $\zeta\ell_\sigma>1$ & $0<\ell_\gamma\le1$ & Tracking and Lyapunov stability \\
Fixed-variance unbiased gradient & $1/2<\ell_\gamma\le1$ & Direct proximal Lyapunov bound \\
\bottomrule
\end{tabular}
\end{table}

Exact gradients produce a diminishing-step discretization of the projected dynamics. Increasing batch sizes can enforce almost-surely vanishing oracle error. Fixed-variance unbiased gradients need not have vanishing samplewise error: their weighted martingale contribution is controlled by square-summable steps. The boundary condition of the Lyapunov certificate also controls each proximal normal correction, so this third case requires no additional mean-dynamics tracking assumption. Optimistic, extra-gradient, and one-point payoff schemes require separate analyses and are not covered here.

\subsection{Considered regularization geometries}
\label{sec:regularizer-geometries}

We parameterize the regularization geometry by the Tsallis family
\begin{equation}
h^{(q)}(x)=\sum_{i=1}^m h_i^{(q)}(x_i), \qquad
h_i^{(q)}(x_i)=
\begin{cases}
\dfrac{x_i^q-qx_i}{q(q-1)}, & q\notin\{0,1\},\\
x_i\log x_i-x_i, & q=1,\\
-\log x_i-x_i, & q=0,
\end{cases}
\end{equation}
whose inverse Hessian is
\begin{equation}
H_q(x)^{-1}=\operatorname{diag}\left(x_1^{2-q},\ldots,x_m^{2-q}\right).
\end{equation}
We focus on \(q\in\{0,1,2\}\), corresponding to the Burg, entropic, and Euclidean regularizers, respectively. These choices lead directly to polynomial preconditioned dynamics and therefore admit particularly compact semidefinite-programming (SDP) formulations; see Section~\ref{sec:polynomializing}.
\begin{remark}
     The framework can be extended to integer \(q\geq 0\). For \(q>2\), however, \(H_q(x)^{-1}\) contains negative powers of the state variables, so obtaining a polynomial representation of the dynamics requires clearing the resulting denominators. This increases the polynomial degree of the resulting constraints and, consequently, the size of the associated SDP relaxations, which can grow substantially with \(q\).
\end{remark}

The choice of regularizer therefore affects both the preconditioned learning dynamics and the computational complexity of their certification, whereas the Minty property introduced in Section~\ref{sec:continuous-minty} depends only on the underlying game operator.

\section{Lyapunov Stability and Constructive Certification}
\label{sec:convergence}

We now turn from the specification of the learning rule to its convergence analysis. The argument has two steps. First, we formulate the deterministic projected Riemannian dynamics associated with the RRM recursion and state the Lyapunov condition that implies almost-sure last-iterate convergence of the stochastic iterates. Second, we show how the required global Lyapunov functions can be constructed computationally by sum-of-squares (SOS) semidefinite programming.

Throughout this section, let \( \mathcal X=\mathcal B_\delta^m \) and suppose that the regularizer $h$ is three times continuously differentiable on a neighborhood of $\mathcal X$, with uniformly positive-definite Hessian \( H(x)=\nabla^2 h(x) \). These conditions are satisfied by the regularizers considered in Section~\ref{sec:rrm} because $x_i\ge\delta>0$ on $\mathcal X$.

\subsection{Projected dynamics and Lyapunov stability}
\label{sec:projected-dynamics}

The unconstrained continuous-time field associated with the RRM update is $H(x)^{-1}v(x)$. Since $\mathcal X$ is constrained, this field must be projected onto the feasible tangent cone whenever an iterate reaches the boundary.

\begin{definition}[Projected Riemannian dynamics]
	\label{def:primal-dynamics}
	For $x\in\mathcal X$, let $T_{\mathcal X}(x)$ denote the tangent cone of $\mathcal X$ at $x$. The projected Riemannian dynamics associated with $(\mathcal X,h,v)$ are
	\begin{equation}
		\label{eq:PD}
		\dot x = \Pi^h_{T_{\mathcal X}(x)} \bigl(H(x)^{-1}v(x)\bigr),
	\end{equation}
	where \( \Pi^h_{T_{\mathcal X}(x)}(z) = \operatorname*{arg\,min}_{y\in T_{\mathcal X}(x)} (z-y)^\top H(x)(z-y) \) is the projection onto the tangent cone in the local metric induced by $h$. We denote these dynamics by \( \mathcal D(\mathcal X,h,v) \).
\end{definition}

\begin{remark}[Well-posedness]
\label{rem:pds-wellposed}
The vector field in \eqref{eq:PD} can be discontinuous where the active face changes. We use absolutely continuous viable solutions of the normal-cone system \eqref{eq:normal-cone-dynamics}. The supplement~\ref{app:tracking-proof} proves existence and uniqueness in the present compact convex setting and establishes continuous dependence on the initial state. Its proof exploits the Hessian metric and monotonicity of the Euclidean normal cone, rather than assuming continuity of the projected field.
\end{remark}

For later use, write the feasible set as
\begin{equation}
	\mathcal X = \{x:g_i(x)\ge0,\ i \in [m+1]\},
	\label{eq:polyhedral-description}
\end{equation}
with \( g_i(x)=x_i-\delta \) for $i\in[m]$ and \( g_{m+1}(x)=m-\sum_{i=1}^m x_i. \)

The following definition gives conditions under which a Lyapunov function for the unconstrained interior field remains a Lyapunov function after projection at the boundary.

\begin{definition}[Global Lyapunov function]
	\label{def:riemannian-Lyapunov}
	A continuously differentiable function $L:\mathcal X\to\mathbb R$ is a \emph{global Lyapunov function} for $x^*$ under $\mathcal D(\mathcal X,h,v)$ if the following conditions hold:
	\begin{enumerate}
		\item[(i)] \emph{Positive definiteness.}
		\[
			L(x^*)=0, \qquad L(x)>0 \quad\text{for every }x\neq x^*.
		\]

		\item[(ii)] \emph{Strict interior decrease.} There exists $w>0$ such that
		\begin{equation}
			\label{eq:interior-decrease}
			\left\langle \nabla L(x),H(x)^{-1}v(x) \right\rangle \le -w\lVert x-x^*\rVert^2 \quad\text{for every }x\in\mathcal X.
		\end{equation}

		\item[(iii)] \emph{Boundary compatibility.} For every active constraint $g_i(x)=0$,
		\begin{equation}
			\label{eq:riemannian-boundary-condition}
			\left\langle \nabla L(x), H(x)^{-1}\nabla g_i(x) \right\rangle \le0.
		\end{equation}
	\end{enumerate}
\end{definition}

Following the constrained-Lyapunov principle of \citet{souaibyComputationLyapunovFunctions2020}, we impose a sufficient facewise condition ensuring that the Riemannian projection correction cannot increase the Lyapunov derivative. In particular, Condition~\eqref{eq:riemannian-boundary-condition} accounts for the projection explicitly. Indeed, at a boundary point the projected direction can be written as the interior field plus a nonnegative combination of terms of the form $H^{-1}\nabla g_i$ associated with active constraints. Hence \eqref{eq:riemannian-boundary-condition} guarantees that the projection cannot increase the derivative of $L$. Together with \eqref{eq:interior-decrease}, this gives strict decrease of $L$ along every nonstationary solution of \eqref{eq:PD}.

An exact certificate therefore establishes the behavior for every initial condition in $\mathcal X$, rather than for a finite collection of simulated trajectories.

The deterministic Lyapunov condition can now be combined with standard stochastic-approximation arguments.

\begin{proposition}[One-step proximal Lyapunov bound]
\label{prop:one-step-lyapunov}
Let $L$ be twice continuously differentiable on a neighborhood of $\mathcal X$ and satisfy the boundary condition~\eqref{eq:riemannian-boundary-condition}. Define
\[
 a_L(x):=H(x)^{-1}\nabla L(x).
\]
There is a finite constant $C_L$, depending only on $\mathcal X,h,L$, such that every proximal step with vector $\hat v_n$ satisfies
\begin{equation}
 \label{eq:one-step-lyapunov}
 L(X_{n+1})-L(X_n)
 \le\gamma_n\langle a_L(X_n),\hat v_n\rangle
       +C_L\gamma_n^2\|\hat v_n\|_*^2.
\end{equation}
\end{proposition}

\begin{proof}
Write $x=X_n$, $z=X_{n+1}$, $d=z-x$, and let $\mu I\preceq H\preceq\Lambda I$ on $\mathcal X$. Optimality gives
\[
 \overline H d=\gamma_n(\hat v_n-\xi),\qquad
 \xi\in N_{\mathcal X}(z),\qquad
 \overline H=\int_0^1H(x+ud)\,du.
\]
Taking the scalar product with $d$ and using $\langle\xi,d\rangle\ge0$ yields
$\|d\|\le C\gamma_n\|\hat v_n\|_*$, and consequently
$\|\xi\|\le C\|\hat v_n\|_*$.
The polyhedral normal cone is generated by $-\nabla g_i(z)$ for active constraints. Boundary compatibility therefore gives
$-\langle a_L(z),\xi\rangle\le0$.
Taylor expansion at $z$ and the Lipschitz continuity of $H^{-1}$ and $a_L$ on $\mathcal X$ imply
\begin{align*}
 L(z)-L(x)
 &\le \langle\nabla L(z),d\rangle+C\|d\|^2\\
 &=\gamma_n\langle\overline H^{-1}\nabla L(z),\hat v_n-\xi\rangle+C\|d\|^2\\
 &\le\gamma_n\langle a_L(x),\hat v_n\rangle
   -\gamma_n\langle a_L(z),\xi\rangle
   +C\gamma_n\|d\|\bigl(\|\hat v_n\|_*+\|\xi\|\bigr)+C\|d\|^2\\
 &\le\gamma_n\langle a_L(x),\hat v_n\rangle
   +C_L\gamma_n^2\|\hat v_n\|_*^2.
\end{align*}
All constants are uniform by compactness and uniform positive definiteness of the metric. Equivalence of norms in finite dimensions is absorbed into these constants.
\end{proof}

\begin{corollary}[Lyapunov criterion for RRM convergence]
\label{cor:Convergence-result}
Suppose that $L$ satisfies Definition~\ref{def:riemannian-Lyapunov} and is twice continuously differentiable on a neighborhood of $\mathcal X$, and that $h$ satisfies the regularity conditions of Proposition~\ref{prop:proximal-tracking}. Let $(X_n)$ follow \eqref{eq:rrm-greedy}. If either
\begin{enumerate}
\item[(a)] $\gamma_n\to0$, $\sum_n\gamma_n=\infty$, and $U_n+b_n\to0$ almost surely; or
\item[(b)] Assumption~\ref{ass:parameter conditions} holds and $\sum_n\gamma_n^2<\infty$,
\end{enumerate}
then $X_n\to x^*$ almost surely.
\end{corollary}

\begin{proof}
In case (a), Proposition~\ref{prop:proximal-tracking} gives an asymptotic pseudo-trajectory of the projected dynamics. The strict global Lyapunov condition makes $x^*$ globally asymptotically stable on the compact state space. The Lyapunov limit-set principle therefore yields convergence; see \citet[Proposition~3.27]{BenaimHofbauerSorin2005}.

For case (b), Proposition~\ref{prop:one-step-lyapunov} and strict decrease give
\begin{equation}
 \label{eq:stochastic-lyapunov-step}
 L(X_{n+1})\le L(X_n)-\gamma_n w\|X_n-x^*\|^2
 +\Delta M_{n+1}+\gamma_n\varepsilon_n+R_n,
\end{equation}
where $\Delta M_{n+1}=\gamma_n\langle a_L(X_n),U_n\rangle$,
$\varepsilon_n=\sup_{x\in\mathcal X}\|a_L(x)\|\,\|b_n\|_*$, and
$R_n=C_L\gamma_n^2\|\hat v_n\|_*^2$.
The boundedness of $a_L$ and (H2) imply convergence of the martingale $M_n$ almost surely. Moreover, $\varepsilon_n\to0$ almost surely by (H3). Square-summable steps, boundedness of $v$, and (H2)--(H3) imply $\sum_nR_n<\infty$ almost surely. The supplement in section ~\ref{app:stochastic-level-crossing} applies a pathwise level-crossing argument to \eqref{eq:stochastic-lyapunov-step}, proving $L(X_n)\to0$ and hence $X_n\to x^*$.
\end{proof}

\begin{remark}[Scope of the convergence implication]
\label{rem:theorem-status}
Every polynomial SOS candidate has the smoothness required above. Case (a) covers exact and increasing-accuracy gradients with general diminishing steps. Case (b) covers fixed-variance unbiased gradients with $\ell_\gamma>1/2$, and also permits the vanishing biases in (H3). The normal-cone correction is controlled directly in each proximal step, so no persistent-noise tracking assumption is needed. These conclusions require an actual Lyapunov function; the status of a floating-point SOS representation is addressed in Section~\ref{sec:numerical-status}.
\end{remark}

\subsection{Constructive SOS certification}
\label{sec:quadratic-sos}

Section~\ref{sec:convergence} reduces convergence of the RRM process to the existence of a global Lyapunov function for the projected deterministic dynamics. We now turn this existence question into a computational problem.

The basic idea is to search for a polynomial Lyapunov function $L$ and to certify the required inequalities---positivity, decrease along the interior dynamics, and compatibility with the boundary projection---by sum-of-squares (SOS) representations. Two observations make this possible. First, the feasible set $\mathcal X$ is semialgebraic. Second, although the projected game gradient is rational, its denominators are strictly positive on $\mathcal X$ and can therefore be cleared without changing the sign of the Lyapunov derivative.

\subsubsection{SOS certificates for positivity}

For an even integer $d$, let
\begin{equation}
	\Sigma_d = \left\{ \sigma\in\mathbb R[x]_{\le d}: \sigma(x)=\sum_{k=1}^N q_k(x)^2 \text{ for some } q_k\in\mathbb R[x]_{\le d/2} \right\}, \qquad \Sigma=\bigcup_d\Sigma_d .
	\label{eq:sos-cone}
\end{equation}
Every SOS polynomial is globally nonnegative, so membership $f\in\Sigma$ certifies $f\ge0$ on all of $\mathbb R^n$. This is more than we need, and in fact too much: a polynomial that is nonnegative on the feasible set is typically negative elsewhere, hence never SOS. To certify nonnegativity only on $\mathcal X = \{x:g_i(x)\ge0,\ i=1,\ldots,m+1\}$, we therefore use weighted SOS representations, in which the constraints themselves are allowed to enter the certificate. If
\begin{equation}
	f(x) = \sigma_0(x) + \sum_{i=1}^{m+1}\sigma_i(x)\,g_i(x), \qquad \sigma_0,\sigma_i\in\Sigma,
	\label{eq:weighted-sos}
\end{equation}
then $f(x)\ge0$ for every $x\in\mathcal X$, because all terms on the right-hand side are nonnegative there: $\sigma_0$ and each $\sigma_i$ are nonnegative everywhere, and each $g_i$ is nonnegative on $\mathcal X$. Polynomials admitting such a decomposition form the quadratic module generated by $g_1,\ldots,g_{m+1}$. Bounding the degrees of the products $\sigma_i g_i$ makes \eqref{eq:weighted-sos} an SDP feasibility problem in the Gram matrices of the $\sigma_i$, and raising that bound enlarges the cone, which is the mechanism behind the hierarchy used below. This is the Putinar-type positivity certificate; see \citet{Putinar1993}.

Thus, instead of checking a polynomial inequality at infinitely many points of $\mathcal X$, we search for an algebraic SOS representation that implies the inequality globally on $\mathcal X$.

\subsubsection{Polynomializing the learning dynamics}\label{sec:polynomializing}

The game gradient $v(x)$ is rational. By Proposition~\ref{prop:game-gradient}, its denominators divide
\begin{equation}
	p(x)=\prod_{k=1}^m x_k .
	\label{eq:common-denominator}
\end{equation}
Since $x_k\ge\delta>0$ on $\mathcal X$, we have $p(x)>0$ throughout the feasible region.

Define the polynomialized vector field
\begin{equation}
	\widetilde w(x) = p(x)H(x)^{-1}v(x).
	\label{eq:polynomialized-field}
\end{equation}
Multiplication by the strictly positive factor $p(x)$ changes the speed of the interior dynamics but not the sign of the Lyapunov derivative. Hence, for every $x\in\mathcal X$,
\begin{equation}
	\left\langle\nabla L(x),H(x)^{-1}v(x)\right\rangle<0 \quad\Longleftrightarrow\quad \left\langle\nabla L(x),\widetilde w(x)\right\rangle<0.
	\label{eq:sign-equivalence}
\end{equation}
This reduces the interior Lyapunov condition to a polynomial positivity problem.

\subsubsection{The SOS certification problem}

We now search directly over polynomial Lyapunov functions. Fix an even degree $d_L$, a scale-fixing
constant $\tau\in\mathbb R_{>0}$, and degree bounds for the SOS multipliers. Rather than search over an
unrestricted polynomial and impose the equilibrium conditions $L(x^*)=0$, $\nabla L(x^*)=0$ as separate
constraints, we build them into the choice of function space itself -- and, at the same time, fold the
positivity requirement and the scale normalization into a single matrix inequality rather than a further
pair of constraints.

Fix $K:=d_L/2\in\mathbb N$ and write $y:=x-x^*$. Stack the componentwise (Hadamard) powers of $y$ up to
order $K$ into the vector
\begin{equation}
	\Phi(y) \;:=\; \big(y,\; y^{\odot2},\;\ldots,\;y^{\odot K}\big)\;\in\;\mathbb R^{Km},
	\qquad y^{\odot k} := (y_1^k,\ldots,y_m^k)^\top,
	\label{eq:phi-stack}
\end{equation}
and let $E:=\begin{pmatrix}I_m&0\\0&0\end{pmatrix}\in\mathbb S^{Km}$ denote the projector onto $\Phi$'s
linear-in-$y$ block, so that $\Phi(y)^\top E\,\Phi(y)=\|y\|^2$. Let
\begin{equation}
	V_{d_L}(x^*) \;:=\; \Big\{\, L(x) = \Phi(x-x^*)^\top M\, \Phi(x-x^*) \;:\; M\in\mathbb S^{Km} \,\Big\}
	\;\subset\; \mathbb R[x]_{\le d_L},
\end{equation}
each $L\in V_{d_L}(x^*)$ being represented by a symmetric matrix $M$ -- not uniquely, since the map
$M\mapsto L$ has a nontrivial kernel once $K\ge3$; the program below therefore optimizes over the
representative $M$.

Every entry of $\Phi(y)$ is $y_i^k$ for some $k\ge1$ and so vanishes at $y=0$; every monomial appearing in
$L(x)=\sum_{p,q}M_{pq}\,\Phi_p(y)\Phi_q(y)$ is a product of two such entries, of total degree at least $2$.
Consequently every $L\in V_{d_L}(x^*)$ satisfies
\[
	L(x^*)=0, \qquad \nabla L(x^*)=0,
\]
identically in $M$ -- the equilibrium condition required by Definition~\ref{def:riemannian-Lyapunov} again
holds by construction, for every choice of coefficients. The matrix $E$ plays the analogous role for the
remaining conditions: since $\Phi^\top E\,\Phi=\|y\|^2$, the lower bound $\tau E\preceq M$ gives
$L(x)\ge\tau\|x-x^*\|^2$, so positive definiteness holds by construction as well, while the upper bound
$M\preceq I_{Km}$ gives $L(x)\le\|\Phi(x-x^*)\|^2$ and fixes the scale that the homogeneity of the
remaining constraints would otherwise leave free. Both are delivered by the single inequality
\eqref{eq:r-scale}.

For $m\ge2$ and $K\ge2$, $V_{d_L}(x^*)$ is a strict subspace of the degree-$d_L$ polynomials vanishing to
second order at $x^*$: quadratic forms in $\Phi$ reach only the monomials $y_i^{k_1}y_j^{k_2}$ with
$k_1,k_2\le K$, so both the number of distinct coordinates and each individual exponent are capped
(already $y_1^3y_2$ is unreachable at $K=2$). The construction trades this representational flexibility
for the two simplifications above. The following program imposes the remaining requirements of
Definition~\ref{def:riemannian-Lyapunov}.

\begin{problem}[Rescaled SOS Lyapunov certificate]
	\label{prob:SDP-rescaled}
	Let $A=H(x^*)^{-1}Dv(x^*)$ and define its spectral abscissa by $\alpha(A):=\max\{\operatorname{Re}\lambda:\lambda\in\sigma(A)\}$. Fix a positive rate scale $\rho>0$, a scale-fixing constant $\tau\in\mathbb R_{>0}$, and multiplier degree bounds $d_{\mathrm{dec}},d_{\mathrm{bnd}}\in\mathbb N$. In the experiments, after numerically checking $\alpha(A)<0$, we set $\rho=-\alpha(A)$. Decide feasibility over $M\in\mathbb S^{Km}$ with $L(x)=\Phi(x-x^*)^\top M\,\Phi(x-x^*)$,
	SOS multipliers $\chi_i\in\Sigma_{d_{\mathrm{dec}}}$, $\eta_{i,j}\in\Sigma_{d_{\mathrm{bnd}}}$,
	and polynomial multipliers $\varphi_i\in\mathbb R[x]_{d_{\mathrm{bnd}}}$, subject to
	\begin{itemize}
		\item Scale fixing:
		\begin{align}
			\tau E \preceq M \preceq I_{Km}
			\label{eq:r-scale}
		\end{align}
		\item Strict interior decrease against the rate-normalized field:
		\begin{align}
			-\tfrac1\rho\big\langle \nabla L(x),\widetilde w(x)\big\rangle
			- p(x)\,L(x)
			- \sum_{i=1}^{m+1}\chi_i(x)\,g_i(x)\ \in\ \Sigma
			\label{eq:r-decrease}
		\end{align}
		\item Boundary compatibility:
		\begin{align}
			&-\big\langle \nabla L(x),\,H(x)^{-1}\nabla g_i(x)\big\rangle
			- \varphi_i(x)\,g_i(x) \nonumber\\
			&\qquad - \sum_{\substack{j=1\\ j\neq i}}^{m+1}\eta_{i,j}(x)\,g_j(x)
			\in\ \Sigma, \qquad i\in[m+1],
			\label{eq:r-boundary}
		\end{align}
	\end{itemize}
	A certificate is returned iff the program is feasible.
\end{problem}
\begin{remark}
	The constraints of Problem~\ref{prob:SDP-rescaled} have a direct interpretation.

	First, since $\tau E \preceq M \preceq I_{Km}$ \eqref{eq:r-scale} are properties of the
	matrices themselves, they hold when tested against $\Phi(y)$ for \emph{every} $y\in\mathbb R^m$,
	so for every $x\in\mathbb R^m$
	\begin{equation}
		\tau\,\|x-x^*\|^2 \;\le\; L(x) \;\le\; \|\Phi(x-x^*)\|^2 .
		\label{eq:r-positive-implication}
	\end{equation}
	In particular $L(x)\ge\tau\|x-x^*\|^2$ for all $x\in\mathcal X$, so $L$ is positive away from
	$x^*$; unlike the SOS argument used below for \eqref{eq:r-decrease}--\eqref{eq:r-boundary},
	\eqref{eq:r-positive-implication} holds globally and needs no reference to $\mathcal X$ or to
	$g_1,\ldots,g_{m+1}$. Here \eqref{eq:r-scale} only fixes the scale of the representative $L$: it
	no longer carries a decrease margin, and the upper bound $M\preceq I_{Km}$ plays no role beyond
	keeping $M$ bounded.

	Second, because $p(x)>0$ on $\mathcal X$ and $\rho>0$, \eqref{eq:r-decrease} implies
	\begin{equation}
		-\big\langle \nabla L(x),H(x)^{-1}v(x) \big\rangle \;\ge\; \rho\,L(x)
		\qquad \text{for all }x\in\mathcal X.
		\label{eq:r-decrease-implication}
	\end{equation}
	Thus $L$ decreases strictly along the interior dynamics for every $x\neq x^*$, at the
	relative rate $\rho$. The experimental choice $\rho=-\alpha(A)$ normalizes the requested rate by the
	local linear contraction scale. Combined with
	\eqref{eq:r-positive-implication} this gives
	$-\langle\nabla L,H^{-1}v\rangle\ge \rho\,\tau\|x-x^*\|^2$, so the decrease is strict off $x^*$
	with a margin that vanishes only as the instance itself approaches marginality
	($-\alpha(A)\to0$).

	Finally, on the active face $g_i(x)=0$, \eqref{eq:r-boundary} gives
	\begin{equation}
		\big\langle \nabla L(x), H(x)^{-1}\nabla g_i(x) \big\rangle \;\le\; 0.
		\label{eq:r-boundary-implication}
	\end{equation}
	Hence the Riemannian projection correction cannot increase the Lyapunov derivative. This condition
	is homogeneous in $M$ and non-strict, so it carries no scale and needs no margin; equality
	(the field tangent to the face) is admitted, as required by
	Definition~\ref{def:riemannian-Lyapunov}.
\end{remark}
We therefore obtain the following implication.

\begin{proposition}[Certificate implication]
	\label{prop:sos-certificate}
	If Problem~\ref{prob:SDP-rescaled} admits an exact feasible solution, then $L$ is a global Lyapunov function for $x^*$ under $\mathcal D(\mathcal X,h,v)$. Consequently, under the hypotheses of Corollary~\ref{cor:Convergence-result}, the corresponding RRM iterates converge almost surely to $x^*$.
\end{proposition}

\begin{proof}
	Since $L\in V_{d_L}(x^*)$, the equilibrium conditions $L(x^*)=0$ and $\nabla L(x^*)=0$ hold by construction. Condition~\eqref{eq:r-scale} gives positive definiteness of $L$ away from $x^*$, and \eqref{eq:r-decrease} gives strict decrease along the unconstrained interior field. On every active boundary face, \eqref{eq:r-boundary} ensures that the Riemannian projection correction cannot increase the Lyapunov derivative. Hence $L$ satisfies Definition~\ref{def:riemannian-Lyapunov}. The convergence statement follows from Corollary~\ref{cor:Convergence-result}.
\end{proof}

\subsubsection{Numerical certification}
\label{sec:numerical-status}

Problem~\ref{prob:SDP-rescaled} is solved in floating-point arithmetic. Its output consists of a Lyapunov matrix $M$, SOS Gram matrices, and polynomial multipliers. We distinguish the solver's feasible termination status from the reconstructed numerical checks used to assess that representation. The instance table records solver-feasible outputs throughout the experimental design. Numerical certification evaluates their polynomial identities and semidefinite constraints over the complete algebraic representation, up to explicitly specified numerical tolerances.

\paragraph{Software and solver status.}
The programs are modelled with PICOS~2.6.2 and
\textsc{SumOfSquares}~1.3.1, with symbolic construction in
SymPy~1.14.0 and linear algebra in NumPy~2.3.5, and solved
with MOSEK~11.2.0 under Python~3.13.8 on Windows~11,
Intel64 Family~6 Model~142, eight logical cores.
Problem~\ref{prob:SDP-rescaled} is posed as a pure feasibility
problem: no performance objective is optimized.
MOSEK is called at its default interior-point tolerances
with a configured wall-clock time limit of $900$\,s per instance.
A run is \emph{solver-feasible} when the modelling layer reports
primal status \texttt{optimal} or \texttt{feasible}; all $243$
reported systems returned \texttt{optimal}.
An infeasibility outcome concerns the fixed template
$(d_L,d_{\mathrm{dec}},d_{\mathrm{bnd}},\tau,c)$ and does not
establish instability of the dynamics.
Every other termination, including a time-limit cutoff without
an accepted feasible solution, is classified as inconclusive.
All code, archived solver outputs, and the scripts that regenerate every table
and figure in this section are openly available \citep{hoehener2026bertrand}.
The archive records one row per reported configuration, together with a driver
that replays any of them from its stored degrees and constants.

\paragraph{Numerical analysis.}
The returned polynomials are treated as numerical Lyapunov
candidates.
The computational study documents the outcomes of the SOS search
across the experimental design and provides numerical evidence
for feasibility of the implemented constraint systems.
The theoretical results establish global asymptotic stability
when the full positive-rate Lyapunov conditions are satisfied;
solver termination alone is not used to assert that conclusion
for an individual candidate.
The six zero-rate outputs are considered separately below.

\paragraph{Parameter choices and small decrease rates.}
The design fixes $\delta=1/250$, $\tau=10^{-4}$,
the cap $M\preceq I$, the decrease scaling $c=1$,
and the boundary offset $\gamma=0$.
The requested rate is obtained by approximating
$\lvert\alpha(A)\rvert$ by a nearby rational with denominator
at most $10^6$.
This approximation produces $\rho=0$ in $6$ of the $243$
systems: $(\ell,n,s,m)=(3,4,1,4)$ and $(3,4,2,4)$
under all three geometries.
Their computed spectral abscissae remain negative, with
magnitudes between $1.5\times10^{-7}$ and $4.0\times10^{-7}$.
The zero requested rates therefore reflect a loss of
resolution in the rate approximation and provide no evidence
of instability of the equilibrium.
The computed linearizations support local attraction.

These six systems are reported separately from those with
a positive requested rate.
Since \eqref{eq:r-decrease} is normalized by $\rho$, it applies
only when $\rho>0$.
At $\rho=0$, numerical assessment must instead use the
unnormalized non-increase condition
$\langle\nabla L,\widetilde w\rangle\le0$.
Even when satisfied, this condition establishes non-increase
of $L$ alone and does not imply global asymptotic convergence.
A further $15$ systems have $0<\rho<10^{-5}$, making the
assessment of numerical errors relative to the decrease
margins particularly relevant.

\subsection{Experimental design and results}
\label{sec:certified-scope}
\paragraph{Experimental design}
The computational analysis varies four primitives in a full factorial design: the number of firms $n\in\{2,3,4\}$, the cost-prior parameter $\ell\in\{1,2,3\}$, the demand exponent $s\in\{0,1,2\}$, and the projection dimension $m\in\{2,3,4\}$. This yields $81$ distinct projected Bertrand games. Each game is studied under three learning geometries, $q \in \{0,1,2\}$---Burg, entropic, and Euclidean---resulting in $243$ projected dynamical systems.

The lower-slope bound is fixed at $\delta=1/250$ throughout. Theorem~\ref{thm:no-minty} alone guarantees projected Minty failure only for sufficiently small $\delta$, so we do not classify the experimental instances from $m\ge m^\star$ without an additional check. For every cell with $m\ge m^\star$, we evaluate the theorem's explicit piecewise-linear witness at $\delta=1/250$ using rational arithmetic. All $37$ such cells satisfy
\[
	(x(\delta)-x^*)^\top v(x(\delta))>0,
\]
and therefore violate the projected Minty condition at the experimental bound. The smallest value occurs at $(\ell,n,s,m)=(2,2,2,3)$ and equals
\[
	\frac{9678291936089123}{1050044310000000000}\approx9.217\times10^{-3}.
\]
For cells marked inconclusive, no claim that Minty holds is made.

We then search Problem~\ref{prob:SDP-rescaled}, using Lyapunov degrees $d_L\in\{2,4,6,8,10\}$, decomposition degrees $d_{\mathrm{dec}}\in\{4,6,8,10,12,14\}$, and boundary degrees $d_{\mathrm{bnd}}\in\{0,2,4,6,8,10\}$. The reported $d_L^\star$ is the smallest Lyapunov degree at which this implemented search returned a solver-feasible representation. It is a search outcome, not a proof of the minimum possible certificate degree. The archive must specify the search order, termination criteria, and results of unsuccessful attempts.

\begin{remark}[Cost of the relaxation and the $m\le4$ design]
\label{rem:sdp-cost}
The multiplier Gram matrices dominate the size of Problem~\ref{prob:SDP-rescaled}. There are $\Theta(m^2)$ semidefinite blocks, because every boundary face has a cross-multiplier for each other constraint. A degree-$2d$ SOS multiplier in $m$ variables has a Gram matrix of order $\binom{m+d}{d}$. The number and size of these blocks therefore grow rapidly with approximation dimension. We use $m\le4$ for the full factorial design; all 243 searches succeed in that range, and the three five-segment probes also yield feasible representations. These outcomes identify the tested computational range, while larger representations require further searches and potentially more scalable positivity methods.
\end{remark}

An overview of the experimental design can be found in Table \ref{tab:stacked-certificates-scope}

\begin{table}[!htbp]
	\centering
	\caption{Scope of the computational analysis.}
	\label{tab:stacked-certificates-scope}
	\begin{tabular}{ll}
		\toprule
		\textbf{Quantity} & \textbf{Values considered} \\
		\midrule
		\multicolumn{2}{l}{\emph{Held fixed}} \\
		Lower-slope bound & $\delta=1/250$ \\
		Positivity scale & $\tau = 10^{-4}$ \\
		Dynamics & symmetric reduced flow, cf.\ Remark~\ref{rem:symmetric-scope} \\
		\midrule
		\multicolumn{2}{l}{\emph{Varied across instances} ($81$ games, $243$ systems)} \\
		Number of firms $n$ & $2,3,4$ \\
		Cost prior & Beta$(1,\ell)$, $\ell\in\{1,2,3\}$ \\
		Demand exponent $s$ & $0$ (all-or-nothing), $1$ (linear), $2$ (quadratic) \\
		Projection dimension $m$ & $2,3,4$ \\
		Learning geometry & Burg ($q=0$), entropic ($q=1$), Euclidean ($q=2$) \\
		\midrule
		\multicolumn{2}{l}{\emph{Stacked certificate search}} \\
		Lyapunov degree $d_L$ & $2,4,6,8,10$ \\
		Decomposition degree $d_{\mathrm{dec}}$ & $4,6,8,10,12,14$ \\
		Boundary degree $d_{\mathrm{bnd}}$ & $0,2,4,6,8,10$ \\
		\bottomrule
	\end{tabular}
\end{table}
\paragraph{Experimental results}
For each projected dynamical system, the search records the smallest
Lyapunov degree $d_L^\star$ at which it returns a solver-feasible
SOS representation within the tested degree ranges.
Figure~\ref{fig:stacked-certificates-landscape} summarizes these
search outcomes for the $243$ systems.
The displayed degrees describe returned candidates; their
numerical-certificate status is determined by the checks in
Section~\ref{sec:numerical-status}.
Instance-level degrees, problem sizes, and solver times are
reported in Table~\ref{tab:stacked-certificate-details}
of the supplement.
Figure~\ref{fig:lyapunov} illustrates the projected dynamics and
candidate Lyapunov level sets for $\ell=1$, $n=2$, and $m=2$
under the three demand specifications and learning geometries.

\textit{First, the search returns solver-feasible SOS representations}
throughout the main design, including every system with an explicit
Minty violation.
The witness establishes Minty failure in $37$ of the $81$
projected games, corresponding to $111$ dynamical systems across
the three geometries.
The mean candidate Lyapunov degree is $2.38$ both within and
outside this set, while the descriptive degree-budget proxy
$d_L^\star+d_{\mathrm{dec}}+d_{\mathrm{bnd}}$ averages
$11.6$ and $11.8$, respectively.
Thus, in the tested design, explicit Minty failure is compatible
with successful SOS searches and is not associated with larger
average degree budgets. Of the $243$ outputs, $237$ have a positive requested rate and
six have a rate rounded to zero, as detailed in
Section~\ref{sec:numerical-status}.
Three additional probes at $m=5$, one per geometry, also return solver-feasible degree-two candidates.

\begin{figure}[!htbp]
	\centering
	\includegraphics[width=0.98\linewidth]{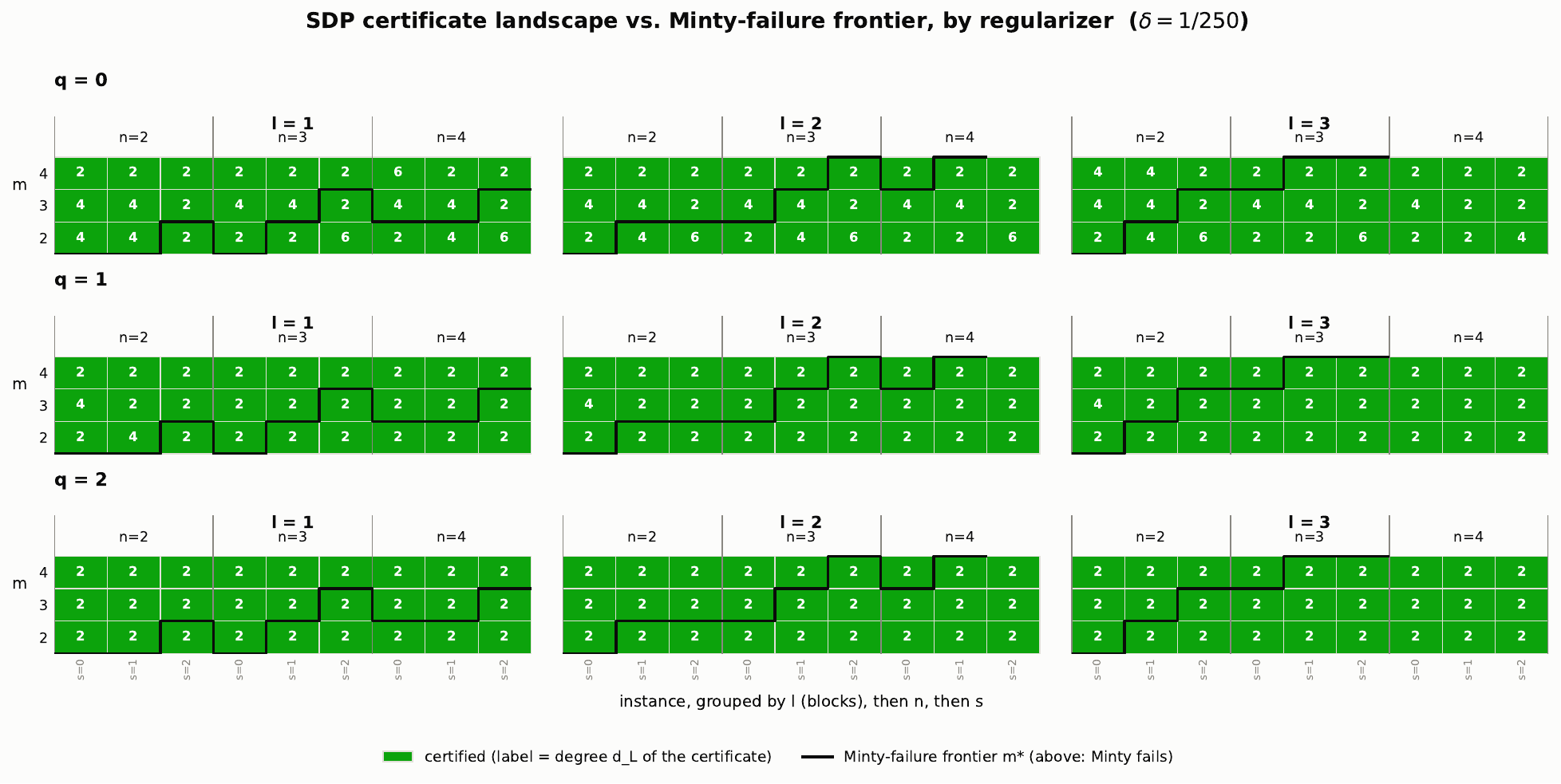}
	\caption{SOS search outcomes against the sufficient Minty-witness threshold, by learning geometry.
Rows within each panel index the projection dimension $m$;
columns index the $27$ games $(\ell,n,s)$.
Each cell displays the degree $d_L^\star$ of a solver-feasible
candidate.
The six cells with $(\ell,n,s,m)=(3,4,1,4)$ or $(3,4,2,4)$
across the three geometries have a zero requested rate and are
assessed separately.
The black line marks the sufficient Minty-witness threshold
$m^\star$.
For every displayed instance with $m\ge m^\star$, the explicit
witness establishes Minty failure at $\delta=1/250$.
No conclusion about Minty validity is drawn below this threshold.}
	\label{fig:stacked-certificates-landscape}
\end{figure}

\textit{Second, the regularizer substantially shapes certificate complexity.} The three panels of Figure~\ref{fig:stacked-certificates-landscape} use identical games and differ only in the Riemannian metric. Under Euclidean geometry, constant boundary multipliers ($d_{\mathrm{bnd}}=0$) suffice in all $81$ systems; this does not remove the boundary conditions, but it shows that higher-degree face multipliers are unnecessary in these numerical representations. The mean boundary degrees are $2.5$ for entropic and $3.8$ for Burg dynamics. Every Euclidean system has a degree-two candidate; $77$ of $81$ entropic candidates are quadratic and four quartic; Burg yields $47$ quadratic, $26$ quartic, and eight degree-six candidates. The larger degrees under Burg are consistent with stronger state dependence near the lower-slope faces, but the experiment does not establish that geometry as a causal mechanism or rule out lower-degree certificates in a different basis.

\begin{figure}[!htbp]
    \centering
    \begin{subfigure}[b]{0.32\textwidth}
        \centering
        \includegraphics[width=\textwidth]{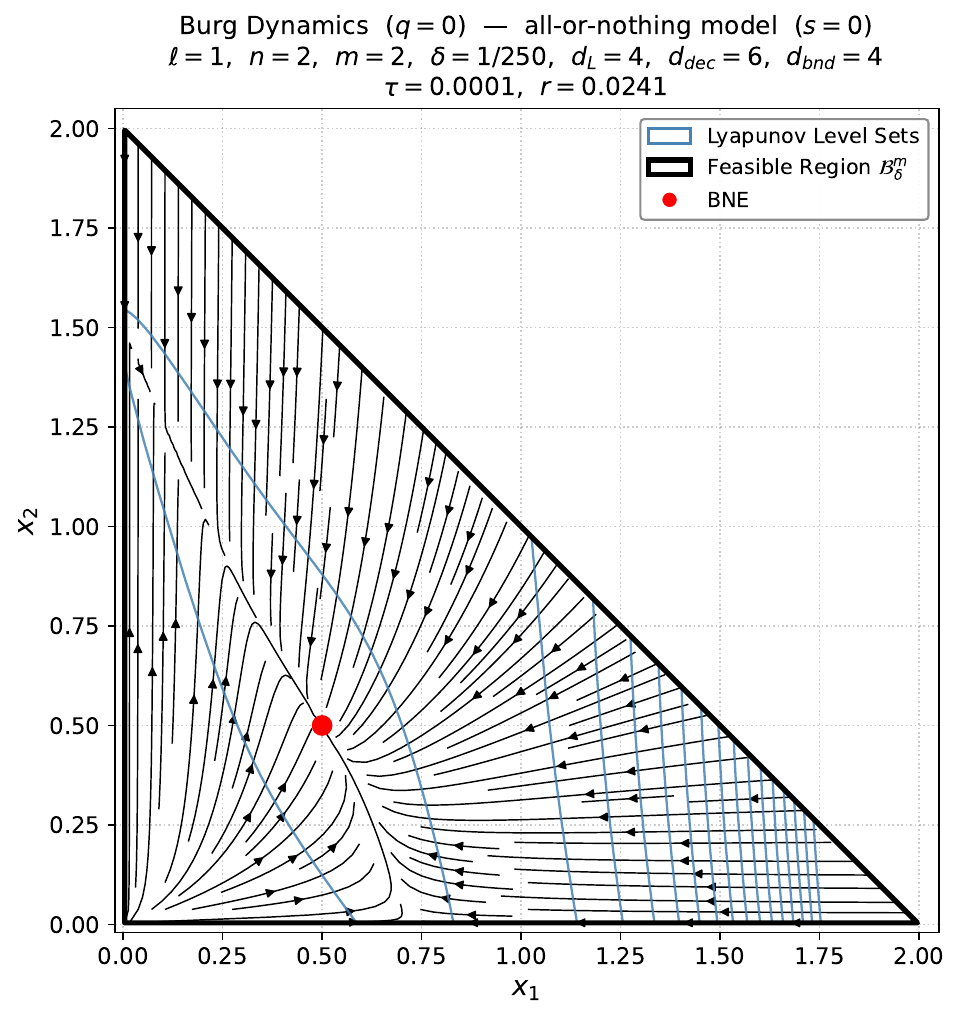}
        \caption{AON ($s=0$), Burg ($q=0$)}
        \label{fig:s0_q0}
    \end{subfigure}\hfill
    \begin{subfigure}[b]{0.32\textwidth}
        \centering
        \includegraphics[width=\textwidth]{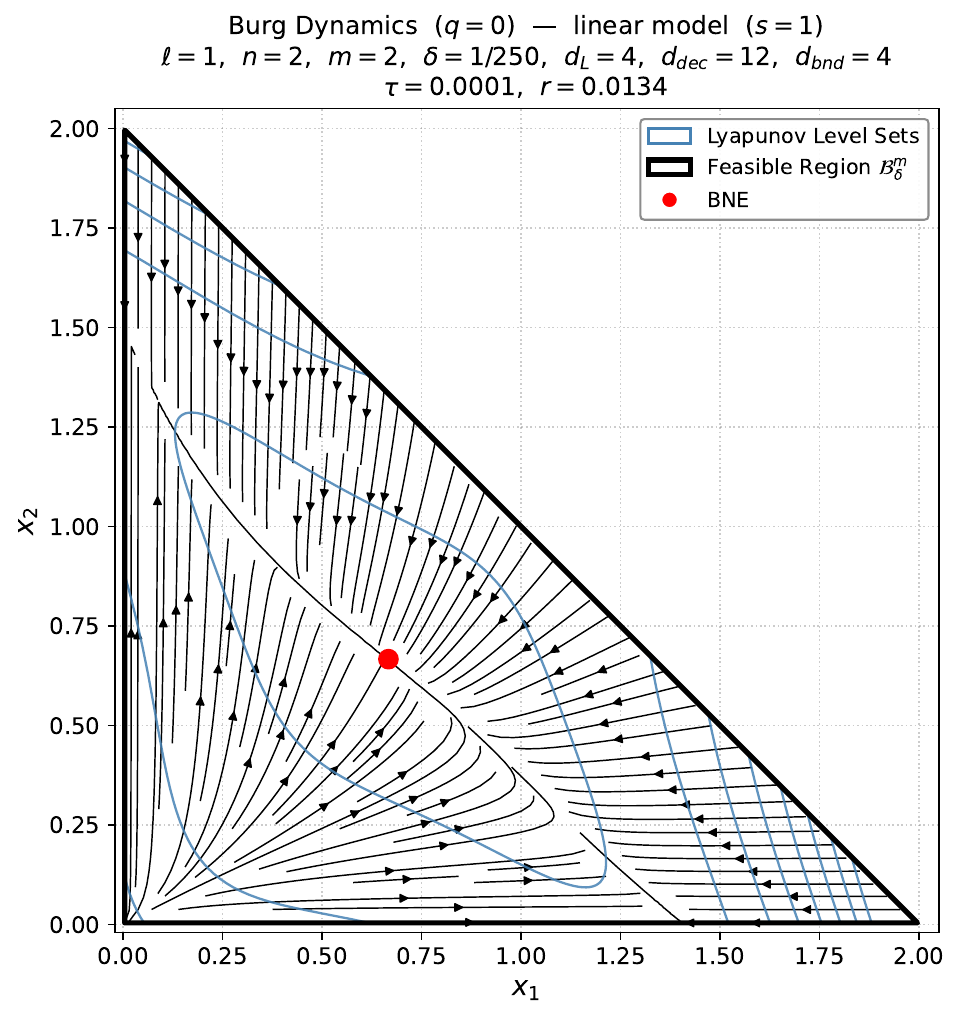}
        \caption{Linear ($s=1$), Burg ($q=0$)}
        \label{fig:s1_q0}
    \end{subfigure}\hfill
    \begin{subfigure}[b]{0.32\textwidth}
        \centering
        \includegraphics[width=\textwidth]{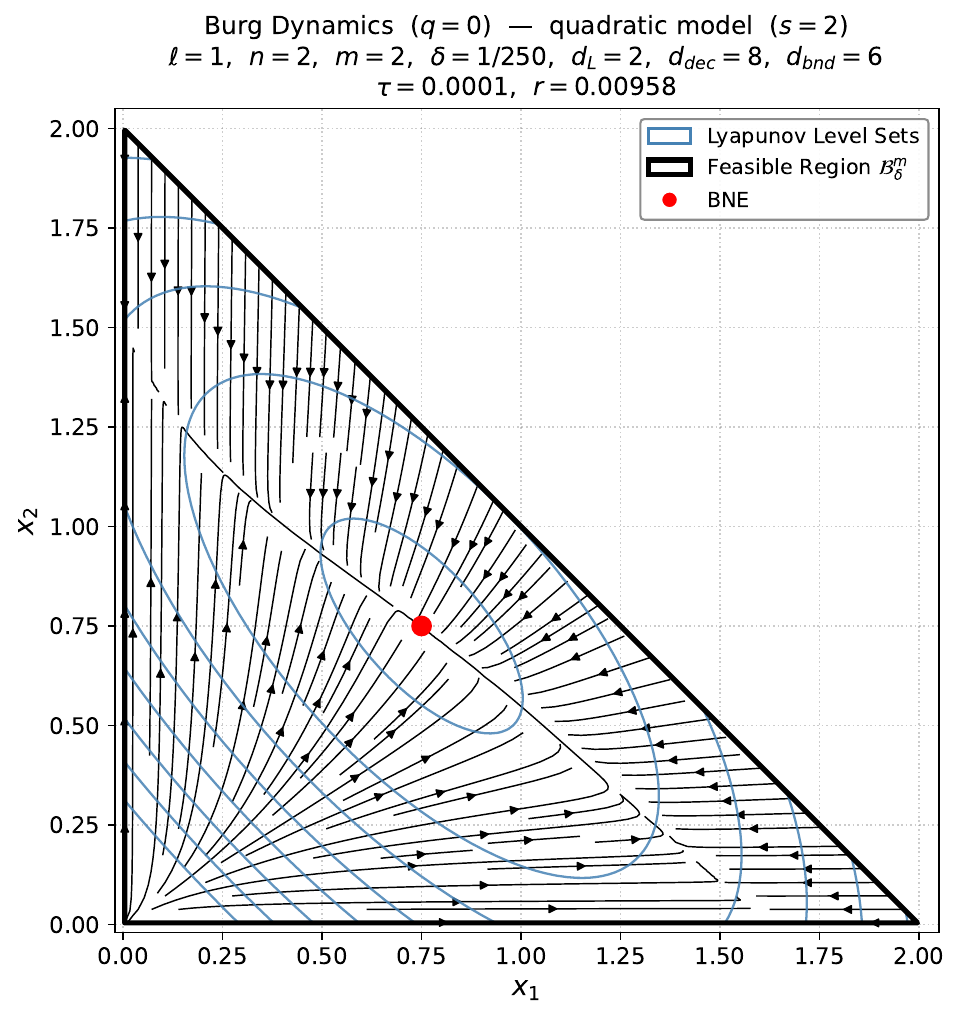}
        \caption{Quadratic ($s=2$), Burg ($q=0$)}
        \label{fig:s2_q0}
    \end{subfigure}
    \vspace{0.5cm}
    \begin{subfigure}[b]{0.32\textwidth}
        \centering
        \includegraphics[width=\textwidth]{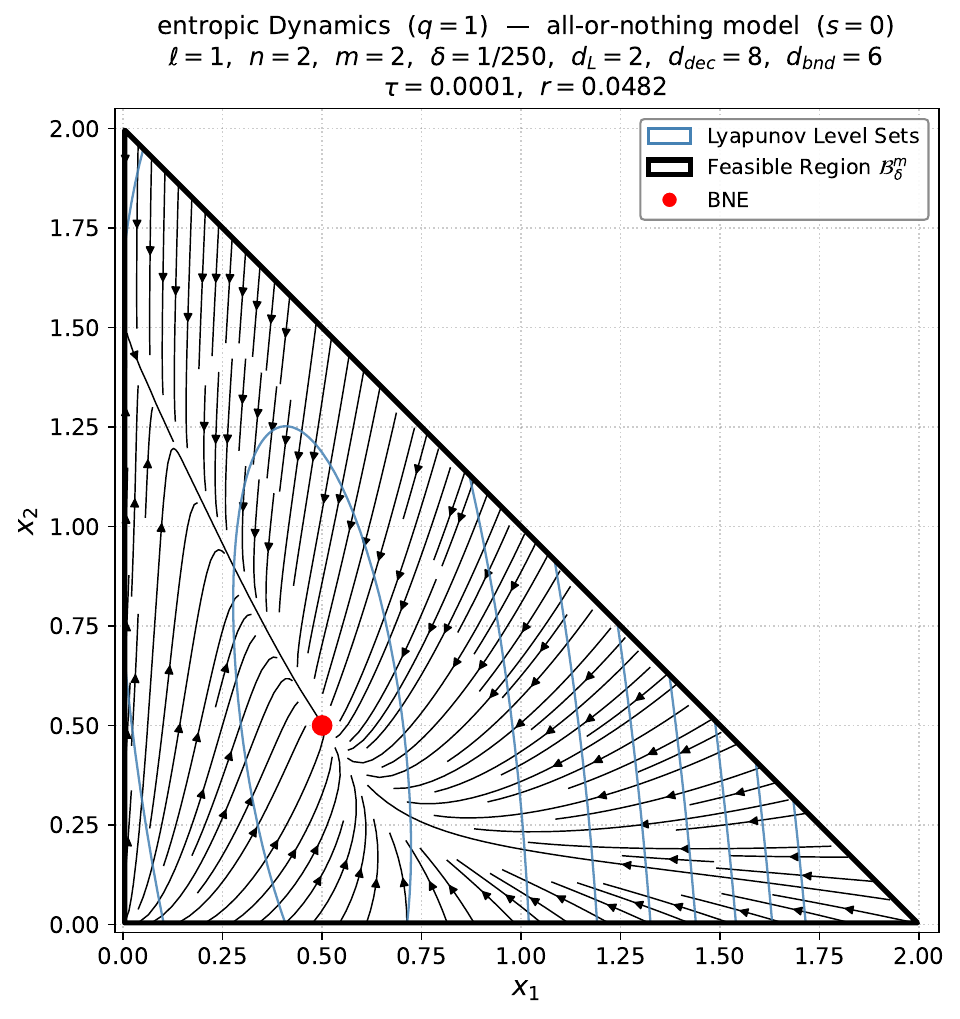}
        \caption{AON ($s=0$), Entropic ($q=1$)}
        \label{fig:s0_q1}
    \end{subfigure}\hfill
    \begin{subfigure}[b]{0.32\textwidth}
        \centering
        \includegraphics[width=\textwidth]{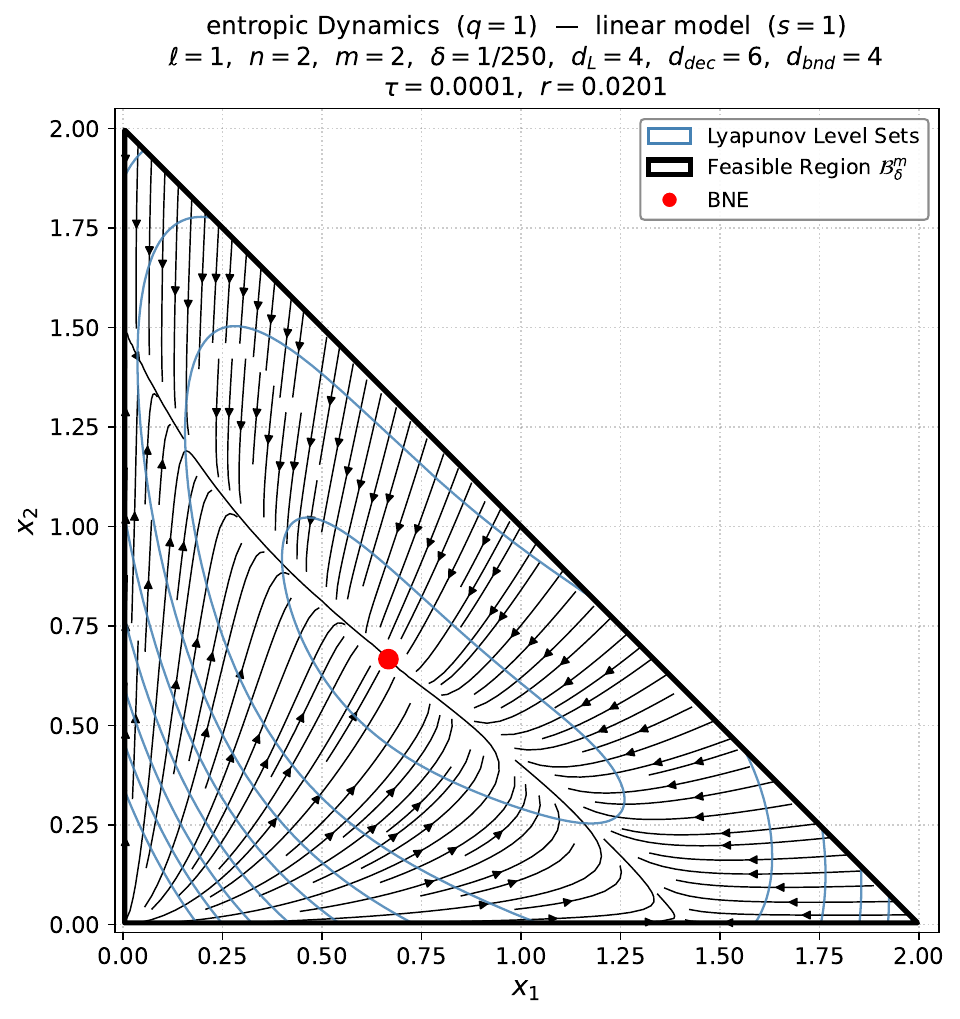}
        \caption{Linear ($s=1$), Entropic ($q=1$)}
        \label{fig:s1_q1}
    \end{subfigure}\hfill
    \begin{subfigure}[b]{0.32\textwidth}
        \centering
        \includegraphics[width=\textwidth]{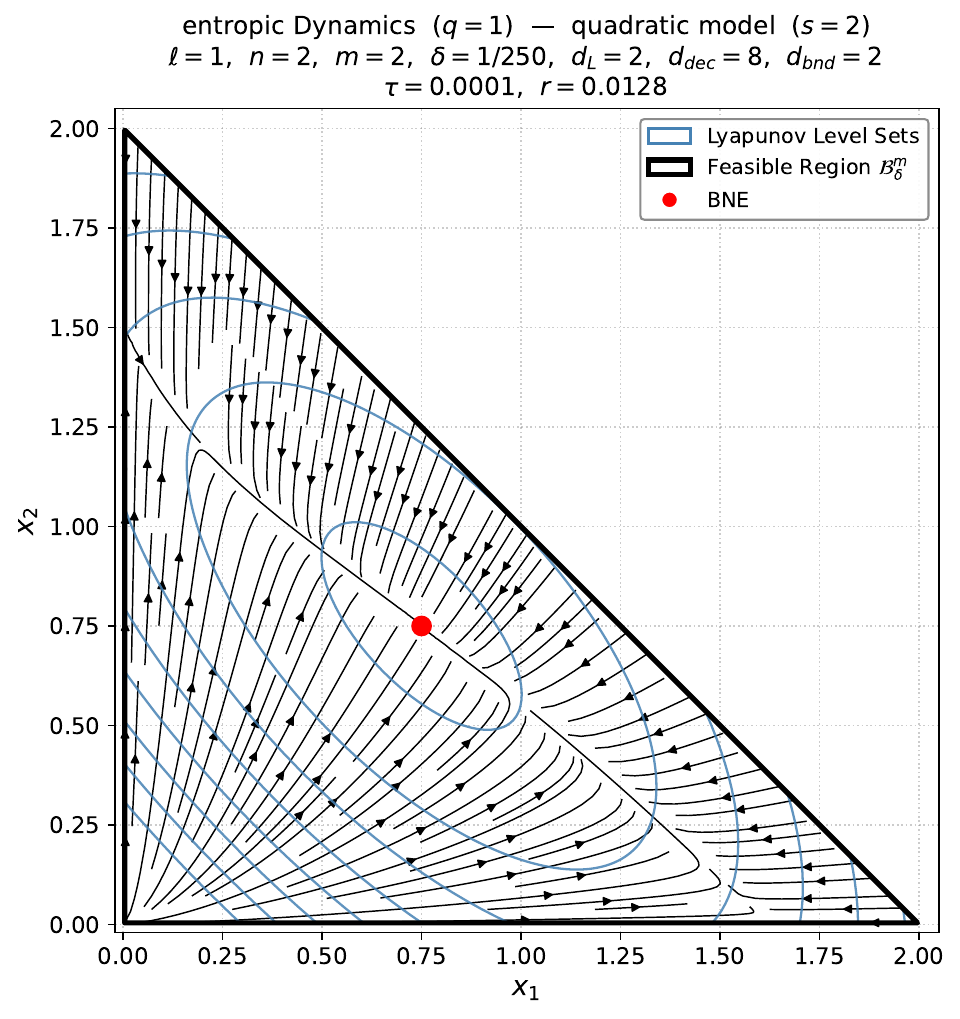}
        \caption{Quadratic ($s=2$), Entropic ($q=1$)}
        \label{fig:s2_q1}
    \end{subfigure}
    \vspace{0.5cm}
    \begin{subfigure}[b]{0.32\textwidth}
        \centering
        \includegraphics[width=\textwidth]{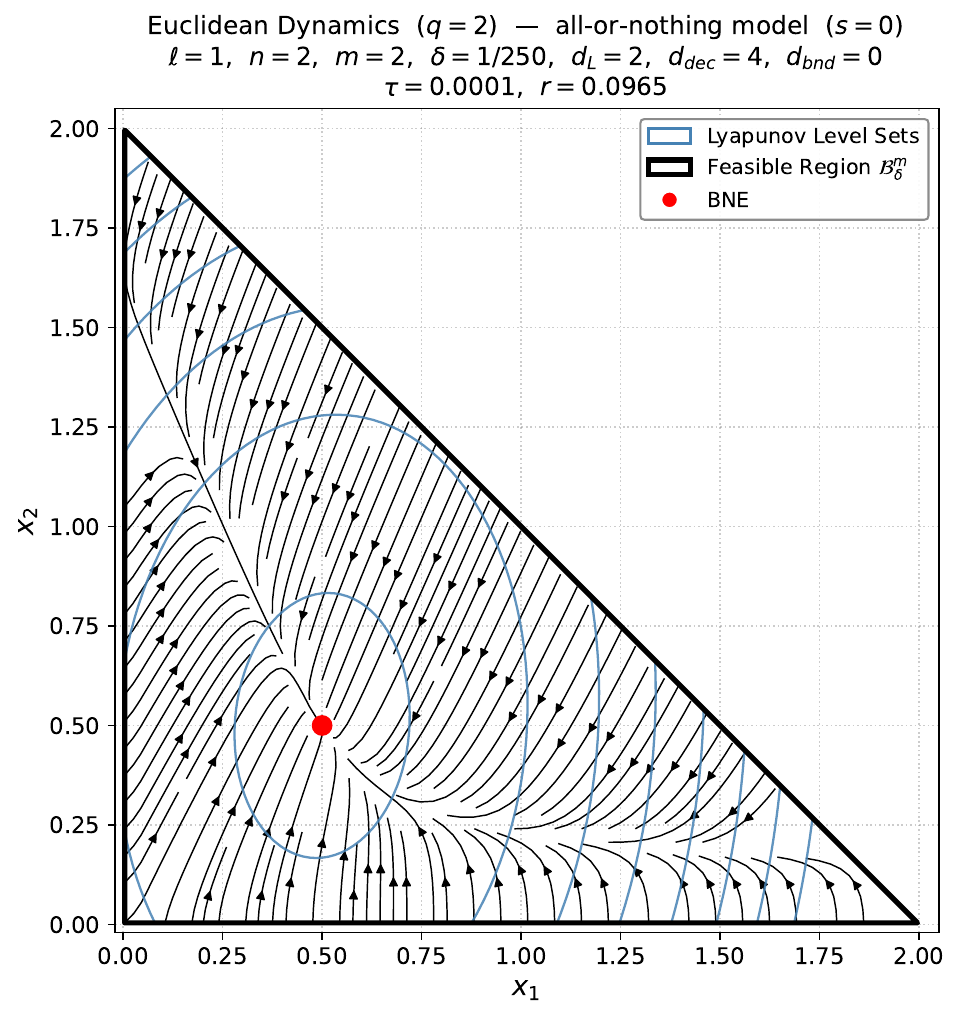}
        \caption{AON ($s=0$), Euclidean ($q=2$)}
        \label{fig:s0_q2}
    \end{subfigure}\hfill
    \begin{subfigure}[b]{0.32\textwidth}
        \centering
        \includegraphics[width=\textwidth]{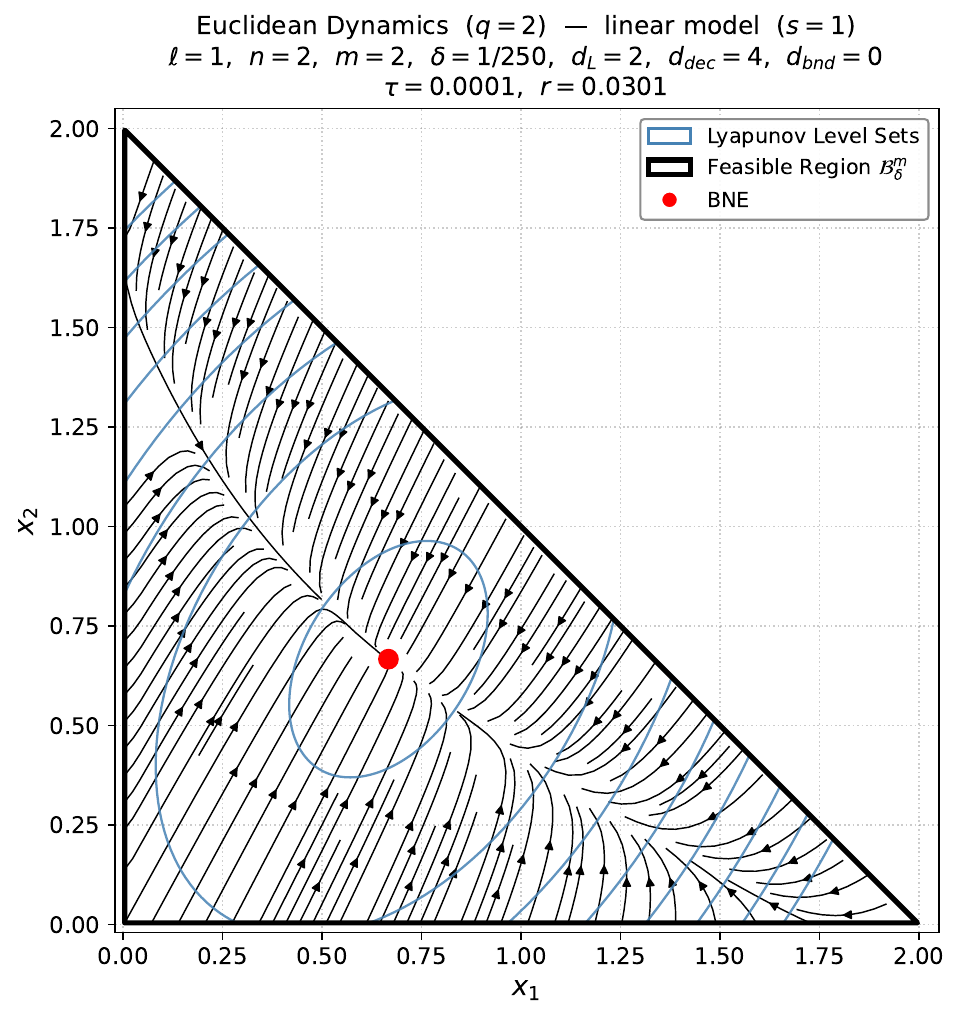}
        \caption{Linear ($s=1$), Euclidean ($q=2$)}
        \label{fig:s1_q2}
    \end{subfigure}\hfill
    \begin{subfigure}[b]{0.32\textwidth}
        \centering
        \includegraphics[width=\textwidth]{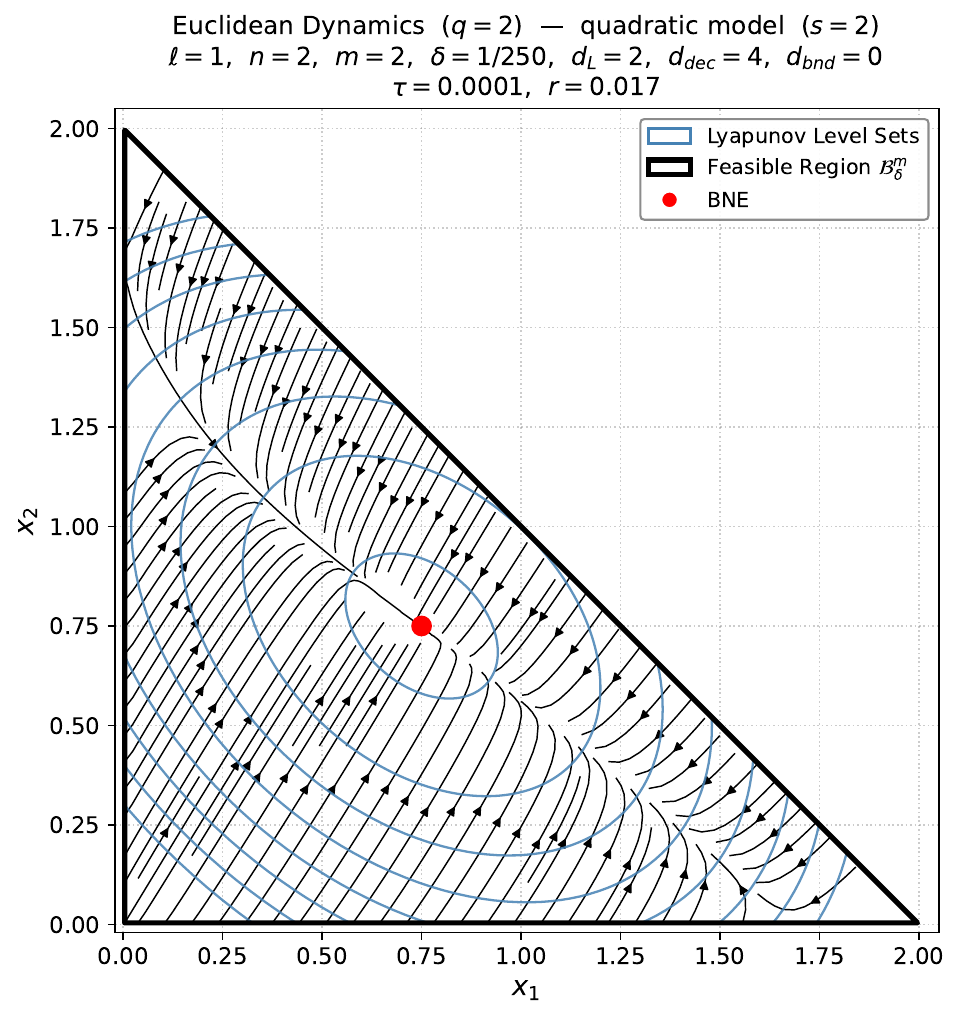}
        \caption{Quadratic ($s=2$), Euclidean ($q=2$)}
        \label{fig:s2_q2}
    \end{subfigure}
    \caption{
        Phase portraits and Lyapunov level sets for the projected
        two-dimensional dynamics ($m=2$, $n=2$, $\ell=1$, $\delta=1/250$).
        Rows correspond to Burg ($q=0$), entropic ($q=1$), and Euclidean ($q=2$);
        columns correspond to all-or-nothing ($s=0$), linear ($s=1$), and
        quadratic ($s=2$) demand.
        The black polygon is $\mathcal{B}_\delta^2$, arrows show the Riemannian
        vector field, the red marker is the symmetric Bayesian Nash equilibrium
        $x^*$, and blue contours are level sets of $L$. Certificate degrees per
        panel are given in the supplement in Table~\ref{tab:stacked-certificate-details}.
    }
    \label{fig:lyapunov}
\end{figure}
\FloatBarrier

\textit{Third, demand curvature is associated with higher multiplier degrees under Burg geometry.} The mean decomposition degree rises from $7.3$ at $s=0$ to $8.8$ at $s=1$ and $10.6$ at $s=2$; at $m=2$ the corresponding means are $6.0$, $9.3$, and $12.9$. The pattern is descriptive rather than causal: changing $s$ changes the entire polynomial vector field, and the design does not isolate one mechanism.

Seven of the eight degree-six Burg candidates occur at $m=2$ with quadratic demand and use the largest decomposition degree, $d_{\mathrm{dec}}=14$. Under the entropic metric the mean decomposition-degree sweep is comparatively flat ($7.1$, $6.8$, $7.4$), and under the Euclidean metric it decreases mildly ($6.1$, $5.8$, $5.3$). The remaining primitives show no comparably systematic pattern within the tested range.

Two further observations concern computational cost. First, the Lyapunov degree alone understates complexity: $94$ of the $205$ degree-two candidates require $d_{\mathrm{dec}}\ge8$. Under the Euclidean geometry, the decomposition degree rises with projection dimension even though $d_L^\star=2$ throughout. Second, the SDPs range from $199$ to $181\,800$ scalar variables, and a descriptive log--log fit relates runtime to problem size with exponent $1.11$ and correlation $0.98$. Solver times of the reported successful runs sum to approximately $3.70$ hours, with Burg systems accounting for $71\%$. This total concerns the accepted runs in Table~\ref{tab:stacked-certificate-details}; it excludes unreported unsuccessful attempts, model construction, and post-solve verification. These figures explain the $m\le4$ ceiling of the factorial experiment: it is a natural scalability limit of dense SOS/SDP certification, not an observed stability limit. Within that computational range, the search returns feasible representations without exception.

\section{Discussion and Conclusion}
\label{sec:conclusion}

We study the convergence of proximal equilibrium solvers in symmetric Bayesian Bertrand competition. The continuous-strategy game violates Minty in every neighborhood of its affine equilibrium. A constructive resolution threshold shows how sufficiently flexible piecewise-linear approximations inherit this obstruction when the lower-slope bound is small. Explicit rational witnesses confirm the obstruction at the experimental bound in 37 projected games.

The absence of Minty motivates direct analysis of the learning dynamics. We formulate global Lyapunov conditions that incorporate the metric-dependent projection at every boundary face and encode them through SOS semidefinite constraints. The same boundary condition yields a one-step proximal Lyapunov estimate, establishing almost-sure convergence for persistent martingale errors with square-summable steps. A complementary tracking theorem covers exact and vanishing-error gradients with more general diminishing steps. These algorithmic conclusions follow from the existence of a Lyapunov function satisfying the stated inequalities.

The computational study returns solver-feasible SOS representations
for all $243$ systems in the main design and three additional
five-segment probes.
The main design includes $111$ systems with explicitly verified
Minty failure.
Euclidean candidates are quadratic throughout, entropic candidates
are almost always quadratic, and Burg candidates sometimes have
degree four or six.
These results demonstrate that the SOS search remains computationally
effective across the tested games despite the failure of Minty,
and reveal substantial differences in search complexity across
learning geometries.
The six outputs whose requested rate rounds to zero are assessed
separately and do not establish global asymptotic convergence.

The scope is symmetric shared-parameter equilibrium computation with a known affine target and a fixed finite-dimensional approximation. All main instances have two to four segments; the largest SDP has $181\,800$ scalar variables. These results do not establish stability of independent firms' learning processes or of the original function-space dynamics. They also do not establish the minimum possible Lyapunov degree outside the chosen basis and search. The natural extensions are certificates for non-affine and numerically determined equilibria, richer economic environments, and sparse positivity methods that can handle larger strategy representations.

\section*{Data availability}
The code and archived numerical results are openly available at
\url{https://doi.org/10.5281/zenodo.22663920}. No other data were used.

\section*{Declaration of generative AI and AI-assisted technologies in the manuscript preparation process}
During the preparation of this work, the authors used OpenAI ChatGPT to assist with language editing, organization, clarity of exposition, and the development of parts of the computational code. ChatGPT and Anthropic Claude were also used to examine proof arguments and suggest implementation improvements. The authors are responsible for reviewing and verifying all mathematical arguments, code, computations, and written content, and take full responsibility for the article.


\clearpage
\phantomsection
\addcontentsline{toc}{section}{Supplement}
\begin{center}
{\Large\bfseries Supplement}\par\medskip
{\large Stability Beyond Minty}\par\smallskip
Detailed experimental results and proofs
\end{center}
\bigskip

\setcounter{section}{0}
\setcounter{equation}{0}
\setcounter{figure}{0}
\setcounter{table}{0}
\setcounter{theorem}{0}
\renewcommand{\thesection}{S\arabic{section}}
\renewcommand{\theequation}{S\arabic{equation}}
\renewcommand{\thefigure}{S\arabic{figure}}
\renewcommand{\thetable}{S\arabic{table}}
\renewcommand{\thetheorem}{S\arabic{theorem}}
\renewcommand{\theHsection}{supp.\arabic{section}}
\renewcommand{\theHequation}{supp.\arabic{equation}}
\renewcommand{\theHfigure}{supp.\arabic{figure}}
\renewcommand{\theHtable}{supp.\arabic{table}}
\renewcommand{\theHtheorem}{supp.\arabic{theorem}}

\section{Experimental Analysis}
\label{app:experimental-analysis}
Table~\ref{tab:full-SDP} gives an overview of the experimental design and outcome measures. Table~\ref{tab:stacked-certificate-details} gives the detailed numerical-certification results.

{
\small\setlength{\tabcolsep}{2pt}
\begin{longtable}{@{}p{0.18\textwidth} p{0.13\textwidth} p{0.22\textwidth} p{0.39\textwidth}@{}}
	\caption{Experimental design and outcome measures}
	\label{tab:full-SDP} \\
	\toprule
	\textbf{Dimension} & \textbf{Parameter} & \textbf{Values} & \textbf{Role} \\
	\midrule
	\endfirsthead
    
	\caption*{Table \ref{tab:full-SDP}: Experimental design and outcome measures (continued)} \\
	\toprule
	\textbf{Dimension} & \textbf{Parameter} & \textbf{Values} & \textbf{Role} \\
	\midrule
	\endhead
    
	\midrule
	\multicolumn{4}{r}{\textit{Continued on next page}} \\
	\endfoot
    
	\bottomrule
	\endlastfoot

	\multicolumn{4}{l}{\textit{Experimental factors}} \\[0.2em]
	Number of firms & $n$ & $2,3,4$ & Changes the intensity of Bertrand competition and the projected game gradient. \\[0.3em]
	Cost distribution & $\ell$ & $1,2,3$ & Beta$(1,\ell)$ cost prior. Larger $\ell$ places more probability mass on lower costs. \\[0.3em]
	Demand curvature & $s$ & $0,1,2$ & Determines $D^s(p)=(1-p)^s$: all-or-nothing, linear, and quadratic demand. \\[0.3em]
	Strategy dimension & $m$ & $2,3,4$ & Number of segments in the piecewise-linear pricing strategy. Larger $m$ enriches the deviation space and increases the dimension of the dynamics. Three additional probes are run at $m=5$. \\[0.3em]
	Learning geometry & $q$ & $0$ (Burg), $1$ (entropic), $2$ (Euclidean) & Determines the Riemannian scaling of the game gradient. Smaller $q$ produces stronger state dependence near the lower-slope faces of $\mathcal{B}_\delta^m$. \\
	\midrule
	\multicolumn{4}{l}{\textit{Fixed settings}} \\[0.2em]
	Lower slope bound & $\delta$ & $1/250$ & Keeps admissible pricing functions uniformly increasing and fixes the feasible strategy region. \\[0.3em]
	Positivity scale & $\tau$ & $10^{-4}$ & Lower matrix bound in $\tau E\preceq M$. \\[0.3em]
	Matrix normalization & --- & $M\preceq I$ & Fixes the scale of the Lyapunov representative and removes homogeneity. \\
	\midrule
	\multicolumn{4}{l}{\textit{Outcome measures}} \\[0.2em]
	Minty witness & $M$ & positive / inconclusive & Whether the explicit witness satisfies $(x(\delta)-x^*)^\top v(x(\delta))>0$ at $\delta=1/250$. This property is independent of the learning geometry. \\[0.3em]
	Lyapunov certificate & $C$ & yes / no & Whether a numerical Lyapunov certificate is found within the polynomial and SOS classes tested. A solver-feasible representation is found in all $243$ systems of the design; its reconstructed diagnostics determine numerical-certificate status. \\[0.3em]
	Minimum Lyapunov degree & $d_L^\star$ & $2,4,6$ & Lowest tested polynomial degree at which a numerical certificate is found. \\[0.3em]
	SOS multiplier degrees & $d_{\text{dec}},d_{\text{bnd}}$ & $4,\ldots,14$ and $0,\ldots,10$ & Degrees of the interior and boundary SOS multipliers. The value $d_{\text{bnd}}=0$ means that constant face multipliers suffice; the boundary constraints remain part of the program. \\[0.3em]
	Computational effort & $N_{\mathrm{SDP}},\,T$ & $199$--$181\,800$ variables; $0.5$--$2303$\,s & SDP size and MOSEK runtime required for the reported certification problem. \\
\end{longtable}
}

\FloatBarrier

{\footnotesize\setlength{\tabcolsep}{2.2pt}
\begin{longtable}{ccccrrrrr @{\hspace{2.2em}} ccccrrrrr}
\caption{Instance-level numerical-certification results for all $243$ projected dynamical systems, together with three additional probes at $m=5$; the table is read column-block by column-block, left before right. In the Minty-witness column, a \checkmark{} means that the explicit witness satisfies $(x(\delta)-x^*)^\top v(x(\delta))>0$ at $\delta=1/250$ and hence proves Minty failure; a $\times$ is inconclusive. The degree $d_L^\star$ is the smallest Lyapunov degree at which a numerical certificate is obtained; $d_{\mathrm{dec}}$ and $d_{\mathrm{bnd}}$ are the interior and boundary multiplier degrees; $N_{\mathrm{SDP}}$ is the number of scalar SDP variables and $T$ its MOSEK runtime. All reported runs terminated with a solver-feasible solution. Throughout $\delta=1/250$ and $\tau=10^{-4}$.}\\
\label{tab:stacked-certificate-details}\\
\toprule
$\ell$ & $n$ & $s$ & $M$ & $d_L^\star$ & $d_{\mathrm{dec}}$ & $d_{\mathrm{bnd}}$ & $N_{\mathrm{SDP}}$ & $T$ [s] & $\ell$ & $n$ & $s$ & $M$ & $d_L^\star$ & $d_{\mathrm{dec}}$ & $d_{\mathrm{bnd}}$ & $N_{\mathrm{SDP}}$ & $T$ [s] \\
\midrule
\endfirsthead
\multicolumn{18}{l}{\small\itshape Table~\ref{tab:stacked-certificate-details} continued from previous page}\\
\toprule
$\ell$ & $n$ & $s$ & $M$ & $d_L^\star$ & $d_{\mathrm{dec}}$ & $d_{\mathrm{bnd}}$ & $N_{\mathrm{SDP}}$ & $T$ [s] & $\ell$ & $n$ & $s$ & $M$ & $d_L^\star$ & $d_{\mathrm{dec}}$ & $d_{\mathrm{bnd}}$ & $N_{\mathrm{SDP}}$ & $T$ [s] \\
\midrule
\endhead
\midrule \multicolumn{18}{r}{\small\itshape continued on next page}\\
\endfoot
\bottomrule
\endlastfoot
\addlinespace[0.25em]
\multicolumn{9}{l}{\itshape Burg ($q=0$), $m=2$} & 2 & 4 & 0 & $\times$ & 2 & 6 & 2 & 2\,312 & 4.6 \\
1 & 2 & 0 & \checkmark & 4 & 6 & 4 & 805 & 2.6 & 2 & 4 & 1 & $\times$ & 2 & 6 & 2 & 2\,312 & 4.9 \\
1 & 2 & 1 & \checkmark & 4 & 12 & 4 & 2\,593 & 5.0 & 2 & 4 & 2 & $\times$ & 2 & 6 & 2 & 2\,312 & 5.5 \\
1 & 2 & 2 & $\times$ & 2 & 8 & 6 & 1\,671 & 2.4 & 3 & 2 & 0 & \checkmark & 4 & 6 & 4 & 3\,887 & 12.1 \\
1 & 3 & 0 & \checkmark & 2 & 6 & 2 & 525 & 1.0 & 3 & 2 & 1 & \checkmark & 2 & 8 & 2 & 5\,282 & 9.6 \\
1 & 3 & 1 & $\times$ & 2 & 10 & 6 & 2\,242 & 3.5 & 3 & 2 & 2 & $\times$ & 2 & 6 & 2 & 2\,312 & 5.4 \\
1 & 3 & 2 & $\times$ & 6 & 14 & 6 & 4\,356 & 8.5 & 3 & 3 & 0 & $\times$ & 2 & 8 & 2 & 5\,282 & 11.7 \\
1 & 4 & 0 & $\times$ & 2 & 6 & 2 & 525 & 1.0 & 3 & 3 & 1 & $\times$ & 2 & 6 & 2 & 2\,312 & 5.8 \\
1 & 4 & 1 & $\times$ & 4 & 8 & 6 & 1\,678 & 3.7 & 3 & 3 & 2 & $\times$ & 2 & 6 & 2 & 2\,312 & 6.0 \\
1 & 4 & 2 & $\times$ & 6 & 14 & 6 & 4\,356 & 11.0 & 3 & 4 & 0 & $\times$ & 2 & 6 & 2 & 2\,312 & 4.2 \\
2 & 2 & 0 & \checkmark & 2 & 6 & 2 & 525 & 0.8 & 3 & 4 & 1 & $\times$ & 2 & 6 & 2 & 2\,312 & 4.5 \\
2 & 2 & 1 & $\times$ & 4 & 8 & 6 & 1\,678 & 3.1 & 3 & 4 & 2 & $\times$ & 2 & 6 & 2 & 2\,312 & 5.7 \\
\addlinespace[0.25em]
2 & 2 & 2 & $\times$ & 6 & 14 & 10 & 6\,612 & 17.0 & \multicolumn{9}{l}{\itshape Entropic ($q=1$), $m=4$} \\
2 & 3 & 0 & $\times$ & 2 & 6 & 2 & 525 & 0.8 & 1 & 2 & 0 & \checkmark & 2 & 8 & 2 & 24\,186 & 87.1 \\
2 & 3 & 1 & $\times$ & 4 & 8 & 4 & 1\,162 & 2.9 & 1 & 2 & 1 & \checkmark & 2 & 6 & 2 & 7\,970 & 26.0 \\
2 & 3 & 2 & $\times$ & 6 & 14 & 10 & 6\,612 & 15.4 & 1 & 2 & 2 & \checkmark & 2 & 8 & 2 & 24\,186 & 100.1 \\
2 & 4 & 0 & $\times$ & 2 & 6 & 2 & 525 & 1.0 & 1 & 3 & 0 & \checkmark & 2 & 6 & 2 & 7\,970 & 25.9 \\
2 & 4 & 1 & $\times$ & 2 & 10 & 6 & 2\,242 & 3.7 & 1 & 3 & 1 & \checkmark & 2 & 6 & 2 & 7\,970 & 28.4 \\
2 & 4 & 2 & $\times$ & 6 & 14 & 6 & 4\,356 & 10.2 & 1 & 3 & 2 & \checkmark & 2 & 8 & 2 & 24\,186 & 83.5 \\
3 & 2 & 0 & \checkmark & 2 & 6 & 2 & 525 & 0.8 & 1 & 4 & 0 & \checkmark & 2 & 6 & 2 & 7\,970 & 18.4 \\
3 & 2 & 1 & $\times$ & 4 & 8 & 4 & 1\,162 & 2.6 & 1 & 4 & 1 & \checkmark & 2 & 6 & 2 & 7\,970 & 22.2 \\
3 & 2 & 2 & $\times$ & 6 & 14 & 6 & 4\,356 & 10.3 & 1 & 4 & 2 & \checkmark & 2 & 8 & 2 & 24\,186 & 118.9 \\
3 & 3 & 0 & $\times$ & 2 & 6 & 2 & 525 & 0.8 & 2 & 2 & 0 & \checkmark & 2 & 8 & 2 & 24\,186 & 102.0 \\
3 & 3 & 1 & $\times$ & 2 & 10 & 6 & 2\,242 & 3.0 & 2 & 2 & 1 & \checkmark & 2 & 6 & 2 & 7\,970 & 31.5 \\
3 & 3 & 2 & $\times$ & 6 & 14 & 6 & 4\,356 & 10.5 & 2 & 2 & 2 & \checkmark & 2 & 8 & 2 & 24\,186 & 104.6 \\
3 & 4 & 0 & $\times$ & 2 & 6 & 2 & 525 & 0.8 & 2 & 3 & 0 & \checkmark & 2 & 8 & 2 & 24\,186 & 92.8 \\
3 & 4 & 1 & $\times$ & 2 & 10 & 2 & 1\,453 & 1.8 & 2 & 3 & 1 & \checkmark & 2 & 6 & 2 & 7\,970 & 22.4 \\
3 & 4 & 2 & $\times$ & 4 & 10 & 4 & 1\,733 & 3.4 & 2 & 3 & 2 & $\times$ & 2 & 8 & 2 & 24\,186 & 91.2 \\
\addlinespace[0.25em]
\multicolumn{9}{l}{\itshape Burg ($q=0$), $m=3$} & 2 & 4 & 0 & \checkmark & 2 & 8 & 2 & 24\,186 & 97.4 \\
1 & 2 & 0 & \checkmark & 4 & 8 & 4 & 6\,857 & 20.4 & 2 & 4 & 1 & $\times$ & 2 & 6 & 2 & 7\,970 & 38.1 \\
1 & 2 & 1 & \checkmark & 4 & 8 & 4 & 6\,857 & 22.1 & 2 & 4 & 2 & $\times$ & 2 & 8 & 2 & 24\,186 & 111.9 \\
1 & 2 & 2 & \checkmark & 2 & 8 & 2 & 5\,282 & 12.5 & 3 & 2 & 0 & \checkmark & 2 & 8 & 2 & 24\,186 & 108.1 \\
1 & 3 & 0 & \checkmark & 4 & 8 & 4 & 6\,857 & 18.3 & 3 & 2 & 1 & \checkmark & 2 & 8 & 2 & 24\,186 & 92.7 \\
1 & 3 & 1 & \checkmark & 4 & 8 & 4 & 6\,857 & 21.5 & 3 & 2 & 2 & \checkmark & 2 & 8 & 2 & 24\,186 & 82.4 \\
1 & 3 & 2 & $\times$ & 2 & 12 & 4 & 25\,426 & 73.7 & 3 & 3 & 0 & \checkmark & 2 & 8 & 2 & 24\,186 & 78.7 \\
1 & 4 & 0 & \checkmark & 4 & 8 & 4 & 6\,857 & 19.8 & 3 & 3 & 1 & $\times$ & 2 & 6 & 2 & 7\,970 & 26.7 \\
1 & 4 & 1 & \checkmark & 4 & 8 & 4 & 6\,857 & 24.4 & 3 & 3 & 2 & $\times$ & 2 & 8 & 2 & 24\,186 & 90.9 \\
1 & 4 & 2 & $\times$ & 2 & 12 & 2 & 23\,866 & 56.7 & 3 & 4 & 0 & $\times$ & 2 & 8 & 2 & 24\,186 & 91.4 \\
2 & 2 & 0 & \checkmark & 4 & 8 & 4 & 6\,857 & 22.5 & 3 & 4 & 1 & $\times$ & 2 & 6 & 2 & 7\,970 & 25.8 \\
2 & 2 & 1 & \checkmark & 4 & 8 & 4 & 6\,857 & 24.9 & 3 & 4 & 2 & $\times$ & 2 & 8 & 2 & 24\,186 & 86.4 \\
\addlinespace[0.25em]
2 & 2 & 2 & \checkmark & 2 & 12 & 2 & 23\,866 & 55.9 & \multicolumn{9}{l}{\itshape Euclidean ($q=2$), $m=2$} \\
2 & 3 & 0 & \checkmark & 4 & 8 & 4 & 6\,857 & 19.4 & 1 & 2 & 0 & \checkmark & 2 & 4 & 0 & 199 & 1.4 \\
2 & 3 & 1 & $\times$ & 4 & 8 & 4 & 6\,857 & 25.2 & 1 & 2 & 1 & \checkmark & 2 & 4 & 0 & 199 & 1.4 \\
2 & 3 & 2 & $\times$ & 2 & 8 & 4 & 6\,842 & 19.4 & 1 & 2 & 2 & $\times$ & 2 & 4 & 0 & 199 & 1.0 \\
2 & 4 & 0 & $\times$ & 4 & 8 & 4 & 6\,857 & 18.9 & 1 & 3 & 0 & \checkmark & 2 & 4 & 0 & 199 & 0.6 \\
2 & 4 & 1 & $\times$ & 4 & 8 & 4 & 6\,857 & 25.4 & 1 & 3 & 1 & $\times$ & 2 & 4 & 0 & 199 & 0.6 \\
2 & 4 & 2 & $\times$ & 2 & 12 & 2 & 23\,866 & 55.9 & 1 & 3 & 2 & $\times$ & 2 & 4 & 0 & 199 & 1.1 \\
3 & 2 & 0 & \checkmark & 4 & 8 & 4 & 6\,857 & 23.1 & 1 & 4 & 0 & $\times$ & 2 & 4 & 0 & 199 & 0.8 \\
3 & 2 & 1 & \checkmark & 4 & 8 & 4 & 6\,857 & 20.9 & 1 & 4 & 1 & $\times$ & 2 & 4 & 0 & 199 & 0.5 \\
3 & 2 & 2 & $\times$ & 2 & 8 & 4 & 6\,842 & 18.6 & 1 & 4 & 2 & $\times$ & 2 & 4 & 0 & 199 & 0.8 \\
3 & 3 & 0 & $\times$ & 4 & 8 & 4 & 6\,857 & 18.6 & 2 & 2 & 0 & \checkmark & 2 & 4 & 0 & 199 & 0.8 \\
3 & 3 & 1 & $\times$ & 4 & 8 & 4 & 6\,857 & 20.4 & 2 & 2 & 1 & $\times$ & 2 & 4 & 0 & 199 & 0.8 \\
3 & 3 & 2 & $\times$ & 2 & 12 & 2 & 23\,866 & 64.9 & 2 & 2 & 2 & $\times$ & 2 & 4 & 0 & 199 & 0.7 \\
3 & 4 & 0 & $\times$ & 4 & 8 & 4 & 6\,857 & 23.9 & 2 & 3 & 0 & $\times$ & 2 & 4 & 0 & 199 & 0.6 \\
3 & 4 & 1 & $\times$ & 2 & 10 & 2 & 11\,604 & 24.1 & 2 & 3 & 1 & $\times$ & 2 & 4 & 0 & 199 & 0.7 \\
3 & 4 & 2 & $\times$ & 2 & 8 & 2 & 5\,282 & 10.5 & 2 & 3 & 2 & $\times$ & 2 & 4 & 0 & 199 & 1.0 \\
\addlinespace[0.25em]
\multicolumn{9}{l}{\itshape Burg ($q=0$), $m=4$} & 2 & 4 & 0 & $\times$ & 2 & 4 & 0 & 199 & 0.6 \\
1 & 2 & 0 & \checkmark & 2 & 8 & 2 & 24\,186 & 91.3 & 2 & 4 & 1 & $\times$ & 2 & 4 & 0 & 199 & 1.9 \\
1 & 2 & 1 & \checkmark & 2 & 10 & 4 & 74\,475 & 489.0 & 2 & 4 & 2 & $\times$ & 2 & 4 & 0 & 199 & 0.7 \\
1 & 2 & 2 & \checkmark & 2 & 8 & 4 & 30\,211 & 154.8 & 3 & 2 & 0 & \checkmark & 2 & 4 & 0 & 199 & 0.5 \\
1 & 3 & 0 & \checkmark & 2 & 6 & 4 & 13\,995 & 60.6 & 3 & 2 & 1 & $\times$ & 2 & 6 & 0 & 405 & 1.1 \\
1 & 3 & 1 & \checkmark & 2 & 8 & 4 & 30\,211 & 158.4 & 3 & 2 & 2 & $\times$ & 2 & 4 & 0 & 199 & 0.6 \\
1 & 3 & 2 & \checkmark & 2 & 8 & 2 & 24\,186 & 88.6 & 3 & 3 & 0 & $\times$ & 2 & 4 & 0 & 199 & 0.6 \\
1 & 4 & 0 & \checkmark & 6 & 10 & 6 & 97\,518 & 823.5 & 3 & 3 & 1 & $\times$ & 2 & 6 & 0 & 405 & 1.3 \\
1 & 4 & 1 & \checkmark & 2 & 12 & 4 & 181\,800 & 2303.3 & 3 & 3 & 2 & $\times$ & 2 & 4 & 0 & 199 & 0.8 \\
1 & 4 & 2 & \checkmark & 2 & 8 & 2 & 24\,186 & 111.3 & 3 & 4 & 0 & $\times$ & 2 & 4 & 0 & 199 & 0.6 \\
2 & 2 & 0 & \checkmark & 2 & 8 & 4 & 30\,211 & 145.3 & 3 & 4 & 1 & $\times$ & 2 & 6 & 0 & 405 & 0.8 \\
2 & 2 & 1 & \checkmark & 2 & 8 & 4 & 30\,211 & 152.6 & 3 & 4 & 2 & $\times$ & 2 & 4 & 0 & 199 & 0.7 \\
\addlinespace[0.25em]
2 & 2 & 2 & \checkmark & 2 & 8 & 4 & 30\,211 & 305.8 & \multicolumn{9}{l}{\itshape Euclidean ($q=2$), $m=3$} \\
2 & 3 & 0 & \checkmark & 2 & 8 & 6 & 53\,186 & 377.4 & 1 & 2 & 0 & \checkmark & 2 & 6 & 0 & 1\,880 & 6.1 \\
2 & 3 & 1 & \checkmark & 2 & 8 & 4 & 30\,211 & 163.1 & 1 & 2 & 1 & \checkmark & 2 & 6 & 0 & 1\,880 & 9.1 \\
2 & 3 & 2 & $\times$ & 2 & 10 & 4 & 74\,475 & 634.9 & 1 & 2 & 2 & \checkmark & 2 & 6 & 0 & 1\,880 & 7.1 \\
2 & 4 & 0 & \checkmark & 2 & 8 & 2 & 24\,186 & 98.4 & 1 & 3 & 0 & \checkmark & 2 & 6 & 0 & 1\,880 & 7.3 \\
2 & 4 & 1 & $\times$ & 2 & 8 & 2 & 24\,186 & 82.9 & 1 & 3 & 1 & \checkmark & 2 & 6 & 0 & 1\,880 & 10.7 \\
2 & 4 & 2 & $\times$ & 2 & 10 & 2 & 68\,450 & 453.7 & 1 & 3 & 2 & $\times$ & 2 & 6 & 0 & 1\,880 & 14.7 \\
3 & 2 & 0 & \checkmark & 4 & 8 & 4 & 30\,237 & 167.1 & 1 & 4 & 0 & \checkmark & 2 & 6 & 0 & 1\,880 & 5.5 \\
3 & 2 & 1 & \checkmark & 4 & 10 & 4 & 74\,501 & 598.9 & 1 & 4 & 1 & \checkmark & 2 & 6 & 0 & 1\,880 & 4.3 \\
3 & 2 & 2 & \checkmark & 2 & 10 & 2 & 68\,450 & 498.6 & 1 & 4 & 2 & $\times$ & 2 & 6 & 0 & 1\,880 & 7.4 \\
3 & 3 & 0 & \checkmark & 2 & 8 & 4 & 30\,211 & 156.9 & 2 & 2 & 0 & \checkmark & 2 & 6 & 0 & 1\,880 & 9.6 \\
3 & 3 & 1 & $\times$ & 2 & 8 & 2 & 24\,186 & 108.8 & 2 & 2 & 1 & \checkmark & 2 & 6 & 0 & 1\,880 & 8.9 \\
3 & 3 & 2 & $\times$ & 2 & 8 & 2 & 24\,186 & 104.5 & 2 & 2 & 2 & \checkmark & 2 & 6 & 0 & 1\,880 & 8.5 \\
3 & 4 & 0 & $\times$ & 2 & 8 & 2 & 24\,186 & 89.1 & 2 & 3 & 0 & \checkmark & 2 & 6 & 0 & 1\,880 & 8.5 \\
3 & 4 & 1 & $\times$ & 2 & 8 & 2 & 24\,186 & 96.7 & 2 & 3 & 1 & $\times$ & 2 & 6 & 0 & 1\,880 & 6.5 \\
3 & 4 & 2 & $\times$ & 2 & 8 & 2 & 24\,186 & 100.6 & 2 & 3 & 2 & $\times$ & 2 & 6 & 0 & 1\,880 & 8.4 \\
\addlinespace[0.25em]
\multicolumn{9}{l}{\itshape Entropic ($q=1$), $m=2$} & 2 & 4 & 0 & $\times$ & 2 & 8 & 0 & 4\,850 & 12.6 \\
1 & 2 & 0 & \checkmark & 2 & 8 & 6 & 1\,671 & 2.5 & 2 & 4 & 1 & $\times$ & 2 & 6 & 0 & 1\,880 & 4.1 \\
1 & 2 & 1 & \checkmark & 4 & 6 & 4 & 805 & 2.4 & 2 & 4 & 2 & $\times$ & 2 & 6 & 0 & 1\,880 & 4.9 \\
1 & 2 & 2 & $\times$ & 2 & 8 & 2 & 882 & 1.2 & 3 & 2 & 0 & \checkmark & 2 & 6 & 0 & 1\,880 & 8.6 \\
1 & 3 & 0 & \checkmark & 2 & 6 & 4 & 798 & 1.4 & 3 & 2 & 1 & \checkmark & 2 & 6 & 0 & 1\,880 & 6.0 \\
1 & 3 & 1 & $\times$ & 2 & 8 & 2 & 882 & 1.5 & 3 & 2 & 2 & $\times$ & 2 & 6 & 0 & 1\,880 & 7.7 \\
1 & 3 & 2 & $\times$ & 2 & 8 & 2 & 882 & 1.5 & 3 & 3 & 0 & $\times$ & 2 & 8 & 0 & 4\,850 & 18.8 \\
1 & 4 & 0 & $\times$ & 2 & 8 & 2 & 882 & 1.3 & 3 & 3 & 1 & $\times$ & 2 & 6 & 0 & 1\,880 & 7.3 \\
1 & 4 & 1 & $\times$ & 2 & 8 & 2 & 882 & 1.6 & 3 & 3 & 2 & $\times$ & 2 & 6 & 0 & 1\,880 & 7.3 \\
1 & 4 & 2 & $\times$ & 2 & 6 & 2 & 525 & 1.1 & 3 & 4 & 0 & $\times$ & 2 & 8 & 0 & 4\,850 & 12.6 \\
2 & 2 & 0 & \checkmark & 2 & 6 & 6 & 1\,314 & 1.9 & 3 & 4 & 1 & $\times$ & 2 & 6 & 0 & 1\,880 & 5.5 \\
2 & 2 & 1 & $\times$ & 2 & 8 & 4 & 1\,155 & 1.6 & 3 & 4 & 2 & $\times$ & 2 & 6 & 0 & 1\,880 & 7.3 \\
\addlinespace[0.25em]
2 & 2 & 2 & $\times$ & 2 & 8 & 2 & 882 & 1.3 & \multicolumn{9}{l}{\itshape Euclidean ($q=2$), $m=4$} \\
2 & 3 & 0 & $\times$ & 2 & 8 & 2 & 882 & 1.9 & 1 & 2 & 0 & \checkmark & 2 & 6 & 0 & 6\,815 & 37.2 \\
2 & 3 & 1 & $\times$ & 2 & 8 & 2 & 882 & 1.5 & 1 & 2 & 1 & \checkmark & 2 & 6 & 0 & 6\,815 & 22.9 \\
2 & 3 & 2 & $\times$ & 2 & 8 & 2 & 882 & 1.5 & 1 & 2 & 2 & \checkmark & 2 & 6 & 0 & 6\,815 & 32.9 \\
2 & 4 & 0 & $\times$ & 2 & 6 & 4 & 798 & 1.5 & 1 & 3 & 0 & \checkmark & 2 & 6 & 0 & 6\,815 & 18.5 \\
2 & 4 & 1 & $\times$ & 2 & 8 & 2 & 882 & 1.4 & 1 & 3 & 1 & \checkmark & 2 & 6 & 0 & 6\,815 & 22.4 \\
2 & 4 & 2 & $\times$ & 2 & 8 & 2 & 882 & 1.3 & 1 & 3 & 2 & \checkmark & 2 & 6 & 0 & 6\,815 & 24.7 \\
3 & 2 & 0 & \checkmark & 2 & 6 & 6 & 1\,314 & 2.1 & 1 & 4 & 0 & \checkmark & 2 & 6 & 0 & 6\,815 & 19.7 \\
3 & 2 & 1 & $\times$ & 2 & 8 & 6 & 1\,671 & 2.9 & 1 & 4 & 1 & \checkmark & 2 & 6 & 0 & 6\,815 & 20.4 \\
3 & 2 & 2 & $\times$ & 2 & 8 & 2 & 882 & 1.5 & 1 & 4 & 2 & \checkmark & 2 & 6 & 0 & 6\,815 & 22.7 \\
3 & 3 & 0 & $\times$ & 2 & 8 & 4 & 1\,155 & 1.8 & 2 & 2 & 0 & \checkmark & 2 & 8 & 0 & 23\,031 & 92.8 \\
3 & 3 & 1 & $\times$ & 2 & 8 & 2 & 882 & 1.8 & 2 & 2 & 1 & \checkmark & 2 & 8 & 0 & 23\,031 & 69.0 \\
3 & 3 & 2 & $\times$ & 2 & 10 & 2 & 1\,453 & 2.3 & 2 & 2 & 2 & \checkmark & 2 & 6 & 0 & 6\,815 & 30.7 \\
3 & 4 & 0 & $\times$ & 2 & 8 & 2 & 882 & 1.2 & 2 & 3 & 0 & \checkmark & 2 & 8 & 0 & 23\,031 & 77.3 \\
3 & 4 & 1 & $\times$ & 2 & 6 & 4 & 798 & 1.1 & 2 & 3 & 1 & \checkmark & 2 & 6 & 0 & 6\,815 & 16.8 \\
3 & 4 & 2 & $\times$ & 2 & 8 & 2 & 882 & 1.2 & 2 & 3 & 2 & $\times$ & 2 & 6 & 0 & 6\,815 & 22.5 \\
\addlinespace[0.25em]
\multicolumn{9}{l}{\itshape Entropic ($q=1$), $m=3$} & 2 & 4 & 0 & \checkmark & 2 & 8 & 0 & 23\,031 & 69.8 \\
1 & 2 & 0 & \checkmark & 4 & 6 & 4 & 3\,887 & 14.7 & 2 & 4 & 1 & $\times$ & 2 & 6 & 0 & 6\,815 & 20.6 \\
1 & 2 & 1 & \checkmark & 2 & 6 & 2 & 2\,312 & 5.3 & 2 & 4 & 2 & $\times$ & 2 & 6 & 0 & 6\,815 & 19.9 \\
1 & 2 & 2 & \checkmark & 2 & 6 & 2 & 2\,312 & 5.3 & 3 & 2 & 0 & \checkmark & 2 & 8 & 0 & 23\,031 & 130.1 \\
1 & 3 & 0 & \checkmark & 2 & 6 & 6 & 8\,196 & 23.4 & 3 & 2 & 1 & \checkmark & 2 & 6 & 0 & 6\,815 & 25.8 \\
1 & 3 & 1 & \checkmark & 2 & 8 & 2 & 5\,282 & 9.8 & 3 & 2 & 2 & \checkmark & 2 & 6 & 0 & 6\,815 & 21.8 \\
1 & 3 & 2 & $\times$ & 2 & 6 & 2 & 2\,312 & 4.8 & 3 & 3 & 0 & \checkmark & 2 & 8 & 0 & 23\,031 & 77.5 \\
1 & 4 & 0 & \checkmark & 2 & 8 & 2 & 5\,282 & 13.0 & 3 & 3 & 1 & $\times$ & 2 & 8 & 0 & 23\,031 & 83.7 \\
1 & 4 & 1 & \checkmark & 2 & 6 & 2 & 2\,312 & 5.7 & 3 & 3 & 2 & $\times$ & 2 & 6 & 0 & 6\,815 & 22.5 \\
1 & 4 & 2 & $\times$ & 2 & 6 & 2 & 2\,312 & 5.7 & 3 & 4 & 0 & $\times$ & 2 & 10 & 0 & 67\,295 & 306.3 \\
2 & 2 & 0 & \checkmark & 4 & 6 & 4 & 3\,887 & 14.0 & 3 & 4 & 1 & $\times$ & 2 & 8 & 0 & 23\,031 & 80.2 \\
2 & 2 & 1 & \checkmark & 2 & 8 & 2 & 5\,282 & 10.3 & 3 & 4 & 2 & $\times$ & 2 & 6 & 0 & 6\,815 & 18.0 \\
\addlinespace[0.25em]
2 & 2 & 2 & \checkmark & 2 & 8 & 2 & 5\,282 & 9.7 & \multicolumn{9}{l}{\itshape Probes at $m=5$ (Burg, entropic, Euclidean)} \\
2 & 3 & 0 & \checkmark & 2 & 8 & 2 & 5\,282 & 9.4 & 1 & 2 & 0 & \checkmark & 2 & 8 & 2 & 90\,393 & 1212.3 \\
2 & 3 & 1 & $\times$ & 2 & 6 & 2 & 2\,312 & 5.0 & 1 & 2 & 0 & \checkmark & 2 & 8 & 2 & 90\,393 & 1176.7 \\
2 & 3 & 2 & $\times$ & 2 & 6 & 2 & 2\,312 & 5.5 & 1 & 2 & 0 & \checkmark & 2 & 8 & 0 & 87\,813 & 1212.1 \\
\end{longtable}
}

\section{Equilibrium Analysis}

\subsection{Proof of Proposition~\ref{prop:symmetric-bne}}
\label{app:proof_symmetric_bne}

\begin{proof}
	Definition~\ref{def:sym-bne-ref} permits all measurable price deviations. Thus the equilibrium price is a pointwise best response for almost every type. For $0<c<1$, strict monotonicity and $\beta(1)=1$ imply $\beta(c)\in(0,1)$ and positive winning probability. Pricing below cost cannot be optimal because price $1$ yields zero profit; pricing exactly at cost has a positive right derivative. Hence an equilibrium has $\beta(c)>c$ almost everywhere, and its interior price first-order condition is necessary. The expected payoff from submitting a price \(p\) in the range of \(\beta\) is
	\[
		U^{\mathrm{exp},s}(p) = \left(1-H(\beta^{-1}(p))\right)(1-p)^s(p-c).
	\]
	Differentiating with respect to \(p\) and imposing the symmetric equilibrium condition \(p=\beta(c)\) gives, for almost every \(c\in(0,1)\),
	\[
		\beta'(c) = \frac{h(c)}{1-H(c)} \frac{(1-\beta(c))(\beta(c)-c)} {1-(s+1)\beta(c)+sc}.
	\]

	Suppose now that costs follow a Beta\((1,\ell)\) distribution. Recall
	\[
		a:=\ell(n-1), \qquad H(c)=1-(1-c)^a, \qquad h(c)=a(1-c)^{a-1},
	\]
	and hence
	\[
		\frac{h(c)}{1-H(c)} = \frac{a}{1-c}.
	\]
	The equilibrium equation therefore becomes
	\begin{equation}
		\label{eq:beta-ode-beta-prior}
		\beta'(c) = \frac{a}{1-c} \frac{(1-\beta(c))(\beta(c)-c)} {1-(s+1)\beta(c)+sc}.
	\end{equation}

	Because \(\beta\in W^{1,1}([0,1])\), we work with its absolutely continuous representative. Introduce
	\[
		x=1-c, \qquad y(x)=1-\beta(1-x).
	\]
	Then \(y\) is absolutely continuous, \(y(0)=0\), and
	\[
		y'(x)=\beta'(1-x)
	\]
	for almost every \(x\in(0,1)\). Equation \eqref{eq:beta-ode-beta-prior} becomes
	\[
		y'(x) = \frac{a}{x} \frac{y(x)(x-y(x))} {(s+1)y(x)-sx}
	\]
	almost everywhere.

	For \(x>0\), set
	\[
		u(x)=\frac{y(x)}{x}.
	\]
	Then \(u\) is locally absolutely continuous on \((0,1]\), and \(y'=u+xu'\). Writing
	\[
		\alpha:=a+s,
	\]
	we obtain
	\begin{equation}
		\label{eq:u-ode-beta-prior}
		xu'(x) = \frac{ u(x)\bigl(\alpha-(\alpha+1)u(x)\bigr) }{ (s+1)u(x)-s }
	\end{equation}
	for almost every \(x\in(0,1)\).

	Separating variables in \eqref{eq:u-ode-beta-prior} gives
	\[
		\frac{(s+1)u-s} {u\bigl(\alpha-(\alpha+1)u\bigr)} \,du = \frac{dx}{x}.
	\]
	Using
	\[
		\frac{(s+1)u-s} {u\bigl(\alpha-(\alpha+1)u\bigr)} = -\frac{s}{\alpha}\frac{1}{u} + \frac{a}{\alpha} \frac{1}{\alpha-(\alpha+1)u},
	\]
	integration yields
	\[
		u(x)^{s(\alpha+1)} \left| \alpha-(\alpha+1)u(x) \right|^a = Kx^{-\alpha(\alpha+1)}
	\]
	for some constant \(K\geq0\). The identity extends across the exceptional values of \(u\) by continuity.

	We now use only the admissibility constraints already contained in \(\mathcal{B}_{\delta}\). Since
	\[
		0\leq\beta(c)\leq1,
	\]
	we have
	\[
		0\leq y(x)\leq1 \qquad\text{and hence}\qquad 0\leq u(x)\leq\frac{1}{x}.
	\]
	Consequently, for \(0<x\leq1\),
	\begin{align*}
		K &= x^{\alpha(\alpha+1)} u(x)^{s(\alpha+1)} \left| \alpha-(\alpha+1)u(x) \right|^a \\
		&\leq x^{\alpha(\alpha+1)} x^{-s(\alpha+1)} \left( \alpha+\frac{\alpha+1}{x} \right)^a \\
		&\leq C x^{\alpha(\alpha+1)-s(\alpha+1)-a} \\
		&= C x^{a\alpha},
	\end{align*}
	where \(C>0\) is independent of \(x\). Since \(a>0\) and \(\alpha=a+s>0\), letting \(x\downarrow0\) gives
	\[
		K=0.
	\]

	It follows that
	\[
		u(x)^{s(\alpha+1)} \left| \alpha-(\alpha+1)u(x) \right|^a = 0.
	\]
	Thus,
	\[
		u(x)\in \left\{ 0,\frac{\alpha}{\alpha+1} \right\}
	\]
	for every \(x\in(0,1]\). Since \(u\) is continuous on the connected set \((0,1]\), it must be constant. The solution \(u\equiv0\) would imply \(\beta\equiv1\), contradicting \(\beta'\ge\delta>0\) almost everywhere. Therefore,
	\[
		u(x) = \frac{\alpha}{\alpha+1} = \frac{\ell(n-1)+s}{\ell(n-1)+s+1}.
	\]

	Transforming back to the original variables gives
	\[
		1-\beta(c) = \frac{\ell(n-1)+s}{\ell(n-1)+s+1}(1-c),
	\]
	and hence
	\[
		\beta^*(c) = \frac{(\ell(n-1)+s)c+1} {\ell(n-1)+s+1}.
	\]
	It remains to verify sufficiency against every price, rather than only the ODE. Write $p_0=1-r=1/(\alpha+1)$. For $p\in[p_0,1]$, the winning probability against $\beta^*$ is $((1-p)/r)^a$, so the conditional payoff is
\[
 r^{-a}(1-p)^\alpha(p-c).
\]
For $c<1$, its derivative on $(p_0,1)$ has the sign of
$1+\alpha c-(\alpha+1)p$. The unique maximizer is therefore
$p^*(c)=(1+\alpha c)/(\alpha+1)=\beta^*(c)$.
For $p<p_0$, the firm wins with probability one and earns $(1-p)^s(p-c)$. When $s=0$, this is strictly increasing in $p$. When $s\ge1$, its derivative is
\[
 (1-p)^{s-1}\bigl[1-(s+1)p+sc\bigr]>0,
\]
because $p<p_0$ and $\alpha=a+s>s$. Thus a price below $p_0$ cannot improve on $p_0$, and hence cannot improve on $\beta^*(c)$. For $c=1$, price $1$ yields zero and every lower price yields a loss. This verifies the candidate against unrestricted price deviations and, after integration, against all measurable strategy deviations. The ODE argument above applies to every admissible BNE under Definition~\ref{def:sym-bne-ref}, proving uniqueness in that class. Additional equilibria caused solely by restricting the deviation set to a finite-dimensional class are not excluded by this ODE argument; a global Lyapunov certificate excludes additional rest points of the particular projected dynamics.
\end{proof}

\subsection{Proof of Proposition~\ref{prop:affine-class}}
\label{app:proof_affine}

\begin{proof}
	Suppose the equilibrium equation~\eqref{eq:sym-bne-ode} admits an affine solution \( \beta(c)=1-r(1-c) \) with \(r\in(0,1)\). Then
	\[
		1-\beta(c)=r(1-c),\qquad \beta(c)-c=(1-r)(1-c),
	\]
	and
	\[
		1-(s+1)\beta(c)+sc = (1-c)\bigl[(s+1)r-s\bigr].
	\]
	Moreover, since
	\[
		1-H(c)=(1-F(c))^{n-1},
	\]
	we have
	\[
		\frac{h(c)}{1-H(c)} = (n-1)\frac{f(c)}{1-F(c)}.
	\]
	Substitution into~\eqref{eq:sym-bne-ode}, using \(\beta'(c)=r\), gives
	\[
		r = (n-1)\frac{f(c)}{1-F(c)} \frac{r(1-r)(1-c)} {(s+1)r-s}.
	\]
	Since \(r>0\), this is equivalent to
	\begin{equation}
		(1-c)\frac{f(c)}{1-F(c)} = \frac{(s+1)r-s}{(n-1)(1-r)} =:\kappa,
		\label{eq:affine-hazard-condition}
	\end{equation}
	so the left-hand side must be constant.

	Conversely, suppose \( (1-c)f(c)/(1-F(c))\equiv\kappa \) for some \(\kappa>0\). Since
	\[
		-\frac{d}{dc}\log(1-F(c)) = \frac{f(c)}{1-F(c)} = \frac{\kappa}{1-c},
	\]
	integration and \(F(0)=0\) yield
	\[
		1-F(c)=(1-c)^\kappa,
	\]
	i.e., \(F\) is Beta$(1,\kappa)$.

	Finally, inserting the constant \(\kappa\) into \eqref{eq:affine-hazard-condition} gives
	\[
		(n-1)\kappa(1-r)=(s+1)r-s,
	\]
	and therefore
	\[
		r= \frac{\kappa(n-1)+s} {\kappa(n-1)+s+1}.
	\]
	With this value of \(r\), direct substitution verifies that \(\beta(c)=1-r(1-c)\) satisfies~\eqref{eq:sym-bne-ode}.
\end{proof}
\subsection{Proof of Proposition~\ref{prop:game-gradient}}
\label{app:proof-game-gradient}

\begin{proof}
	Fix a common rival parameter vector $\mathbf x\in\mathcal B_\delta^m$ and let \( W_{\mathbf x}(p) \) denote the probability that a price $p$ beats all $n-1$ rivals using $\beta_{\mathbf x}$. By strict monotonicity of $\beta_{\mathbf x}$,
	\begin{equation}
		W_{\mathbf x}(p) = 1-H\!\left(\beta_{\mathbf x}^{-1}(p)\right)
	\end{equation}
	for $p$ in the interior of the range of $\beta_{\mathbf x}$. Hence
	\begin{equation}
		u(\mathbf y;\mathbf x) = \int_0^1 k\bigl(\beta_{\mathbf y}(c),c\bigr) W_{\mathbf x}\bigl(\beta_{\mathbf y}(c)\bigr) \,dF(c).
		\label{eq:app-u-gradient}
	\end{equation}

	From the parametrization in \eqref{eq:pl-parametrization} and the definition of $\psi_i$,
	\begin{equation}
		\frac{\partial\beta_{\mathbf y}(c)}{\partial y_i} = -\psi_i(c).
		\label{eq:app-beta-derivative}
	\end{equation}
	Moreover, by the inverse-function rule,
	\begin{equation}
		W'_{\mathbf x}(p) = - \frac{ h\!\left(\beta_{\mathbf x}^{-1}(p)\right) }{ \beta'_{\mathbf x} \!\left(\beta_{\mathbf x}^{-1}(p)\right) }.
	\end{equation}
	Evaluating at $p=\beta_{\mathbf x}(c)$ therefore gives, for almost every $c\in(0,1)$,
	\begin{equation}
		W_{\mathbf x}\bigl(\beta_{\mathbf x}(c)\bigr) = 1-H(c), \qquad W'_{\mathbf x}\bigl(\beta_{\mathbf x}(c)\bigr) = -\frac{h(c)}{\beta'_{\mathbf x}(c)}.
		\label{eq:app-win-derivatives}
	\end{equation}

	Differentiating \eqref{eq:app-u-gradient} with respect to $y_i$ and then setting $\mathbf y=\mathbf x$ yields
	\begin{align}
		v_i(\mathbf x) &= \int_0^1 \frac{\partial\beta_{\mathbf y}(c)}{\partial y_i}
		\Bigl[ \partial_p k\bigl(\beta_{\mathbf x}(c),c\bigr) W_{\mathbf x}\bigl(\beta_{\mathbf x}(c)\bigr) \nonumber\\
		&\hspace{2.4cm} {}+ k\bigl(\beta_{\mathbf x}(c),c\bigr) W'_{\mathbf x}\bigl(\beta_{\mathbf x}(c)\bigr) \Bigr]_{\mathbf y=\mathbf x} dF(c) \nonumber\\
		&= \int_0^1 \psi_i(c) \left[ \frac{h(c)\,k\bigl(\beta_{\mathbf x}(c),c\bigr)} {\beta'_{\mathbf x}(c)} \right. \nonumber\\
		&\hspace{3.6cm}\left. {}- (1-H(c)) \partial_p k\bigl(\beta_{\mathbf x}(c),c\bigr) \right] f(c)\,dc .
		\label{eq:app-general-gradient}
	\end{align}
	Differentiation under the integral is justified by dominated convergence: $\psi_i$ is bounded, $\beta'_{\mathbf x}\ge\delta>0$, and $k$ and $\partial_p k$ are bounded on the compact price--cost domain.

	For a Beta$(1,\ell)$ prior,
	\begin{equation}
		f(c)=\ell(1-c)^{\ell-1}, \qquad 1-H(c)=(1-c)^{\ell(n-1)},
	\end{equation}
	and
	\begin{equation}
		h(c) = \ell(n-1)(1-c)^{\ell(n-1)-1}.
	\end{equation}
	Substituting these expressions into \eqref{eq:app-general-gradient} gives
	\begin{align}
		v_i(\mathbf x) &= \int_0^1 \ell(1-c)^{\ell n-2}\psi_i(c)
		\left[ \frac{ \ell(n-1)k\bigl(\beta_{\mathbf x}(c),c\bigr) }{ \beta'_{\mathbf x}(c) } \right. \nonumber\\
		&\hspace{5.3cm}\left. {}- (1-c) \partial_p k\bigl(\beta_{\mathbf x}(c),c\bigr) \right]dc,
	\end{align}
	which is \eqref{eq:compact-gradient}.
\end{proof}

\subsection{Equilibria under more general priors}
\label{app:general-priors}

The restriction to Beta$(1,\ell)$ costs is not required for the polynomial structure of the projected game gradient. Let costs instead follow a Beta$(k,\ell)$ distribution with $k,\ell\in\mathbb N$. Then the density $f$ and cumulative distribution function $F$ are polynomials in $c$ with rational coefficients. Since \( 1-H(c)=(1-F(c))^{n-1} \) and \( h(c)=(n-1)(1-F(c))^{n-2}f(c), \) the general unilateral-deviation gradient can be written as
\begin{align}
	v_i(\mathbf x) &= \int_0^1 \psi_i(c)
	\left[ \frac{(n-1)(1-F(c))^{n-2}f(c)^2 k(\beta_{\mathbf x}(c),c)} {\beta'_{\mathbf x}(c)} \right. \nonumber\\
	&\hspace{2.8cm}\left. {}- (1-F(c))^{n-1}f(c) \,\partial_p k(\beta_{\mathbf x}(c),c) \right]dc .
	\label{eq:gradient-general-prior}
\end{align}
For polynomial demand $D^s(p)=(1-p)^s$ and the piecewise-linear parametrization of Section~\ref{sec:parametrization}, each integral over a segment is therefore rational in $\mathbf x$. The common denominator may again be chosen as \( p(\mathbf x)=\prod_{j=1}^m x_j \), and \( \deg_{\mathbf x}(p\,v_i)\le s+m \). Thus the polynomial/SOS machinery itself extends directly to these priors; only the coefficients of the vector field change.

What is generally lost is the exact affine equilibrium target. By Proposition~\ref{prop:affine-class}, a Beta$(k,\ell)$ prior with $k\ge2$ does not admit an affine solution of the equilibrium equation. Hence the continuous equilibrium need not belong to the piecewise-linear space $\mathcal L_m$, and the projected equilibrium must instead be determined from \( \mathbf v(\mathbf x^*)=\mathbf 0 \). Its coordinates need not be rational. The present numerical-certification pipeline therefore exploits Beta$(1,\ell)$ not because polynomial certification requires this family, but because it provides the same known, rational equilibrium target for every projection dimension $m$.

For all-or-nothing demand, $s=0$, a closed-form representation is available for an arbitrary absolutely continuous prior. In this case Equation~\eqref{eq:sym-bne-ode} reduces to
\begin{equation}
	\beta'(c) = (n-1)\frac{f(c)}{1-F(c)}\bigl(\beta(c)-c\bigr).
\end{equation}
Writing \( S(c)=(1-F(c))^{n-1} \) and imposing the terminal condition $\beta(1)=1$ gives
\begin{equation}
	\label{eq:aon-general-prior}
	\beta^*(c) = c+ \frac{\displaystyle\int_c^1 S(t)\,dt} {S(c)} = c+ \frac{\displaystyle\int_c^1(1-F(t))^{n-1}\,dt} {(1-F(c))^{n-1}}.
\end{equation}
Thus lack of affineness does not imply lack of an analytical equilibrium. For example, with $n=2$ and a Beta$(2,1)$ prior, \(F(c)=c^2\), and \eqref{eq:aon-general-prior} yields
\begin{equation}
	\beta^*(c) = c+\frac{(1-c)(c+2)}{3(1+c)},
\end{equation}
which is available in closed form but is not affine.

\section{Variational Analysis}

\subsection{The G\^ateaux derivative}
\label{app:gateaux}

We derive the unilateral G\^ateaux derivative used in the variational analysis of Section~\ref{sec:continuous-minty}. Let \(V=W^{1,1}([0,1])\), using the absolutely continuous representative of each element, and let \(\mathcal B_\delta\subset V\) denote the admissible strategy set introduced in Assumption~\ref{ass:iid}.

For a focal strategy $\beta\in\mathcal B_\delta$ and a common rival strategy $\tilde\beta\in\mathcal B_\delta$, write
\[
	\boldsymbol{\tilde\beta}_{-1} = (\tilde\beta,\ldots,\tilde\beta)
\]
for the strategies of the $n-1$ rivals. The G\^ateaux derivative with respect to the focal firm's strategy, holding the rivals fixed, is
\begin{equation}
	\label{eq:gateaux-first-argument}
	D_1U(\beta,\boldsymbol{\tilde\beta}_{-1})[d] := \lim_{\varepsilon\to0} \frac{ U(\beta+\varepsilon d,\boldsymbol{\tilde\beta}_{-1}) - U(\beta,\boldsymbol{\tilde\beta}_{-1}) }{ \varepsilon },
\end{equation}
whenever the two-sided limit exists. At boundary points of the admissible set, the same formula is used for one-sided feasible directions. The subscript $1$ emphasizes that only the focal firm's strategy---the first argument of $U$---is differentiated.

In the symmetric Minty condition, the rivals use the same strategy $\beta$ as the focal symmetric profile. Hence the relevant condition is the following.

\begin{definition}[Symmetric Minty condition]
	\label{def:symmetric-minty-app}
	The symmetric equilibrium $\beta^*\in\mathcal B_\delta$ satisfies the symmetric Minty variational inequality if
	\begin{equation}
		\label{ineq:symMVI}
		D_1U(\beta,\boldsymbol\beta_{-1}) [\beta-\beta^*] \le 0 \qquad \text{for every }\beta\in\mathcal B_\delta,
	\end{equation}
	where \( \boldsymbol\beta_{-1}=(\beta,\ldots,\beta) \).
\end{definition}

\paragraph{Winning probability.} Against rivals using the common strictly increasing strategy $\tilde\beta$, define
\begin{equation}
	\label{eq:winning-probability-app}
	W_{\tilde\beta}(p) := \int_0^1 \chi_{\{p<\tilde\beta(z)\}} \,dH(z).
\end{equation}
Because $\tilde\beta$ is strictly increasing and $\tilde\beta(1)=1$,
\begin{equation}
	\label{eq:winning-probability-piecewise}
	W_{\tilde\beta}(p) =
	\begin{cases}
		1, & p\le\tilde\beta(0), \\[1mm] 1-H\!\left(\tilde\beta^{-1}(p)\right), & \tilde\beta(0)<p<1, \\[1mm] 0, & p\ge1.
	\end{cases}
\end{equation}
Indeed, for $p$ in the interior of the range of $\tilde\beta$,
\[
	p<\tilde\beta(Z) \quad\Longleftrightarrow\quad Z>\tilde\beta^{-1}(p),
\]
where $Z$ is the lowest rival cost with distribution function $H$.

The expected payoff of the focal firm can therefore be written as
\begin{equation}
	\label{eq:expected-payoff-app}
	U(\beta,\boldsymbol{\tilde\beta}_{-1}) = \int_0^1 k(\beta(c),c) W_{\tilde\beta}(\beta(c)) \,dF(c).
\end{equation}

\begin{lemma}[Unilateral G\^ateaux derivative]
	\label{lem:gateaux-derivative}
	Suppose $k(p,c)$ is continuously differentiable with respect to $p$, with $k$ and $\partial_p k$ bounded on the price--cost domain, and suppose the density $h=H'$ is bounded. Then, for $\beta,\tilde\beta\in\mathcal B_\delta$ and every direction $d\in V$ for which $\beta+\varepsilon d$ remains in the price domain for sufficiently small admissible $\varepsilon$,
	\begin{align}
		D_1U(\beta,\boldsymbol{\tilde\beta}_{-1})[d]
		= \int_0^1 d(c) \Bigg[&\partial_p k(\beta(c),c)W_{\tilde\beta}(\beta(c)) \nonumber\\
		&{}- k(\beta(c),c)\chi_{\{\tilde\beta(0)<\beta(c)<1\}} \nonumber\\
		&\qquad\times
		\frac{h\!\left(\tilde\beta^{-1}(\beta(c))\right)}
		{\tilde\beta'\!\left(\tilde\beta^{-1}(\beta(c))\right)}\Bigg]f(c)\,dc .
		\label{eq:derivative-general}
	\end{align}
\end{lemma}

\begin{proof}
	By \eqref{eq:expected-payoff-app},
	\begin{align*}
		&D_1U(\beta,\boldsymbol{\tilde\beta}_{-1})[d] \\
		&\qquad= \lim_{\varepsilon\to0} \int_0^1 \frac{ k(\beta(c)+\varepsilon d(c),c) W_{\tilde\beta}(\beta(c)+\varepsilon d(c)) - k(\beta(c),c) W_{\tilde\beta}(\beta(c)) }{ \varepsilon } \,dF(c).
	\end{align*}
	Adding and subtracting
	\[
		k(\beta(c),c) W_{\tilde\beta}(\beta(c)+\varepsilon d(c))
	\]
	inside the numerator gives
	\begin{align*}
		&D_1U(\beta,\boldsymbol{\tilde\beta}_{-1})[d] \\
		&= \lim_{\varepsilon\to0} \int_0^1 \Bigg[ \frac{ k(\beta(c)+\varepsilon d(c),c)-k(\beta(c),c) }{ \varepsilon } W_{\tilde\beta}(\beta(c)+\varepsilon d(c)) \\
		&\hspace{4cm} + k(\beta(c),c) \frac{ W_{\tilde\beta}(\beta(c)+\varepsilon d(c)) - W_{\tilde\beta}(\beta(c)) }{ \varepsilon } \Bigg] \,dF(c).
	\end{align*}

	Since $k$ is continuously differentiable with respect to $p$,
	\[
		\frac{ k(\beta(c)+\varepsilon d(c),c)-k(\beta(c),c) }{ \varepsilon } \longrightarrow d(c)\partial_p k(\beta(c),c)
	\]
	for $F$-almost every $c$. Moreover,
	\[
		W_{\tilde\beta}(\beta(c)+\varepsilon d(c)) \longrightarrow W_{\tilde\beta}(\beta(c)).
	\]
	Hence the first term converges pointwise to
	\[
		d(c)\partial_p k(\beta(c),c) W_{\tilde\beta}(\beta(c)).
	\]

	For \( p\in(\tilde\beta(0),1) \), equation~\eqref{eq:winning-probability-piecewise} and the inverse-function rule give
	\begin{equation}
		\label{eq:winning-probability-derivative}
		W_{\tilde\beta}'(p) = - \frac{ h\!\left(\tilde\beta^{-1}(p)\right) }{ \tilde\beta'\!\left(\tilde\beta^{-1}(p)\right) }.
	\end{equation}
	Outside this interval, $W_{\tilde\beta}$ is constant. Therefore,
	\begin{equation}
		\label{eq:winning-probability-derivative-full}
		W_{\tilde\beta}'(p) = - \chi_{\{\tilde\beta(0)<p<1\}} \frac{ h\!\left(\tilde\beta^{-1}(p)\right) }{ \tilde\beta'\!\left(\tilde\beta^{-1}(p)\right) }
	\end{equation}
	for almost every $p$.

	Because $\beta$ is strictly increasing, the sets
	\[
		\{c:\beta(c)=\tilde\beta(0)\}, \qquad \{c:\beta(c)=1\}
	\]
	contain at most one point each and are therefore $F$-null. Consequently, for $F$-almost every $c$,
	\begin{align*}
		& \frac{ W_{\tilde\beta}(\beta(c)+\varepsilon d(c)) - W_{\tilde\beta}(\beta(c)) }{ \varepsilon } \\
		&\qquad\longrightarrow d(c)W_{\tilde\beta}'(\beta(c)) \\
		&\qquad= - d(c) \chi_{\{\tilde\beta(0)<\beta(c)<1\}} \frac{ h\!\left(\tilde\beta^{-1}(\beta(c))\right) }{ \tilde\beta'\!\left( \tilde\beta^{-1}(\beta(c)) \right) }.
	\end{align*}

	It remains to justify interchange of the limit and the integral. Since $\tilde\beta'\ge\delta$ almost everywhere and $h$ is bounded, $W_{\tilde\beta}$ is Lipschitz with
	\[
		\operatorname{Lip}(W_{\tilde\beta}) \le \frac{\|h\|_\infty}{\delta}.
	\]
	Hence
	\[
		\left| \frac{ W_{\tilde\beta}(\beta(c)+\varepsilon d(c)) - W_{\tilde\beta}(\beta(c)) }{ \varepsilon } \right| \le \frac{\|h\|_\infty}{\delta}|d(c)|.
	\]
	Similarly, by the mean value theorem,
	\[
		\left| \frac{ k(\beta(c)+\varepsilon d(c),c)-k(\beta(c),c) }{ \varepsilon } \right| \le \|\partial_p k\|_\infty |d(c)|
	\]
	for sufficiently small admissible $\varepsilon$.

	Since $d\in W^{1,1}([0,1])$ has an absolutely continuous representative and is bounded on $[0,1]$, these bounds provide an $F$-integrable dominating function. Dominated convergence therefore yields \eqref{eq:derivative-general}.
\end{proof}

\paragraph{Symmetric specialization.} Setting $\tilde\beta=\beta$ in \eqref{eq:derivative-general}, strict monotonicity gives
\[
	W_\beta(\beta(c))=1-H(c), \qquad \beta^{-1}(\beta(c))=c
\]
for almost every $c$. Hence
\begin{equation}
	\label{eq:derivative-symmetric}
	D_1U(\beta,\boldsymbol\beta_{-1})[d] = \int_0^1 d(c) \left[ \partial_p k(\beta(c),c)(1-H(c)) - k(\beta(c),c)\frac{h(c)}{\beta'(c)} \right] f(c)\,dc.
\end{equation}
This is the expression used in the proof of Theorem~\ref{thm:no-minty}.

\subsection{Proof of Theorem~\ref{thm:no-minty}}
\label{app:proof-no-minty}

We now prove the two statements of Theorem~\ref{thm:no-minty}. Throughout this subsection, rivals use the same strategy as the focal symmetric profile. Thus the relevant Minty expression is exactly the one in \eqref{eq:symmetric-minty}.

Recall from \eqref{eq:alpha-r} that
\[
	a=\ell(n-1), \qquad \alpha=a+s, \qquad r=\frac{\alpha}{\alpha+1},
\]
so that
\[
	\beta^*(c) = rc+(1-r), \qquad \beta^{*\prime}(c)=r.
\]
For Beta$(1,\ell)$ costs,
\[
	1-H(c)=(1-c)^a, \qquad f(c)=\ell(1-c)^{\ell-1}, \qquad h(c)=a(1-c)^{a-1}.
\]

For a symmetric strategy $\beta$, the G\^ateaux derivative in the direction $d$ specializes to
\begin{equation}
	\label{eq:app-symmetric-gateaux}
	D_1U(\beta,\boldsymbol\beta_{-1})[d] = \int_0^1 d(c) \left[ \partial_p k^s(\beta(c),c)(1-H(c)) - k^s(\beta(c),c)\frac{h(c)}{\beta'(c)} \right] f(c)\,dc .
\end{equation}

\subsubsection{Local Minty failure in the continuous-strategy game}
\label{app:continuous-minty-failure}

We first prove Theorem~\ref{thm:no-minty}(i). Because $\beta^{*\prime}=r$, the equilibrium belongs to $\mathcal B_\delta$ whenever $0<\delta\le r$. We establish local Minty failure for every $0<\delta<r$.

Fix such a $\delta$ and choose
\[
	0<\eta<r-\delta.
\]
For every integer $j\ge4$, define
\begin{equation}
	\label{eq:app-local-perturbation}
	\beta_j(c) =
	\begin{cases}
		\beta^*(c)-\eta c, & 0\le c\le \dfrac1j, \\[6pt]
		\beta^*(c)-\dfrac{2\eta}{j}+\eta c, & \dfrac1j<c\le\dfrac2j, \\[6pt]
		\beta^*(c), & \dfrac2j<c\le1.
	\end{cases}
\end{equation}
The function $\beta_j$ is continuous and piecewise linear, with slopes
\[
	r-\eta,\qquad r+\eta,\qquad r.
\]
Since $r-\eta>\delta$, we have $\beta_j'\ge\delta$ almost everywhere. Moreover, $\beta_j(1)=1$, and $\beta_j$ is increasing with $\beta_j(0)=\beta^*(0)>0$. Hence
\[
	\beta_j\in\mathcal B_\delta.
\]

Let
\[
	d_j(c):=\beta_j(c)-\beta^*(c).
\]
Then
\begin{equation}
	\label{eq:app-local-direction}
	d_j(c) =
	\begin{cases}
		-\eta c, & 0\le c\le\dfrac1j, \\[6pt]
		\eta\left(c-\dfrac2j\right), & \dfrac1j<c\le\dfrac2j, \\[6pt]
		0, & \dfrac2j<c\le1.
	\end{cases}
\end{equation}
In particular,
\[
	\|d_j\|_\infty\le\frac{\eta}{j},
\]
and the derivative differs from $\beta^{*\prime}$ only on an interval of length $2/j$. Consequently,
\[
	\beta_j\longrightarrow\beta^* \qquad \text{in }W^{1,1}([0,1]).
\]

At the equilibrium, the pointwise first-order condition gives
\begin{equation}
	\label{eq:app-equilibrium-foc}
	\partial_p k^s(\beta^*(c),c)(1-H(c)) = k^s(\beta^*(c),c)\frac{h(c)}{r}.
\end{equation}
Since $d_j(c)=\mathcal O(j^{-1})$ uniformly on $[0,2/j]$ and $k^s$ is smooth in its price argument, inserting \eqref{eq:app-local-direction} into \eqref{eq:app-symmetric-gateaux} and using \eqref{eq:app-equilibrium-foc} yields
\begin{equation}
	\label{eq:app-local-expansion}
	D_1U(\beta_j,\boldsymbol\beta_{j,-1})[d_j] = \int_0^{2/j} B(c)d_j(c) \left( \frac1r-\frac1{\beta_j'(c)} \right)dc + \mathcal O(j^{-3}),
\end{equation}
where
\[
	B(c) = k^s(\beta^*(c),c)h(c)f(c).
\]

At $c=0$,
\[
	f(0)=\ell, \qquad h(0)=\ell(n-1),
\]
and
\[
	k^s(\beta^*(0),0) = (1-\beta^*(0))^s\beta^*(0)>0.
\]
Thus
\[
	B(0)>0.
\]
Moreover,
\[
	B(c)=B(0)+\mathcal O(j^{-1})
\]
uniformly on $[0,2/j]$.

From \eqref{eq:app-local-direction},
\[
	\int_0^{1/j}d_j(c)\,dc = -\frac{\eta}{2j^2}, \qquad \int_{1/j}^{2/j}d_j(c)\,dc = -\frac{\eta}{2j^2}.
\]
On the first interval,
\[
	\frac1r-\frac1{r-\eta} = -\frac{\eta}{r(r-\eta)},
\]
whereas on the second,
\[
	\frac1r-\frac1{r+\eta} = \frac{\eta}{r(r+\eta)}.
\]
Therefore,
\begin{align}
	D_1U(\beta_j,\boldsymbol\beta_{j,-1})[d_j] &= \frac{B(0)\eta^2}{2rj^2} \left( \frac1{r-\eta} - \frac1{r+\eta} \right) + \mathcal O(j^{-3}) \nonumber\\
	&= \frac{B(0)\eta^3} {r(r^2-\eta^2)} \frac1{j^2} + \mathcal O(j^{-3}).
	\label{eq:app-local-positive}
\end{align}
The leading coefficient is strictly positive. Hence, for all sufficiently large $j$,
\[
	D_1U(\beta_j,\boldsymbol\beta_{j,-1}) [\beta_j-\beta^*] >0.
\]
Since $\beta_j\to\beta^*$, every neighborhood of $\beta^*$ in $\mathcal B_\delta$ contains an admissible Minty violation. This proves Theorem~\ref{thm:no-minty}(i).

\subsubsection{A flat-block expansion for projected strategies}
\label{app:flat-block-expansion}

To analyze the uniform piecewise-linear projection, it is useful to employ a different perturbation whose first cell becomes nearly flat. The expansion below is an auxiliary small-slope calculation. Importantly, we do \emph{not} hold the lower-slope bound $\delta$ fixed while taking the slope parameter to zero. In the projected application in Section~\ref{app:projected-minty-failure}, we set the slope parameter equal to $\delta$ and then consider sufficiently small $\delta$.

For $t\in(0,1)$, $b\in[0,1)$, and $\varepsilon>0$, define
\begin{equation}
	\label{eq:flat-block}
	\beta_{b,t,\varepsilon}(c) =
	\begin{cases}
		b+\varepsilon c, & 0\le c\le t, \\[5pt]
		b+\varepsilon t + \dfrac{1-b-\varepsilon t}{1-t}(c-t), & t<c\le1.
	\end{cases}
\end{equation}
Write
\[
	w:=1-b, \qquad \nu:=a+\ell=\ell n.
\]
For values satisfying $rt<w$, the second slope remains bounded away from zero for all sufficiently small $\varepsilon$.

Define
\begin{equation}
	\label{eq:Lambda}
	\Lambda_b(t) := \int_0^t \bigl(c-\beta^{*-1}(b)\bigr)(b-c) (1-c)^{\nu-2}\,dc .
\end{equation}
Also write
\begin{equation}
	\label{eq:Ak}
	A_k(t) := \int_{1-t}^{1}u^{\nu-2+k}\,du = \frac{1-(1-t)^{\nu+k-1}} {\nu+k-1}, \qquad k=0,1,2.
\end{equation}
A change of variables gives
\begin{equation}
	\label{eq:Lambda-closed}
	\Lambda_b(t) = \frac{(1+r)w}{r}A_1(t) - A_2(t) - \frac{w^2}{r}A_0(t).
\end{equation}

\begin{lemma}[Flat-block expansion]
	\label{lem:flat-block-expansion}
	Fix $b$ and $t$ with $rt<w$. As $\varepsilon\downarrow0$,
	\begin{equation}
		\label{eq:master}
		D_1U \bigl( \beta_{b,t,\varepsilon}, \boldsymbol\beta_{b,t,\varepsilon,-1} \bigr) \bigl[ \beta_{b,t,\varepsilon}-\beta^* \bigr] = \frac{a\ell r w^s}{\varepsilon} \Lambda_b(t) + \mathcal O(1).
	\end{equation}
	The remainder is bounded uniformly for sufficiently small $\varepsilon>0$.
\end{lemma}

\begin{proof}
	Write
	\[
		\beta:=\beta_{b,t,\varepsilon}, \qquad d:=\beta-\beta^*.
	\]
	We split \eqref{eq:app-symmetric-gateaux} at $t$.

	On $(t,1]$, the slope is
	\[
		\sigma_\varepsilon = \frac{w-\varepsilon t}{1-t}.
	\]
	Because $rt<w$, $\sigma_\varepsilon$ is bounded away from zero for $\varepsilon$ in a sufficiently small neighborhood of zero. All terms in the corresponding part of \eqref{eq:app-symmetric-gateaux} are therefore uniformly bounded, so that contribution is $\mathcal O(1)$.

	On $[0,t]$,
	\[
		\beta'(c)=\varepsilon.
	\]
	Since $\beta^*$ is affine with slope $r$,
	\[
		d(c) = r\bigl(\beta^{*-1}(b)-c\bigr) + \varepsilon c,
	\]
	and
	\[
		k^s(\beta(c),c) = \bigl[(b-c)+\varepsilon c\bigr] \bigl[w-\varepsilon c\bigr]^s.
	\]
	Moreover,
	\[
		h(c)f(c) = a\ell(1-c)^{\nu-2}.
	\]
	Hence
	\[
		P(\varepsilon) := \int_0^t d(c)k^s(\beta(c),c)h(c)f(c)\,dc
	\]
	is a polynomial in $\varepsilon$, and
	\[
		P(0) = -a\ell r w^s\Lambda_b(t).
	\]
	Consequently,
	\[
		-\frac1\varepsilon \int_0^t d(c)k^s(\beta(c),c)h(c)f(c)\,dc = \frac{a\ell r w^s}{\varepsilon} \Lambda_b(t) + \mathcal O(1).
	\]
	The remaining term
	\[
		\int_0^t d(c) \partial_p k^s(\beta(c),c) (1-H(c))f(c)\,dc
	\]
	is uniformly $\mathcal O(1)$. Combining the two intervals gives \eqref{eq:master}.
\end{proof}

\subsubsection{Minty failure after uniform projection}
\label{app:projected-minty-failure}

We now prove Theorem~\ref{thm:no-minty}(ii). Set
\[
	t=\frac1m.
\]
For fixed $m$, the expression in \eqref{eq:Lambda-closed} is a strictly concave quadratic function of $w=1-b$. Its unconstrained maximizer is
\begin{equation}
	\label{eq:wdagger}
	w^\dagger_m = \frac{(1+r)A_1(1/m)} {2A_0(1/m)}, \qquad b^\dagger_m = 1-w^\dagger_m.
\end{equation}
Since $A_1(1/m)\le A_0(1/m)$,
\[
	0<w^\dagger_m \le \frac{1+r}{2} <1,
\]
so $b^\dagger_m\in(0,1)$.

At the maximizer,
\begin{equation}
	\label{eq:Lambda-max}
	\max_b\Lambda_b(1/m) = \frac{(1+r)^2A_1(1/m)^2} {4rA_0(1/m)} - A_2(1/m).
\end{equation}
Define
\begin{equation}
	\label{eq:rho-m}
	\rho_m := \frac{ A_0(1/m)A_2(1/m) }{ A_1(1/m)^2 }.
\end{equation}
Then \eqref{eq:Lambda-max} is strictly positive if and only if
\begin{equation}
	\label{eq:grid-criterion}
	\rho_m < \frac{(1+r)^2}{4r} = 1+\frac{1}{4\alpha(\alpha+1)}.
\end{equation}

Because $A_k(1/m)$ is an average over an interval shrinking to $u=1$,
\[
	\rho_m\longrightarrow1 \qquad\text{as }m\to\infty.
\]
Since the right-hand side of \eqref{eq:grid-criterion} is strictly greater than one, there exists a finite resolution beyond which \eqref{eq:grid-criterion} holds. Define
\begin{equation}
	\label{eq:mstar-app}
	m^\star := \min \left\{ m_0\ge2: \rho_m < 1+\frac{1}{4\alpha(\alpha+1)} \text{ for every }m\ge m_0 \right\}.
\end{equation}

A sharper explicit upper bound is available. Let
\[
	q:=1-\frac1m, \qquad \gamma:=\nu-2\ge0,
\]
and, for $z\ge0$, define
\[
	I_z(q):=\int_q^1 u^z\,du.
\]
Then
\[
	A_0(1/m)=I_\gamma(q), \qquad A_1(1/m)=I_{\gamma+1}(q), \qquad A_2(1/m)=I_{\gamma+2}(q),
\]
so that
\[
	\rho_m=R_\gamma(q), \qquad R_z(q):=\frac{I_z(q)I_{z+2}(q)}{I_{z+1}(q)^2}.
\]
We first show that
\begin{equation}
	\label{eq:rho-uniform-bound-step}
	R_z(q)\le R_0(q) \qquad\text{for every }z\ge0.
\end{equation}
Write $q=e^{-L}$, with $L>0$, and define
\[
	g(\lambda):=\log\left(\frac{1-e^{-\lambda L}}{\lambda}\right), \qquad \lambda>0.
\]
Since
\[
	I_z(q)=\frac{1-q^{z+1}}{z+1}=\frac{1-e^{-(z+1)L}}{z+1},
\]
we have
\[
	\log R_z(q)=g(z+1)+g(z+3)-2g(z+2)=\int_0^1\int_0^1 g''(z+1+x+y)\,dx\,dy.
\]
It is therefore enough to show that $g''$ is decreasing. A direct calculation gives
\[
	g'''(\lambda)=L^3\left[\frac{e^\xi(e^\xi+1)}{(e^\xi-1)^3}-\frac{2}{\xi^3}\right], \qquad \xi:=\lambda L.
\]
Writing $\xi=2y$, this becomes
\[
	g'''(\lambda)=\frac{L^3}{4}\left[\frac{\cosh y}{\sinh^3 y}-\frac1{y^3}\right].
\]
The bracket is strictly negative for $y>0$. Indeed, this is equivalent to
\[
	\sinh^3 y>y^3\cosh y.
\]
For
\[
	\phi(y):=3\log\left(\frac{\sinh y}{y}\right)-\log(\cosh y),
\]
we have $\phi(0)=0$ and
\[
	\phi'(y)=\frac{N(y)}{y\sinh y\cosh y}, \qquad N(y):=y\bigl(3+2\sinh^2y\bigr)-3\sinh y\cosh y.
\]
Moreover,
\[
	N'(y)=4\sinh y\bigl(y\cosh y-\sinh y\bigr)>0 \qquad (y>0),
\]
because $y\cosh y-\sinh y$ vanishes at zero and has derivative $y\sinh y>0$. Hence $N(y)>0$, so $\phi'(y)>0$ and therefore $\phi(y)>0$ for every $y>0$. Thus $g'''<0$, and $g''$ is strictly decreasing. It follows from the integral representation above that $R_z(q)$ is decreasing in $z$, proving \eqref{eq:rho-uniform-bound-step}.

It remains to evaluate $R_0(q)$. Writing $t=1/m$, so that $q=1-t$, we obtain
\[
	I_0(q)=t, \qquad I_1(q)=t\left(1-\frac t2\right), \qquad I_2(q)=t\left(1-t+\frac{t^2}{3}\right).
\]
Consequently,
\begin{equation}
	\label{eq:rho-uniform-bound}
	\rho_m\le R_0(q)=\frac{1-t+t^2/3}{(1-t/2)^2}=1+\frac{1}{3(2m-1)^2}.
\end{equation}

First suppose $\alpha\ge2$. For every $m\ge\alpha$, \eqref{eq:rho-uniform-bound} gives
\[
	\rho_m\le1+\frac{1}{3(2m-1)^2}\le1+\frac{1}{3(2\alpha-1)^2}.
\]
Moreover,
\[
	3(2\alpha-1)^2-4\alpha(\alpha+1)=8(\alpha-1)^2-5>0,
\]
so
\[
	\frac{1}{3(2\alpha-1)^2}<\frac{1}{4\alpha(\alpha+1)}.
\]
Therefore
\[
	\rho_m<1+\frac{1}{4\alpha(\alpha+1)} \qquad\text{for every }m\ge\alpha.
\]
By \eqref{eq:mstar-app}, this implies $m^\star\le\alpha$.

If $\alpha=1$, then for every $m\ge2$,
\[
	\rho_m\le1+\frac{1}{3(2m-1)^2}\le1+\frac1{27}<1+\frac18=1+\frac{1}{4\alpha(\alpha+1)},
\]
so $m^\star=2$. Combining the two cases yields
\begin{equation}
	\label{eq:mstar-bound}
	m^\star\le\max\{2,\alpha\}=\max\{2,\ell(n-1)+s\}.
\end{equation}

Fix now any $m\ge m^\star$ and abbreviate
\[
	w^\dagger:=w^\dagger_m, \qquad b^\dagger:=b^\dagger_m.
\]
For $\delta>0$, set
\begin{equation}
	\label{eq:projected-witness}
	\mathbf x(\delta) = (\delta,\kappa,\ldots,\kappa), \qquad \kappa = \frac{mw^\dagger-\delta}{m-1}.
\end{equation}
The associated pricing function is precisely
\[
	\beta_{\mathbf x(\delta)} = \beta_{b^\dagger,1/m,\delta}.
\]
For sufficiently small $\delta$,
\[
	\kappa\ge\delta, \qquad \sum_{i=1}^m x_i = mw^\dagger \le m,
\]
so
\[
	\mathbf x(\delta)\in\mathcal B_\delta^m.
\]

Because the parametrization $\mathbf x\mapsto\beta_{\mathbf x}$ is affine,
\begin{equation}
	\label{eq:chain-rule-minty}
	\bigl( \mathbf x-\mathbf x^* \bigr)^\top v(\mathbf x) = D_1U \bigl( \beta_{\mathbf x}, \boldsymbol\beta_{\mathbf x,-1} \bigr) \bigl[ \beta_{\mathbf x}-\beta^* \bigr].
\end{equation}
Applying Lemma~\ref{lem:flat-block-expansion} with $t=1/m$ and $\varepsilon=\delta$ therefore yields
\begin{equation}
	\label{eq:projected-leading}
	\bigl( \mathbf x(\delta)-\mathbf x^* \bigr)^\top v\bigl(\mathbf x(\delta)\bigr) = \frac{a\ell r(w^\dagger)^s}{\delta} \Lambda_{b^\dagger}(1/m) + \mathcal O(1).
\end{equation}
By the definition of $m^\star$,
\[
	\Lambda_{b^\dagger}(1/m)>0.
\]
Hence there exists
\[
	\delta_0=\delta_0(n,\ell,s,m)>0
\]
such that
\begin{equation}
	\label{eq:projected-positive}
	\bigl( \mathbf x(\delta)-\mathbf x^* \bigr)^\top v\bigl(\mathbf x(\delta)\bigr) >0 \qquad \text{for every } 0<\delta<\delta_0.
\end{equation}
Thus the projected Minty condition fails for every $m\ge m^\star$ and all sufficiently small lower-slope bounds $\delta$. This proves Theorem~\ref{thm:no-minty}(ii).

\newpage
\section{Projected Dynamics and Stochastic Convergence}
\label{app:algorithmic-proofs}

\subsection{Well-posedness and proof of Proposition~\ref{prop:proximal-tracking}}
\label{app:tracking-proof}

We provide a direct argument for the compact convex polyhedron $\mathcal X$ and Hessian metric used in the paper. All norms in this proof are Euclidean; norm equivalence gives the corresponding dual-norm estimates. Fix constants $0<\mu\le\Lambda$ such that
\[
 \mu I\preceq H(x)\preceq\Lambda I\qquad(x\in\mathcal X).
\]
The Hessian is Lipschitz on $\mathcal X$, since $h$ is $C^3$ on a neighborhood of this compact set.

\paragraph{Bounds for a proximal step.}
For an arbitrary oracle vector $q$, let $z$ be the proximal update from $x$ with step $\gamma$. Strong convexity gives a unique update, and its optimality condition is
\begin{equation}
 \label{eq:app-proximal-kkt}
 \frac{\nabla h(z)-\nabla h(x)}{\gamma}=q-\xi,
 \qquad \xi\in N_{\mathcal X}(z).
\end{equation}
Let $d=z-x$ and $\overline H=\int_0^1H(x+ud)\,du$. Because $x\in\mathcal X$ and $\xi\in N_{\mathcal X}(z)$, $\langle\xi,d\rangle\ge0$. Thus
\[
 \mu\|d\|^2\le\langle\overline H d,d\rangle
 =\gamma\langle q-\xi,d\rangle\le\gamma\|q\|\|d\|,
\]
which implies
\begin{equation}
 \label{eq:app-proximal-bounds}
 \frac{\|d\|}{\gamma}\le\frac{\|q\|}{\mu},
 \qquad
 \|\xi\|\le(1+\Lambda/\mu)\|q\|.
\end{equation}
The normal cone itself may be unbounded, but the correction selected by each step is bounded whenever the oracle vector is bounded.

\paragraph{Limits of proximal interpolations.}
Consider interpolated proximal paths on a fixed time interval whose largest mesh size tends to zero. Assume that their oracle vectors equal $v$ at the left endpoint plus errors tending uniformly to zero. The bounds in \eqref{eq:app-proximal-bounds} make these paths uniformly Lipschitz. Compactness of $\mathcal X$ gives a subsequence converging uniformly to an absolutely continuous path $z(\cdot)\in\mathcal X$. The piecewise-constant selected normals are uniformly bounded, so a further subsequence converges weak-* in $L^\infty$ to some $\xi(\cdot)$.

The left- and right-endpoint step functions both converge uniformly to $z$. For every $y\in\mathcal X$ the discrete normal condition is
\[
 \langle\xi_k(t),y-X_k^+(t)\rangle\le0.
\]
Multiplying by an arbitrary nonnegative integrable test function and passing to the limit gives
\[
 \langle\xi(t),y-z(t)\rangle\le0\quad\text{for almost every }t.
\]
First use a countable dense set of $y$ and then continuity in $y$ to obtain the statement simultaneously for all $y\in\mathcal X$. Consequently,
\begin{equation}
 \label{eq:app-normal-limit}
 \xi(t)\in N_{\mathcal X}(z(t))\quad\text{for almost every }t.
\end{equation}

For an affine interpolation $Z_k$, the chain rule and \eqref{eq:app-proximal-kkt} give
\[
 \frac{d}{dt}\nabla h(Z_k(t))
 =v(X_k^-(t))+e_k(t)-\xi_k(t)+r_k(t),
\]
where $r_k(t)=[H(Z_k(t))-\overline H_k(t)]\dot Z_k(t)$ converges uniformly to zero by Lipschitz continuity of $H$, the vanishing mesh, and \eqref{eq:app-proximal-bounds}. Integrating and taking limits yields
\[
 \nabla h(z(t))-\nabla h(z(0))
 =\int_0^t\bigl[v(z(u))-\xi(u)\bigr]\,du.
\]
Therefore $H(z)\dot z=v(z)-\xi$ almost everywhere.

At almost every interior time where $z$ is differentiable, viability implies both $\dot z\in T_{\mathcal X}(z)$ and $\langle\xi,\dot z\rangle=0$. For the second statement, use the normal inequality with $z(t+h)\in\mathcal X$ for positive and negative $h$ and pass to the derivative. The normal condition and this orthogonality are exactly the optimality conditions for
\[
 \dot z=\operatorname*{arg\,min}_{u\in T_{\mathcal X}(z)}
 (H(z)^{-1}v(z)-u)^\top H(z)(H(z)^{-1}v(z)-u).
\]
Thus every limit path solves the projected dynamics.

\paragraph{Existence, uniqueness, and continuous dependence.}
For existence, start at any $x_0\in\mathcal X$ and use exact-gradient proximal steps of constant mesh $1/k$ on a fixed finite horizon. The preceding compactness argument produces a solution from $x_0$. A diagonal extraction over integer horizons gives a solution on $[0,\infty)$.

Any viable solution has speed at most $\sup_{\mathcal X}\|v\|/\mu$. Indeed, metric projection onto a cone containing zero cannot increase the metric norm, and
\[
 \|\dot x\|^2\le\mu^{-1}\|\dot x\|_{H(x)}^2
 \le\mu^{-1}v(x)^\top H(x)^{-1}v(x)
 \le\mu^{-2}\|v(x)\|^2.
\]
To prove uniqueness, let $x,y$ be two solutions with corresponding outward normals $\xi_x,\xi_y$. The Bregman divergence satisfies
\[
 \frac{\mu}{2}\|x-y\|^2\le D_h(x,y)
 \le\frac{\Lambda}{2}\|x-y\|^2.
\]
Its derivative is
\begin{align*}
 \frac{d}{dt}D_h(x,y)
 &=\langle\nabla h(x)-\nabla h(y),\dot x\rangle
   -\langle x-y,H(y)\dot y\rangle\\
 &=\langle x-y,v(x)-v(y)-\xi_x+\xi_y\rangle
   +\langle R(x,y),\dot x\rangle,
\end{align*}
where $R(x,y)=\nabla h(x)-\nabla h(y)-H(x)(x-y)$ obeys
$\|R(x,y)\|\le C\|x-y\|^2$. Monotonicity of the Euclidean normal cone gives
$\langle x-y,\xi_x-\xi_y\rangle\ge0$. Lipschitz continuity of $v$ and the speed bound then imply
\[
 \frac{d}{dt}D_h(x,y)\le C\|x-y\|^2\le C' D_h(x,y).
\]
Gronwall's inequality proves uniqueness and continuous dependence on initial data, uniformly on every bounded time interval. The autonomous projected system therefore defines a continuous semiflow $\phi_t$ on $\mathcal X$.

\paragraph{Asymptotic tracking.}
Fix a sample path on which the oracle error tends to zero and let $Z$ be its affine interpolation. For any shifts $s_k\to\infty$ and finite horizon $T$, the curves $Z(s_k+\cdot)$ have subsequences converging uniformly on $[0,T]$. The maximal step size on these shifted intervals tends to zero, and their oracle errors tend uniformly to zero. The preceding limit argument therefore identifies every limit as $\phi_{\cdot}(z_0)$, where $z_0$ is the limiting initial point. Continuous dependence gives
\[
 \sup_{0\le t\le T}\|Z(s+t)-\phi_t(Z(s))\|\longrightarrow0
 \qquad(s\to\infty).
\]
Otherwise a sequence of shifts violating this limit would have a convergent subsequence contradicting uniqueness and continuous dependence. This is the asserted asymptotic pseudo-trajectory property and completes Proposition~\ref{prop:proximal-tracking}.

\subsection{Persistent-noise convergence in Corollary~\ref{cor:Convergence-result}}
\label{app:stochastic-level-crossing}

We complete the pathwise argument for case (b). Since $a_L$ is bounded and (H2) gives a deterministic summable bound on the conditional second moments of $\gamma_nU_n$,
\[
 \sum_n\mathbb E[(\Delta M_{n+1})^2]<\infty.
\]
Hence $M_n$ converges almost surely. The same bounds and Tonelli's theorem imply
$\sum_n\gamma_n^2\|U_n\|_*^2<\infty$ almost surely. Condition (H3) and Borel--Cantelli give $b_n\to0$ almost surely, so $b_n$ is bounded on each such sample path. Since $v$ is bounded and $\sum_n\gamma_n^2<\infty$, we also have
\[
 \sum_nR_n=C_L\sum_n\gamma_n^2\|v(X_n)+U_n+b_n\|_*^2<\infty
 \quad\text{almost surely}.
\]

Fix a sample path on which all these conclusions hold. The tails
\[
 S_n=\sum_{k=n}^{\infty}(\Delta M_{k+1}+R_k)
\]
are well defined and tend to zero; the martingale part need not be absolutely summable. Set $Y_n=L(X_n)+S_n$. Since $S_n-S_{n+1}=\Delta M_{n+1}+R_n$, equation~\eqref{eq:stochastic-lyapunov-step} becomes
\begin{equation}
 \label{eq:app-corrected-level}
 Y_{n+1}\le Y_n-\gamma_n w\|X_n-x^*\|^2+\gamma_n\varepsilon_n,
 \qquad\varepsilon_n\to0.
\end{equation}
The sequence $Y_n$ is bounded below because $L\ge0$ and $S_n\to0$.

Fix $\epsilon>0$. If the set $\{x\in\mathcal X:L(x)\ge\epsilon/2\}$ is empty, the desired upper bound on $L(X_n)$ is immediate. Otherwise, positive definiteness and compactness imply
\[
 c_\epsilon:=\min\{w\|x-x^*\|^2:x\in\mathcal X,
 L(x)\ge\epsilon/2\}>0.
\]
For all sufficiently large $n$, $|S_n|\le\epsilon/2$,
$\varepsilon_n\le c_\epsilon/2$, and $\gamma_n\varepsilon_n\le\epsilon$.
Whenever $Y_n\ge\epsilon$, we then have $L(X_n)\ge\epsilon/2$ and
\[
 Y_{n+1}\le Y_n-\gamma_n c_\epsilon/2.
\]
Because $\sum_n\gamma_n=\infty$ and $Y_n$ is bounded below, $Y_n$ must eventually fall below $\epsilon$. After that first crossing, it cannot exceed $2\epsilon$: a step from below $\epsilon$ increases it by at most $\gamma_n\varepsilon_n\le\epsilon$, while a step from at least $\epsilon$ decreases it. Consequently $\limsup_nY_n\le2\epsilon$. As $\epsilon$ is arbitrary and $S_n\to0$, $L(X_n)\to0$. Finally, compactness and the unique zero of $L$ give $X_n\to x^*$.

\end{document}